\documentclass[11pt]{article}
\usepackage{dsfont}

\usepackage[letterpaper,margin=1in]{geometry}
\usepackage[dvipsnames,table,xcdraw]{xcolor}
\usepackage[utf8]{inputenc}
\usepackage{hyphenat}

\def\colorful{0}

\ifnum\colorful=1
\usepackage[normalem]{ulem}
\newcommand{\new}[1]{{\blue #1}}
\newcommand{\anote}[1]{\footnote{{\bf [Ankit: {#1}\bf ]}}}
\newcommand{\todo}[1]{{\textbf{ [{\red{Todo}}: {#1}]}}}
\newcommand{\light}[1]{ {\grey #1}} 
\newcommand{\Snote}[1]{\textcolor{red}{Stefan: #1}}
\newcommand{\stefan}[1]{\textcolor{red}{Stefan: #1}}

\newcommand{\Samnote}[1]{\textcolor{green}{Sam: #1}}
\newcommand{\ankit}[1]{\textcolor{purple}{Ankit: #1}}
\else
\newcommand{\new}[1]{{#1}}
\newcommand{\anote}[1]{}
\newcommand{\todo}[1]{}
\newcommand{\light}[1]{}
\newcommand{\Snote}[1]{}
\newcommand{\stefan}[1]{}
\newcommand{\Samnote}[1]{}
\newcommand{\ankit}[1]{}
\fi

\usepackage[T1]{fontenc}    %
\usepackage{url}      %
\usepackage{booktabs}     %
\usepackage{nicefrac}     %

\usepackage{amsthm,amsfonts,amsmath,amssymb,epsfig,color,float,graphicx,verbatim,
enumitem,mathtools}

\usepackage{bbm}
\usepackage{caption}
\usepackage{turnstile}

\usepackage[
backend=biber,
style=alphabetic,
maxbibnames=15,
maxalphanames=10,
minalphanames=6,
doi=false,
isbn=false,
url=false,
eprint=false
]{biblatex}
\definecolor{deepred}{rgb}{0.67, 0.0, 0.0}
\definecolor{lightblue}{rgb}{0, 0.34, 0.81}
\definecolor{teal}{rgb}{0.04, 0.28, 0.28}
\usepackage[colorlinks,citecolor=teal,linkcolor=deepred,bookmarks=true]{hyperref}

\usepackage[nameinlink,capitalise]{cleveref}

\usepackage{thm-restate}

\DeclareMathOperator{\E}{\mathbb E}
\DeclareMathOperator{\pE}{\widetilde{\mathbb E}}
\def\P{\mathbb P}
\def\R{\mathbb R}
\def\I{\mathbb I}

\def\N{\mathbb N}

\def\eps{\epsilon}
\newcommand{\tO}{\widetilde{O}}

\newcommand\numberthis{\addtocounter{equation}{1}\tag{\theequation}}

\newcommand{\1}{\mathbbm{1}}

\let\vec\mathbf

\newcommand{\bI}{\vec{I}}

\newcommand{\sosBound}{\hyperref[eq:sosBound]{\mathsf{R_{sos}}}}

\newcommand{\poly}{\mathrm{poly}}
\newcommand{\polylog}{\mathrm{polylog}}
\newcommand{\cN}{\mathcal{N}}

\newcommand{\cT}{\mathcal{T}}
\newcommand{\cE}{\mathcal{E}}

\newcommand{\cF}{\mathcal{F}}
\newcommand{\cP}{\mathcal{P}}

\newcommand{\cI}{\mathcal{I}}
\newcommand{\cQ}{\mathcal{Q}}
\newcommand{\cH}{\mathcal{H}}

\newcommand{\cS}{\mathcal{S}}
\newcommand{\cA}{\mathcal{A}}
\newcommand{\cC}{\mathcal{C}}
\newcommand{\cX}{\mathcal{X}}
\newcommand{\cM}{\mathcal{M}}
\newcommand{\cD}{\mathcal{D}}
\newcommand{\cB}{\mathcal{B}}

\newcommand{\nstar}{{n}^*}
\newcommand{\priv}{{\mathrm{priv}}}

\newcommand{\eff}{{\mathrm{eff}}}

\newcommand{\robust}{\textup{\texttt{robust}}}

\newcommand{\trace}{\operatorname{tr}}

\newcommand{\op}{\mathrm{op}}
\newcommand{\fr}{\mathrm{Fr}}

\crefformat{equation}{(#2#1#3)}
\crefname{subfact}{Fact}{Facts}
\crefname{sublemma}{Lemma}{Lemmata}
\crefname{equation}{Equation}{Equations}
\crefname{lemma}{Lemma}{Lemmata}
\crefname{claim}{Claim}{Claims}
\crefname{fact}{Fact}{Facts}
\crefname{theorem}{Theorem}{Theorems}
\crefname{proposition}{Proposition}{Propositions}
\crefname{corollary}{Corollary}{Corollaries}
\crefname{remark}{Remark}{Remarks}
\crefname{definition}{Definition}{Definitions}
\crefname{question}{Question}{Questions}
\crefname{condition}{Condition}{Conditions}
\crefname{figure}{Figure}{Figures}
\crefname{observation}{Observation}{Observations}
\crefname{algorithm}{Algorithm}{Algorithms}

\newtheorem{theorem}{Theorem}[section]
\newtheorem*{theorem*}{Theorem}
\newtheorem{lemma}[theorem]{Lemma}
\newtheorem*{lemma*}{Lemma}
\newtheorem{conjecture}[theorem]{Conjecture}
\newtheorem{claim}[theorem]{Claim}
\newtheorem{proposition}[theorem]{Proposition}
\newtheorem{corollary}[theorem]{Corollary}

\theoremstyle{definition}
\newtheorem{fact}[theorem]{Fact}
\newtheorem{definition}[theorem]{Definition}
\newtheorem*{definition*}{Definition}

\newtheorem{problem}[theorem]{Problem}

\newtheorem{observation}[theorem]{Observation}
\newtheorem{algorithm}[theorem]{Algorithm}

\newtheorem{remark}[theorem]{Remark}

\allowdisplaybreaks

\theoremstyle{definition}

\newcommand\MYcurrentlabel{xxx}
\newcommand{\MYstore}[2]{%
  \global\expandafter \def \csname MYMEMORY #1 \endcsname{#2}%
}
\newcommand{\MYload}[1]{%
  \csname MYMEMORY #1 \endcsname%
}
\newcommand{\MYnewlabel}[1]{%
  \renewcommand\MYcurrentlabel{#1}%
  \MYoldlabel{#1}%
}
\newcommand{\MYdummylabel}[1]{}
\newcommand{\torestate}[1]{%
  \let\MYoldlabel\label%
  \let\label\MYnewlabel%
  #1%
  \MYstore{\MYcurrentlabel}{#1}%
  \let\label\MYoldlabel%
}
\newcommand{\restatetheorem}[1]{%
  \let\MYoldlabel\label
  \let\label\MYdummylabel
  \begin{theorem*}[Restatement of \cref{#1}]
    \MYload{#1}
  \end{theorem*}
  \let\label\MYoldlabel
}
\newcommand{\restatelemma}[1]{%
  \let\MYoldlabel\label
  \let\label\MYdummylabel
  \begin{lemma*}[Restatement of \cref{#1}]
    \MYload{#1}
  \end{lemma*}
  \let\label\MYoldlabel
}
\newcommand{\restateprop}[1]{%
  \let\MYoldlabel\label
  \let\label\MYdummylabel
  \begin{proposition*}[Restatement of \cref{#1}]
    \MYload{#1}
  \end{proposition*}
  \let\label\MYoldlabel
}
\newcommand{\restatefact}[1]{%
  \let\MYoldlabel\label
  \let\label\MYdummylabel
  \begin{fact*}[Restatement of \cref{#1}]
    \MYload{#1}
  \end{fact*}
  \let\label\MYoldlabel
}
\newcommand{\restatedefinition}[1]{%
  \let\MYoldlabel\label
  \let\label\MYdummylabel
  \begin{definition*}[Restatement of \cref{#1}]
    \MYload{#1}
  \end{definition*}
  \let\label\MYoldlabel
}
\newcommand{\restate}[1]{%
  \let\MYoldlabel\label
  \let\label\MYdummylabel
  \MYload{#1}
  \let\label\MYoldlabel
}
\newcommand{\iid}{i.i.d.\ }

\usepackage{color}
\definecolor{Red}{rgb}{1,0,0}
\definecolor{Blue}{rgb}{0,0,1}
\definecolor{DGreen}{rgb}{0,0.55,0}
\definecolor{Purple}{rgb}{.75,0,.25}
\definecolor{Grey}{rgb}{.5,.5,.5}

\def\red{\color{Red}}

\def\grey{\color{Grey}}

\usepackage{turnstile}

\usepackage{tikz}
\usetikzlibrary{shapes,decorations,arrows,positioning,calc}
\tikzstyle{operator}=[circle, radius=0.2cm, text centered, draw=black]
\tikzstyle{arrow}=[thick, ->, >=stealth]
\usepackage{subcaption}
\usepackage{mdframed}
\mdfdefinestyle{sosProgram}{%
    linecolor=blue,
    outerlinewidth=2pt,
    roundcorner=20pt,
    innertopmargin=\baselineskip,
    innerbottommargin=\baselineskip,
    innerrightmargin=20pt,
    innerleftmargin=20pt,
    backgroundcolor=gray!10!white}
\usepackage{soul}

\definecolor{navyblue}{rgb}{0.2, 0.4, 0.8}
\usetikzlibrary{shapes.geometric}
\usetikzlibrary{fit}

\newcommand{\norm}[1]{\lVert#1\rVert}
\newcommand{\Norm}[1]{\left\lVert#1\right\rVert}

\newcommand{\Ind}{\mathbbm{1}}

\newcommand{\iprod}[1]{\langle#1\rangle}

\newcommand{\Paren}[1]{\left(#1\right)}

\newcommand{\Brac}[1]{\left[#1\right]}
\newcommand{\floor}[1]{\lfloor #1 \rfloor}

\newcommand{\Iprod}[1]{\left\langle #1 \right\rangle}

\newcommand{\abs}[1]{\lvert #1 \rvert}

\newcommand{\Set}[1]{\left\{#1\right\}}
\newcommand{\card}[1]{\vert#1\rvert}

\newcommand{\proves}{\vdash}

\newcommand{\cG}{\mathcal{G}}
\newcommand{\sse}{\subseteq}
\newcommand{\Gcycmerge}[1]{\cG_{#1}^{\mathsf{merged\_cycles}}}
\newcommand{\Gcycmergedisjoint}[1]{\cG_{#1}^{\mathsf{merged\_disjoint\_cycles}}}

\newcommand{\degdagger}{\deg^\dagger}

\newcommand{\e}{\varepsilon}

\title{SoS Certificates for Sparse Singular Values and Their Applications: Robust Statistics, Subspace Distortion, and More}
\author{}
\date{}

\begin{document}
\maketitle
\vspace{-0.4in}
\begin{center}
\renewcommand*{\thefootnote}{\fnsymbol{footnote}}
\begin{tabular}{ccc}
		{\large{Ilias Diakonikolas}}\footnotemark[2] & \hspace*{.1in} & {\large{Samuel B. Hopkins}}\footnotemark[3] \\
		{\large{UW Madison}} & \hspace*{.1in} & {\large{MIT}} \\
		{\large{\texttt{ilias@cs.wisc.edu}}} &  & {\large{\texttt{samhop@mit.edu}}} \vspace{.25in}\\
		{\large{Ankit Pensia}}\footnotemark[4]  & \hspace*{.1in}& {\large{Stefan Tiegel}}\footnotemark[5] \\
		{\large{Simons Institute, UC Berkeley}} & \hspace*{.1in}  & {\large{ETH Z\"urich}} \\
		{\large{\texttt{ankitp@berkeley.edu}}}  & \hspace*{.1in}& {\large{\texttt{stefan.tiegel@inf.ethz.ch}}} \\
			\end{tabular}
\footnotetext[2]{Supported by NSF Medium Award CCF-2107079 and an H.I. Romnes Faculty Fellowship.}
\footnotetext[3]{Supported by NSF CAREER award no. 2238080 and MLA@CSAIL.}
\footnotetext[4]{Most of this work was done while the author was supported by IBM Herman Goldstine Fellowship. }
\footnotetext[5]{Supported by the European Union’s Horizon research and innovation programme (grant agreement no. 815464). Most of this work was done while the author was visiting MIT.}
\vspace{0.4in}

\today
\end{center}

\thispagestyle{empty}

\begin{abstract}
    We study \emph{sparse singular value certificates} for random rectangular matrices.
    If $M$ is an $n \times d$ matrix with \new{independent} Gaussian entries, we give a new family of polynomial-time algorithms which can certify upper bounds on the maximum of $\|M u\|$, where $u$ is a unit vector with at most $\eta n$ nonzero entries 
    for a given $\eta \in (0,1)$.
    This \new{basic algorithmic primitive lies} %
    at the heart of a wide range of problems across algorithmic statistics and theoretical computer science, including robust mean and covariance estimation, certification of distortion of random subspaces of $\R^n$, certification of the $2 \rightarrow p$ norm of a random matrix, and sparse principal component analysis.
    
    Our algorithms certify a bound which is asymptotically smaller than the naive one, given by the maximum singular value of $M$, for nearly the widest-possible range of $n,d,$ and $\eta$. \new{Efficiently} certifying such a bound %
    for a range of $n,d$ and $\eta$ which is larger by {\em any} 
    polynomial factor \new{than what is achieved by our algorithm} 
    would violate lower bounds 
    in the statistical query and low-degree polynomials models. %
    Our certification algorithm \new{makes essential use of} the Sum-of-Squares hierarchy. \new{To prove the correctness of our algorithm,} we develop a new combinatorial connection between the \emph{graph matrix} approach \new{to analyze} random matrices 
    with dependent entries, and the Efron-Stein decomposition of functions 
    of independent random variables.

    As applications \new{of our certification algorithm}, 
    we obtain new efficient algorithms for a wide range of well-studied algorithmic tasks.
    In algorithmic robust statistics, we obtain new algorithms for robust mean and covariance estimation with tradeoffs between breakdown point and sample complexity, which are nearly matched by \new{statistical query and low-degree polynomial lower bounds} (that we establish).
    We also obtain new polynomial-time guarantees for certification of $\ell_1/\ell_2$ distortion of random subspaces of $\R^n$ (also with nearly matching lower bounds), sparse principal component analysis, and certification of the $2\rightarrow p$ norm of a random matrix.
\end{abstract}

\newpage

\thispagestyle{empty}
\setcounter{tocdepth}{2}
\tableofcontents
\thispagestyle{empty}

\newpage

\setcounter{page}{1}

\section{Introduction} %
\label{sec:introduction}
\label{sec:intro}
Let $M \in \R^{d \times n}$ be a rectangular matrix, with $n \geq d$.
We study the following maximization problem over sparse vectors in $\R^n$ in the case that $M$ is a random matrix.
\begin{problem}[Maximum $\eta$-Sparse Singular Value (\textsf{SSV})]
\label{prob:sparse-sing-vect}
What is the maximum of $\|M u\|$, where $u \in \R^n$ is an $\eta$-sparse unit vector \new{(i.e., supported on at most $\eta n$ coordinates)}?
\end{problem}
\noindent In \Cref{prob:sparse-sing-vect}, $\|Mu\|$ refers to the Euclidean norm of the vector $Mu$.
Without restricting $u$ to be $\eta n$-sparse, $\|Mu\|$ would be maximized with $u$ as the maximal right singular vector of $M$.
For this reason, \Cref{prob:sparse-sing-vect} is a \emph{sparse} version of the problem of finding a maximum singular value.
Unlike the maximum singular value, even approximating \textsf{SSV} is hard in the worst case \cite{chan2016approximability}.\footnote{No constant-factor approximation exists assuming the Small-Set Expansion Hypothesis (SSEH); 
under the assumption $P \neq NP$ there is no fully polynomial-time approximation scheme.}

\looseness=-1\new{Despite these worst-case hardness results}, 
a remarkable range of average-case algorithmic tasks can be efficiently 
reduced to \emph{average-case} versions of the \textsf{SSV} problem. 
These include sparse principal component analysis (sparse PCA)~\cite{zou2006sparse,amini2008high,d2008optimal, 
johnstone2009consistency,deshpande2014information,deshp2016sparse,
holtzman2020greedy,ding2024subexponential,d2020sparse}, 
computing the distortion of a random subspace of $\R^n$ 
(a task closely related to compressed sensing and error-correcting codes over the reals)~\cite{lindenstrauss1977dimension,kashin1977diameters,
guruswami2024certifying}, 
computing $2 \rightarrow p$ norms of a random matrix \cite{barak2012hypercontractivity}, and \new{outlier}-robust mean and covariance estimation~\cite{DiaKKLMS16-focs, DiaKKLMS17, LaiRV16, SteCV18, KotSte17, HopLi18, KotSS18, DiaKP20, DiaKKPP22-colt, DiaKan22-book} -- see \Cref{ssec:intro-apps,sec:applications}.
Consequently, average-case versions of \textsf{SSV} have been intensely studied \new{(either explicitly or implicitly)}
in the context of these downstream \new{applications} 
over the last two decades.

There are several natural average-case %
\new{formulations} of \Cref{prob:sparse-sing-vect} \new{in the literature}.
In this paper, we focus on \new{the task of \emph{certification}}: 
given $M \sim \cN(0,1)^{d \times n}$\footnote{Our results extend to distributions beyond Gaussian; we require only that the entries of $M$ are mean $0$ and jointly subgaussian.
} 
and $\eta \in (0,1]$, the goal is to \new{compute} an upper bound on the maximum $\eta$-sparse singular value of $M$.
Certification algorithms for \textsf{SSV} \new{translate} 
to strong algorithmic guarantees for a range of \new{related algorithmic tasks, including the ones mentioned in the preceding paragraph.}

This certification variant of \textsf{SSV} \new{exhibits} 
a rich algorithmic landscape, 
with the best known polynomial-time guarantees 
arising from an array of different \new{methods}  
depending on the relationship among $n,d,$ and $\eta$. 
Lower bounds proved in restricted but powerful computational models (e.g., Sum of Squares, the Statistical Query (SQ) model, and Low-Degree Polynomial tests) have been established that match known algorithmic guarantees  
\new{only in fairly restricted settings} (depending on $n,d$, and $\eta$).
We focus here on arguably the most important setting in which existing algorithms and lower bounds are far apart, with numerous downstream consequences.

To set the stage, we start by pointing out 
that the maximum (non-sparse) singular value $\|M\|_{\text{op}}$ is an efficiently-computable upper bound on the maximum sparse singular value.
However, if $\eta \ll 1$ and $n \gg d$, the maximum sparse singular value is much smaller than $\|M\|_{\text{op}}$, with high probability.
\new{Motivated by this observation,} we will call an \textsf{SSV}-certifying algorithm \emph{nontrivial} if it certifies a bound which is $o(\|M\|_{\text{op}})$ with high probability.
In this paper we study the question:
\begin{center}
    \begin{quote}
        \emph{For which $n,d,$ and $\eta$ is there a polynomial-time nontrivial $\eta$-\textsf{SSV} certifying algorithm for $M \sim \cN(0,1)^{d \times n}$?}
    \end{quote}
\end{center}

\looseness=-1\paragraph{Brief overview of results.}
Prior to this work, the state-of-the-art polynomial-time algorithms for $\eta$-\textsf{SSV} certification were nontrivial 
when $n \gg \eta \, d^2$ or $\eta \ll 1/\sqrt{d}$.
On the other hand, known computational lower bounds 
(against SQ algorithms and Low-Degree Polynomial tests) %
suggest that no polynomial-time algorithm 
is nontrivial unless $n \gg \eta^2 \, d^2$. \new{It turns out that this 
quadratic gap in $\eta$ between known algorithms and lower bounds is of critical 
importance, in part due to its implications for a range of fundamental algorithmic tasks.}
Our main contribution is a family of polynomial-time
\textsf{SSV} certifying algorithms which almost match the previously known 
lower bound \cite{mao2021optimal}. Specifically, our algorithms are nontrivial whenever 
$n \gg \eta^2 \, d^{2+\eps}$, for any constant $\eps>0$. 
\new{Moreover, by setting $\eps$ to be sub-constant,}  
we obtain a quasi-polynomial time nontrivial certifying algorithm, 
whenever $n \gg \eta^2 \, d^2$ \new{(Theorem~\ref{thm:main-intro})}.
Our algorithm makes essential use of the sum-of-squares 
hierarchy (SoS), the \emph{graph matrix} approach to analyzing random matrices with dependent entries \cite{AhnMP16}, and the Efron-Stein decomposition of functions of independent random variables \cite{efron1981jackknife}. %

Our new \textsf{SSV} certification algorithms have a number of 
downstream consequences. Specifically, we give new polynomial-time guarantees for: 
robust covariance estimation, robust mean estimation, 
certifying distortion of random subspaces, planted sparse vector, 
and certifying upper bounds on $2\rightarrow p$ norms of random matrices, 
and sparse principal component analysis.
We discuss the applications to algorithmic robust statistics and subspace distortion in this introduction, and defer extensive discussion of the others until \Cref{sec:applications}.

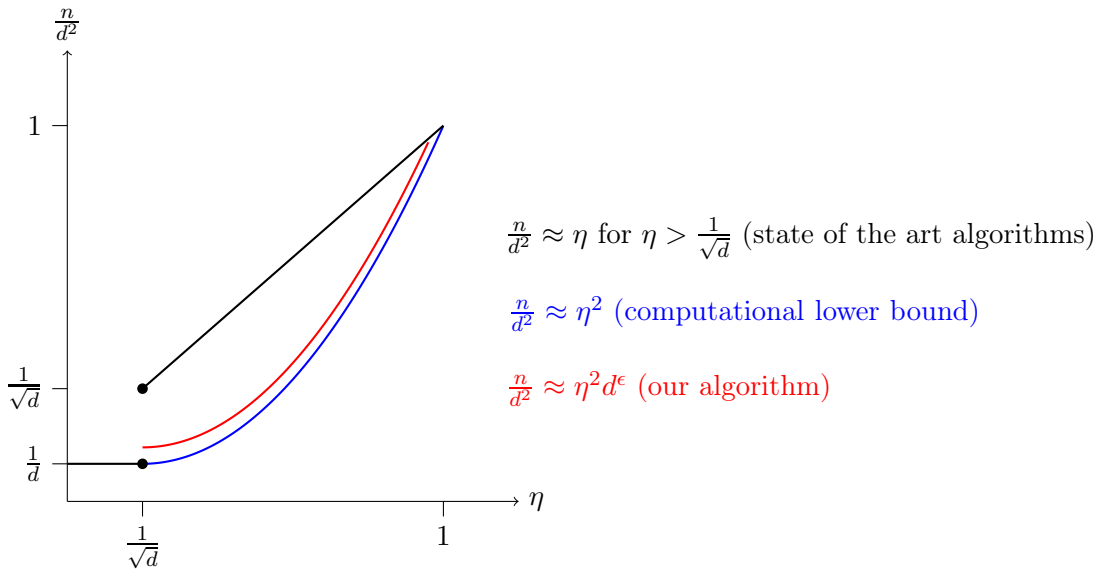
\begin{figure}[ht]
\centering
\begin{tikzpicture}

\draw[->] (0,0) -- (6,0) node[right] {$\eta$};
\draw[->] (0,0) -- (0,6) node[above] {$\frac{n}{d^2}$};

\draw (1,0) -- (1,-0.2) node[below] {$\frac{1}{\sqrt{d}}$};
\draw (5,0) -- (5,-0.2) node[below] {$1$};
\draw (0,0.5) -- (-0.2,0.5) node[left] {$\frac{1}{d}$};
\draw (0,1.5) -- (-0.2,1.5) node[left] {$\frac{1}{\sqrt{d}}$};
\draw (0,5) -- (-0.2,5) node[left] {$1$};

\draw[thick,blue,domain=1:5,samples=100] plot ({\x},{0.28125*(\x)^2 - 0.5625*(\x) + 0.78125});

\draw[thick,red,domain=1:4.8,samples=100] plot ({\x},{0.28125*(\x)^2 - 0.5625*(\x) + 1});

\node[black] at (9.75,3.5) {$\tfrac{n}{d^2} \approx \eta$ for $\eta > \tfrac{1}{\sqrt{d}}$ (state of the art algorithms)};
\node[blue] at (9,2.5) {$\tfrac{n}{d^2} \approx \eta^2$ (computational lower bound)};
\node[red] at (8,1.5) {$\tfrac{n}{d^2} \approx \eta^2 d^{\eps}$ (our algorithm)};

\draw[thick] (0,0.5) -- (1,0.5);
\draw[thick] (1,1.5) -- (5,5);

\fill[black] (1,0.5) circle (2pt);
\fill[black] (1,1.5) circle (2pt);

\end{tikzpicture}
\caption{Minimum value of $\tfrac n {d^2}$ as a function of $\eta$ for polynomial-time $\eta$-SSV refutation algorithms. Note the discontinuity in the state of the art guarantees preceding our work.}
\label{fig:eta_relationship}
\end{figure}

\subsection{Main Algorithmic Result: Sparse Singular Value Certificates}
\label{ssec:intro-main-result}
In what follows, $M$ is a $d \times n$ matrix with independent standard Gaussian entries.
We always consider $n \geq d$ and $n \leq d^{O(1)}$.
With high probability, the maximum singular value of $M$ satisfies $\|M\|_{\text{op}} = \Theta(\sqrt n)$, so we adopt the normalization $\tfrac 1 {\sqrt n} \cdot M$; then nontriviality corresponds to certifying the bound $o(1)$.

We use $\tO(\cdot)$ to hide logarithmic factors in $n$ and $d$, and $f(n) \ll g(n)$ indicates that there is a constant $C$ such that $f(n) \leq g(n) / (\log n)^C$.
``With high probability'' indicates probability $1-o(1)$.

\begin{definition}[$\eta$-Sparse Singular Value Certification]
    An algorithm \emph{certifies} the bound $\textsf{ALG} \geq 0$ on the maximum sparse singular value if, given any $d \times n$ matrix $M$, the algorithm outputs an upper bound on the maximum $\eta$-sparse singular value, and with high probability over $M \sim \cN(0,1)^{d \times n}$ that bound is at most $\textsf{ALG}$.
\end{definition}

\paragraph{Summary of prior work.} 
We start by reviewing the state of the art.
First of all, it is not hard to show by standard concentration of measure that the \emph{true} value of the $\eta$-\textsf{SSV} of $M/\sqrt{n}$ is $o(1)$ with high probability whenever $\eta \ll 1$ and $n \gg d$.

We are aware of two families of nontrivial polynomial-time certification algorithms for sparse singular values -- recall that we call such an algorithm ``nontrivial'' if it certifies a bound which is asymptotically smaller than the maximum singular value $\|M\|_{\text{op}}$.
The first uses the eigenvalues of the fourth moment matrix of the columns of the matrix $M$---that is, if the columns are $M_1,\ldots,M_n$, a certificate can be constructed from the eigenvalues of the matrix $\sum_{i \leq n} (M_i^{\otimes 2})(M_i^{\otimes 2})^\top$.
(Often, this is phrased as a ``degree $4$ SoS proof'' -- see \Cref{sec:prelims}.)
\begin{fact}[Fourth moment certificate, \cite{barak2012hypercontractivity}]
\label{thm:t-th-moment-cert}
  There is a polynomial-time algorithm which certifies that the maximum sparse singular value of $\tfrac 1 {\sqrt n} M$ is $o(1)$ if $n \gg \eta \cdot d^2$.
\end{fact}
The second type of certificate we are aware of (which can also be phrased as an SoS proof) uses pairwise inner products of the rows $\iprod{M_i, M_j}$, and works only when $\eta$ is very small.
\begin{fact}[Certificate from pairwise inner products, folklore]
\label{thm:cert-inner-products}
  There is a polynomial-time algorithm which certifies that the maximum sparse singular value of $\tfrac 1 {\sqrt n} M$ is $o(1)$ if $\eta \ll 1/\sqrt{d}$.
\end{fact}

\noindent For $\eta  \geq 1/\sqrt{d}$, \Cref{thm:t-th-moment-cert} and \Cref{thm:cert-inner-products} leave a wide range of $n$ for which polynomial-time nontrivial certification remains open---if e.g., we fix $\eta = 1/\sqrt{d}$, the range from $n \approx d$ to $n \approx d^{1.5}$ is open.

\medskip 

\paragraph{Main algorithmic contribution.} 
Our main result is a \new{new} polynomial-time certification algorithm which closes this gap.
Up to arbitrarily-small polynomial factors in $d$, our certification algorithm matches or improves upon the guarantees of both \Cref{thm:t-th-moment-cert} and \Cref{thm:cert-inner-products}.
We construct a new family of random matrices out of the columns $M_1,\ldots,M_n$ and use their eigenvalues as certificates.
Our certificates can be phrased as Sum of Squares proofs.
Our main result is:
\begin{theorem}[Main result; see \Cref{thm:sparse_sing_val_full} for the full version]
\label{thm:main-intro}
  For every $\eps > 0$, there is an $(nd)^{O(1/\eps)}$-time algorithm which certifies that the maximum $\eta$-sparse singular value of $\tfrac 1 {\sqrt n} M$ is $o(1)$ 
  if $n \gg \eta^2 \, d^{2+\eps}$ and $\eta \leq o(1)$.
  In particular, choosing $\eps = 1/\log(d)$, there is an algorithm running in time $(nd)^{O(\log d)}$ which performs this task if $n \gg \eta^2 \, d^2$.
\end{theorem}

\paragraph{Almost matching lower bounds.} 
Our main theorem almost matches lower bounds proved in restricted computational models.
In particular,
\cite{mao2021optimal,chen2022well} give evidence in the form of lower bounds against low-degree polynomial tests that when $n \ll \eta^2 d^2$, polynomial-time algorithms cannot distinguish between $M \sim \cN(0,1)^{d \times n}$ and $M$ drawn from a natural distribution on $d \times n$ matrices for which the $\eta$-\textsf{SSV} is $\Omega(\|M\|_{\text{op}})$.
We present a closely related construction in \Cref{sec:low-degree-lower-bounds}.

\paragraph{Extensions of \Cref{thm:main-intro}.}
It is natural to ask how small the $o(1)$ can be made, and how close this is to the true $\eta$-\textsf{SSV}, which is roughly $\sqrt{\eta} + \sqrt{d/n}$.
Our certifying algorithm actually certifies a bound of roughly $\eta^{1/4} + (\eta^2 d^2 / n)^{1/8}$; this more quantitative result is important for several of our applications.
Note that even for large $n$, this bound is never smaller than $\eta^{1/4}$, while the true $\eta$-\textsf{SSV} scales as $\sqrt{\eta}$.
To what extent the bound $\eta^{1/4} + (\eta^2 d^2 / n)^{1/8}$ is likely to be the smallest certifiable by polynomial time algorithms is a subtle issue, see \Cref{rem:phase-transition}.

\Cref{thm:main-intro} also extends to matrices $M$ drawn from more general probability distributions than $\cN(0,1)^{d \times n}$.
Using the techniques of \cite{diakonikolas2024sos}, \Cref{thm:main-intro} extends unchanged to $M$ whose entries are drawn from an $nd$-dimensional \new{centered} subgaussian distribution.

\subsection{Algorithmic Applications of Main Result}
\label{ssec:intro-apps}

We discuss two of the main applications of \Cref{thm:main-intro}: to several problems in algorithmic robust statistics and to certification of the distortion of random subspaces; we define these problems below.
We defer discussion of two other important applications, sparse principal component analysis and certification of the $2 \rightarrow p$ norm of a random matrix, to \Cref{sec:applications}, to avoid an avalanche of definitions.

\subsubsection{Algorithmic Robust Statistics}
Algorithmic robust statistics aims to perform statistical estimation and inference in the face of gross outliers via computationally efficient algorithms.
In this paper, we study robust statistics in the most general model of outliers, the \emph{strong contamination model}, defined below: 
\begin{definition}[Strong Contamination Model~\cite{DiaKKLMS16-focs}]
Given a \emph{contamination parameter} $\eta \in (0,1/2)$ 
and a distribution $\cD$ on uncorrupted samples, 
an algorithm takes samples from $\cD$ with \emph{$\eta$-contamination} 
as follows: 
(i) The algorithm specifies the number $n$ of samples it requires, and 
$n$ i.i.d.\ samples from $\cD$ are drawn but not yet shown to the algorithm. 
These samples are called inliers. 
(ii) An omniscient adversary inspects 
the entirety of the $n$ i.i.d.\ samples, 
before deciding to replace any subset of $\lceil \eta n \rceil$ 
samples with arbitrarily corrupted points, 
and returning the modified set of $n$ points to the algorithm.    
We call the resulting set of samples ``$\eta$-corrupted''.
\end{definition}

\paragraph{Discussion on contamination rate.} 
Over the past decade, the literature on algorithmic robust statistics 
primarily focused on understanding the computational complexity of robust estimation, 
when the fraction of outliers $\eta$ is a universal constant (independent of the dimension 
$d$).
A more ambitious goal is to perform a {\em fine-grained} 
algorithmic analysis to understand the complexity of such tasks 
{\em as a function of $\eta$}---and in particular when $\eta$ is 
``mildly sub-constant''. Roughly speaking, while handing 
the case when $\eta \ll 1/\sqrt{d}$ is typically 
 straightforward (e.g., achievable by standard estimators), 
obtaining optimal algorithms for $\eta \approx 1/d^c$, 
where $c>0$ is an arbitrarily small universal constant, turns out to be 
highly non-trivial, especially if the goal is to use fewer samples than would be required by efficient algorithms if $\eta$ were a universal constant.
The latter goal is important, since efficient algorithms 
may require prohibitively many samples to perform some fundamental robust estimation tasks in high dimensions when $\eta$ is a universal constant; 
but, it may still be possible to improve 
on the robustness guarantees of standard estimators, which break down when $\eta \gg 1/\sqrt{d}$.

An additional concrete motivation to study such mildly 
subconstant contamination rates is a recently-discovered 
technical connection to the study of differentially 
private (DP) estimators. Indeed, inspired by developments 
in differential privacy~\cite{narayanan2022private}, 
recent work in algorithmic robust statistics~\cite{CanonneH0LN23} 
has provided surprising tradeoffs between contamination rate and number of samples required to perform \emph{robust mean testing}. In the 
context of differentially private (DP) estimation, a 
recent line of 
work~\cite{DwoLei09,GeoHop22,HopKMN23,AsiUZ23} 
has established a number of reductions from robust to 
differentially private 
estimators. A high-level message coming out of these 
results is that to obtain computationally efficient
\emph{approximate} DP estimators with optimal 
tradeoffs between privacy and sample complexity, 
one needs to understand the robust setting 
when the contamination parameter is 
$\eta = 1/d^c$, for $c>0$.

\smallskip

The two prototypical problems in this field are 
robust mean and robust covariance estimation 
of a high-dimensional distribution $\cD$. 
Despite significant algorithmic progress on both of these tasks,
some basic settings of interest have remained tantalizingly open. 
Here we study robust mean and covariance estimation 
in the setting that $\cD$ is a subgaussian distribution. 
Our main technical contribution is a set of new, near-optimal 
\new{polynomial-time tradeoffs between sample complexity and contamination rate} for the case that $\cD$ is an unknown Gaussian.
These algorithms transfer to the general subgaussian case, 
via the techniques of~\cite{diakonikolas2024sos}.

\paragraph{Robust covariance estimation.}
We start with the task of robust covariance estimation for an unknown 
Gaussian $\cN(0, \Sigma)$ %
in {\em relative spectral norm}: 
given a desired accuracy parameter $\alpha>0$, 
the goal is to compute an estimate $\hat{\Sigma}$ of the target covariance $\Sigma$
such that $(1-\alpha) \hat{\Sigma} \preceq \Sigma \preceq (1+\alpha) \hat{\Sigma}$. 
In the $\eta$-corrupted setting, this task is information-theoretically possible 
only if $\alpha \geq C \eta$, for a sufficiently large universal constant $C>1$. 
To avoid clutter in the relevant expressions, we focus \new{in this introduction} on the setting that 
$\alpha = 0.1$ and $\eta \leq \eta_0$, for a sufficiently small universal constant $\eta_0$. 
For this parameter setting, the information-theoretically optimal sample complexity of our 
task is $n=\Theta(d)$ (both upper and lower bound); see, e.g., \cite{10.1214/17-AOS1607}.
Interestingly, information-computation tradeoffs (for SQ algorithms and low-degree tests)
strongly suggest a gap between the information-theoretic and computational sample complexity 
of the problem. For the case where $\eta = \Omega(1)$,~\cite{DiaKS17} gave 
an SQ lower bound suggesting that $n \gg d^2$ samples are necessary for polynomial-time 
algorithms. A low-degree testing lower bound of \cite{mao2021optimal}, reinterpreted in this context, suggests that $n \gg d^2 \eta^2 + d$ are needed by efficient algorithms.
(In \Cref{prop:low-degree-cov-estimation-gaussian}, we generalize their result to allow $\alpha = o(1)$.)  
Previous polynomial-time algorithms require $n= \Omega(\eta\, d^2)$ samples~\cite{KotSS18}.
Our main algorithmic result for this task 
is the following (see~\cref{thm:cov_est_full} for a more detailed statement). 

\begin{theorem}[Robust Covariance Estimation] \label{thm:cov-intro}
Let $\eta \in(0,\eta_0)$ for a sufficiently small constant $\eta_0>0$. 
There exists an algorithm with the following guarantee: 
Given access to $\eta$-corrupted samples from $\cN(0, \Sigma)$, and a parameter $\eps>0$, 
the algorithm draws $n =  \tilde{\Omega} (\eta^2 d^{2+2\eps}+d^{1+\eps})$ corrupted samples, 
runs in $n^{O(1/\eps)}$ time and outputs an estimate $\hat{\Sigma}$ that with high 
probability satisfies $0.9  \Sigma \preceq \hat{\Sigma} \preceq 1.1 \Sigma$.  
\end{theorem}

\noindent \new{\Cref{thm:cov-intro} achieves a near-optimal tradeoff between sample 
complexity and corruption rate $\eta$ within the class of efficient algorithms, as evidenced by the low-degree testing lower bounds discussed above.
Conceptually, 
our algorithmic result 
allows for smooth interpolation between the non-robust and robust settings, 
essentially recovering the sample complexity of outlier-free covariance estimation when $\eta \ll 1/\sqrt{d}$ (as opposed to existing algorithms whose sample complexity is $\eta d^2 + d$).}

\paragraph{Robust covariance-aware mean estimation.}
The next task we study is that of robust {\em covariance-aware} mean estimation---which served as one of the main motivations for this work. 
Typically, robust mean estimation algorithms measure their accuracy with respect to 
the Euclidean norm. Here the goal is to instead compute an estimate 
$\widehat{\mu}$ such that $\|\Sigma^{-1/2} (\widehat{\mu} - \mu)\|_2$ is small, 
where $\Sigma$ is the (unknown) target covariance.\footnote{If $\Sigma$ is not invertible, 
then we use \new{the pseudo-inverse} $(\Sigma^{\dagger})^{1/2}$ instead.} As above, we will focus on the 
simplified regime where the desired accuracy $\alpha$ is a small constant (e.g., 
$\alpha = 1/10$) and the contamination parameter $\eta$ at most a sufficiently 
small universal constant. In this regime, the sample complexity of robust 
covariance aware mean estimation is $O(d)$.\footnote{This follows from the sample bound for robust covariance estimation, 
as the covariance-aware mean task can be solved via robust covariance estimation followed by Euclidean mean estimation.}
On the other hand, known efficient algorithms with non-trivial error guarantees 
require 
$n = \Omega(\eta d^2)$ samples~\cite{KotSS18}. 

Prior to this work,  our understanding of this question was fairly limited. For example,  
it was in principle possible that a computationally efficient 
algorithm with $O(d)$ sample complexity exists. 
Here we provide {\em the first evidence} of information-computation tradeoffs 
for this task (for SQ algorithms and low-degree tests). 
Specifically, we show that any efficient algorithm from these families
requires $n = \tilde{\Omega}(\eta^2 d^2)$ samples 
(see Theorem~\ref{thm:low-degree-hardness-cov-aware-mean} 
for the corresponding testing task). 
This lower bound is new even for $\eta = \Omega(1)$, ruling 
out efficient algorithms with sample complexity 
sub-quadratic in $d$ when $\eta = \Omega(1)$. 

Our main positive result is an efficient algorithm 
whose sample complexity nearly matches our low-degree lower bound, 
as a function of $\eta$ and $d$.
Specifically, by combining our robust covariance estimator from Theorem~\ref{thm:cov-intro} 
with known tools, we establish the following 
(see Theorem~\ref{thm:cov_mean_est_full} for a more detailed statement): 
\begin{theorem}[Robust Covariance-Aware Mean Estimation] \label{thm:mean-aware-intro}
Let $\eta \in(0,\eta_0)$ for a sufficiently small constant $\eta_0>0$. 
There exists an algorithm with the following guarantee: 
Given access to $\eta$-corrupted samples from $\cN(0, \Sigma)$, and a parameter 
$\eps>0$, the algorithm draws $n =  \tilde{\Omega} (\eta^2 d^{2+2\eps}+d^{1+\eps})$ corrupted samples, runs in $n^{O(1/\eps)}$ time and outputs an estimate
$\hat{\mu}$ such that 
$\|\Sigma^{-1/2}(\hat{\mu} - \mu)\|_2 \leq O(\new{\sqrt{\eta}})$.
\end{theorem}

\new{\Cref{thm:cov_mean_est_full} achieves a non-trivial (dimension-independent) 
error guarantee with near-optimal computational sample complexity, 
up to the $d^{2\eps}$ factor. In addition to being interesting on its own merits, 
we believe that it may lead to efficient approximate-DP estimators, 
improving on the sample complexity of state-of-the-art 
algorithms~\cite{BroHS23, KDH23} (see also \Cref{rem:dp}).}

\paragraph{Robust mean estimation in Euclidean norm.}
Interestingly, we are also able to get fine-grained algorithmic improvements even for robust mean estimation in Euclidean norm. Specifically, we show (see Theorem~\ref{thm:cov_mean_est_full}): 

\begin{theorem}[Robust Mean Estimation] \label{thm:mean-intro}
Let $\eta \in(0,\eta_0)$ for a sufficiently small constant $\eta_0>0$. 
There exists an algorithm with the following guarantee: 
Given access to $\eta$-corrupted samples from $\cN(0, \Sigma)$, and 
parameters $\eps, \delta>0$, the algorithm draws 
$n =  \tilde{\Omega} (\eta^{2-\delta} d^{2+2\eps}+d^{1+\eps})$
corrupted samples, runs in $n^{O(1/\eps)}$ time and outputs an estimate
$\hat{\mu}$ such that 
$\|\hat{\mu} - \mu\|_2 \leq O(\eta^{1/2 + \Omega(\delta)}  \sqrt{\|\Sigma\|_\op})$. 
\end{theorem}

\new{The above result gives an efficient robust mean estimation for Gaussians 
with unknown identity-bounded covariance that achieves $\ell_2$-error 
$o(\sqrt{\eta})$ using sub-quadratic sample complexity, assuming $\eta$ is mildly 
sub-constant. Prior work~\cite{DiaKS19} gave an SQ lower bound suggesting that 
any efficient algorithm achieving error $o(\sqrt{\eta})$ requires 
quadratic sample complexity, assuming that $\eta = \Omega(1)$.}

\new{
\begin{remark}[Connections to Private Mean Estimation] \label{rem:dp}
\looseness=-1
Recent work gave efficient black-box reductions between robust estimation and differentially private estimation for various statistical tasks. 
Our low-degree lower bound for covariance aware mean estimation (Theorem~\ref{thm:low-degree-hardness-cov-aware-mean}) implies the first information-computation tradeoff for the approximate DP 
version of this task (see Proposition~\ref{prop:private-lower-bound}). 
In particular, this suggests 
that the recent efficient approximate-DP estimators~\cite{BroHS23, KDH23} 
cannot be improved to yield the information-theoretically 
optimal sample size in polynomial time.  
\end{remark}
}

\subsubsection{Subspace Distortion and Planted Sparse Vectors}
We turn to two applications involving random $d$-dimensional subspaces of $\R^n$.
In both the \emph{planted sparse vector problem} and the \emph{distortion certification} problem, the goal is to certify that a random subspace does not contain any sparse or approximately-sparse vectors.
These problems get harder as the number of nonzeros and subspace dimension increase.
Since the subspace distortion problem is strictly harder, we state our result in that context.

The subspace distortion problem is defined as follows: given a random $d$-dimensional subspace $V$ of $\R^n$, produce a certificate that it does not contain any vectors which are \emph{analytically sparse} in the sense of having very small $\ell_1$ norm compared to their $\ell_2$ norm.
Any sparse vector is analytically sparse, but the converse does not hold.
Formally, we define \emph{distortion} of a subspace as follows.

\begin{definition}[Distortion of a Subspace]
\label{def:distortion}
  The distortion $\Delta(V)$ of a subspace $V \subseteq \R^n$ is the maximum of the ratio $\sqrt{n} \|x\|_2 / \|x\|_1$ over all nonzero $x$ in $V$.
\end{definition}
Roughly speaking, a subspace has distortion less than $\Delta$ if it contains no $1/\Delta^2$-sparse vectors.
Random subspaces of dimension as large as $\Omega(n)$ have distortion $O(1)$ with high probability, since they contain no vectors which are (close to) $o(n)$-sparse, while $\R^n$ itself has distortion $\sqrt{n}$.
Subspaces with small distortion have wide-ranging applications: compressed sensing, error-correcting codes over the reals, and approximate nearest-neighbor search \cite{candes2005decoding,guruswami2008euclidean,kashin2007remark,donoho2006compressed,indyk2006stable}.

As the dimension of the subspace decreases, certifying bounds on the distortion becomes computationally easier.
The state of the art is captured by the following fact.

\begin{fact}[\cite{pmlr-v23-spielman12,barak2012hypercontractivity,barak2014rounding}]
    There is a polynomial time which takes a random $d$-dimensional subspace of $\R^n$ and with high probability outputs a certificate that the distortion is at most $O(\min(d/\sqrt{n}, d^{1/4}))$.
\end{fact}

The two terms in the $\min$ correspond to \Cref{thm:t-th-moment-cert} and \Cref{thm:cert-inner-products}, respectively.
We show the following result:

\begin{theorem}[This work; see~\cref{thm:distortion_full} for the full version]
\label{thm:distortion-us}
  For every $\eps > 0$, there is a polynomial-time algorithm which takes a random $d$-dimensional subspace of $\R^n$ and with high probability outputs a certificate that the distortion is at most $\tilde{O}(d^{1/2 + \eps} / n^{1/4})$.
  Furthermore, there is a quasipolynomial time algorithm which outputs with high probability a certificate that the distorition is at most $\tilde{O}(d^{1/2} / n^{1/4})$.
\end{theorem}

\paragraph{Consequences for planted sparse vector.}
Since any $\eta$-sparse unit vector $v$ has $\|v\|_1 \leq \sqrt{n \cdot \eta}$, a subspace which contains an $\eta$-sparse vector has distortion at least $1/\sqrt{\eta}$.
Thus, for every $\eps > 0$ we in particular obtain in polynomial time a certificate that a $d$-dimensional random subspace contains no $\eta$-sparse vector as long as $d \ll (\sqrt{n} / \eta)^{1-\eps}$.
This \emph{planted sparse vector} problem was extensively studied in
\cite{pmlr-v23-spielman12,demanet2014scaling,barak2014rounding,hopkins2016fast,qu2014finding,mao2021optimal,guruswami2024certifying}.
Our algorithm improves on the state of the art from \cite{pmlr-v23-spielman12,barak2014rounding}.
Interestingly, \cite{mao2021optimal} already obtained a similar improvement for a \emph{recovery} variant of the problem, where the goal is to find an $\eta$-sparse vector hidden in an otherwise random subspace; their algorithm works works when $d \ll \sqrt{n} / \eta$.
However, their algorithm does not produce a certificate that a random subspace contains no sparse vector.

\paragraph{Discussion.}
The recent work \cite{guruswami2024certifying} gives a hierarchy of \emph{subexponential-time} algorithms which trade off running time and dimension while certifying constant distortion.
\Cref{thm:distortion-us} can be viewed as an alternative tradeoff, between distortion and dimension, while maintaining polynomial running time.
It is an interesting open problem to identify a full three-way tradeoff among dimension, distortion, and running time.

Per the low-degree lower bound of \cite{mao2021optimal}, the tradeoff we obtain among $d,n,$ and $\eta$ is conjecturally optimal for polynomial-time algorithms for planted sparse vector, up to our usual $d^{\eps}$ factor.\footnote{An important exception is if the sparse vector is promised to have entries in $\{ \pm 1, 0\}$, in which case a \new{careful} application of the LLL lattice basis-reduction algorithm can efficiently detect and recover the sparse vector in a much larger-dimensional subspace \cite{zadik2022lattice,diakonikolas2022non}.}
We refer to~\Cref{sec:distortion} for a more in-depth discussion.

\paragraph{Organization.}
In \Cref{sec:techniques} we give an overview of the proof of \Cref{thm:main-intro}.
\Cref{sec:prelims} contains preliminaries.
In \Cref{sec:sos_certificate_schatten_p} we prove \Cref{thm:main-intro} in the case that $M$ is Gaussian.
\Cref{sec:transfer-lemma} we extend the proof to the subgaussian case.
\Cref{sec:applications} contains all the downstream applications of \Cref{thm:main-intro}, and \Cref{sec:low-degree-lower-bounds} contains lower bounds.

\section{Techniques}
\label{sec:techniques}
We give an overview of the proof of Theorem~\ref{thm:main-intro}.
We ultimately prove Theorem~\ref{thm:main-intro} by constructing an SoS proof.
But, to avoid technical definitions related to SoS in this overview, we will present an outline of an algorithm which can certify the upper bound $O(\eta^{1/4} + \eta^{1/4} \cdot (d^2 / n)^{1/8} \cdot d^{\eps})$ on the $\eta$-\textsf{SSV} of a $d \times n$ Gaussian matrix by computing the eigenvalues of a family of associated random matrices.
Translating such an algorithm into a sum of squares proof is by this point standard \cite{barak2012hypercontractivity,ge2015decomposing,hopkins2015tensor,barak2016noisy,raghavendra2017strongly}.
Once we have a sum of squares proof for the Gaussian setting, to obtain the guarantees of \Cref{thm:main-intro} for every subgaussian distribution, we can use the Gaussian-to-subgaussian transfer technique of \cite{diakonikolas2024sos}.
At the end of this section we give a worked-out example of our techniques to which the reader can jump at any point.

So, let $X_1,\ldots,X_n \sim \cN(0,I_d)$ be the columns of $M$.
We slightly rephrase the $\eta$-\textsf{SSV} in terms of the operator norm of the matrix induced by a subset of the columns $M$.
By Cauchy-Schwarz,

\begin{align}
\label{eq:techniques-thing-to-bound}
\max_{u : \eta\text{-sparse unit}} \|Mu\|^2
= \max_{S \subseteq [n], |S| = \eta n} \Norm{ \frac 1 n \sum_{i \in S} X_i X_i^\top}_{\text{op}} 
= \max_{S \subseteq [n], |S| = \eta n} \max_{\|v\|=1} \frac 1 n \sum_{i \in S} \iprod{X_i,v}^2  \, .
\end{align}
Our goal is to find a family of matrices whose entries are functions of $X_1,\ldots,X_n$ and whose maximum eigenvalues can be combined to upper-bound the latter.
Our goal is to certify the bound of  $O(\sqrt{\eta} + \sqrt{\eta} \cdot (d^2/n)^{1/4} \cdot d^\eps)$ on \Cref{eq:techniques-thing-to-bound}  for arbitrarily small $\eps > 0$.

\paragraph{Prior state-of-the-art: fourth moments.}
The prior state-of-the-art certificate is captured by the following Cauchy-Schwarz inequality:
\[
\frac 1 n \sum_{i \in S} \iprod{X_i,v}^2 \leq \sqrt{\eta } \cdot \sqrt{ \frac 1 n \sum_{i \leq n} \iprod{X_i,v}^4} \, .
\]
By analyzing the eigenvalues of the random matrix $\sum_{i \leq n} (X_i \otimes X_i)(X_i \otimes X_i)^\top$ (after projecting off of a particular one-dimensional subspace), one can with high probability certify the bound $\max_{\|v\|=1} \frac 1 n \sum_{i \leq n} \iprod{X_i,v}^4 \leq O(1 + d^2/n)$, leading to \Cref{thm:t-th-moment-cert} \cite{barak2012hypercontractivity}.

\paragraph{Our approach: from operator norm to Schatten $p$ norm.}
We start with a different inequality to bound \Cref{eq:techniques-thing-to-bound}.
Fix an even integer $p \in \N$---ultimately, we will take $p = O(1/\eps)$, and our algorithm will run in time exponential in $p$.
Recall that the Schatten $p$ norm of a symmetric matrix $M$ is given by $\|M\|_p = (\trace M^p)^{1/p} \geq \|M\|_{\text{op}}$.
So,
\begin{align}
\label{eq:techniques-giant-inner-product}
\Cref{eq:techniques-thing-to-bound} & \leq \max_{|S| = \eta n} \Norm{ \sum_{i \in S} X_i X_i^\top }_p \notag \\ 
&= \max_{|S| = \eta n} \Paren{ \frac 1 {n^p} \sum_{i_1,\ldots,i_p \in S} \iprod{X_{i_1},X_{i_2}} \iprod{X_{i_2},X_{i_3}} \ldots \iprod{X_{i_{p-1}},X_{i_p}} \iprod{X_{i_p},X_{i_1}}}^{1/p} \, .
\end{align}
At this point, it is notationally convenient (and in-line with our \new{eventual} SoS proof approach) to view \Cref{eq:techniques-giant-inner-product} as maximizing a degree-$p$ polynomial in variables $w_1,\ldots,w_n$, representing the $0/1$ indicators $\mathbf{1}(i \in S)$.
That is, it would be enough to certify the inequality $\max_{w \in \{0,1\}^n, \sum w_i = \eta n} P(w) \leq O(\eta^{p/2}) + \eta^{p/2} \cdot (d^2 / n )^{p/4} \cdot d^{O(1)}$, where
\begin{align}
\label{eq:techniques-poly-expansion}
P(w) = \frac 1 {n^p} \sum_{i_1,\ldots,i_p \in [n]} \prod_{j \leq p} w_{i_j} \cdot \prod_{j \leq p} \iprod{X_{i_j},X_{i_{j+1}}} \, ,
\end{align}
and where we view the index $i_{p+1}$ as identified with the index $i_1$.

\paragraph{First attempt: matrix representation of $P$.}
So that the next steps do not appear to come out of nowhere, let us see why the most natural idea at this point fails.
The natural idea is that we could certify an upper-bound on $P$ using the maximum eigenvalue of the following matrix, indexed by $[n]^{p/2} \times [n]^{p/2}$:
\[
A[(i_1,\ldots,i_{p/2}),(i_{p/2+1},\ldots,i_p)] = \frac 1 {n^p} \cdot \prod_{j \leq p} \iprod{X_{i_j},X_{i_{j+1}}}  \, ,
\]
because 
\[
P(w) = (w^{\otimes p/2})^\top A w^{\otimes p/2} \leq \|A\|_{\text{op}} \cdot \|w\|^p = \|A\|_{\text{op}} \cdot (\eta n)^{p/2} \, .
\]
A straightforward \emph{lower bound} on $\|A\|_{\text{op}}$ is $\|A\|_F / \sqrt{\text{side length}} = \|A\|_F / n^{p/4}$, where $\|A\|_F$ is the Frobenius norm of $A$.
The most favorable situation for us would be if the order of magnitude of $A$'s maximum eigenvalue matches this lower bound, as would occur if $A$ were a random matrix with independent random entries.

Let us estimate the bound we would get in this optimistic scenario: a typical entry of $A$ is roughly $d^{p/2}/n^{p}$ in magnitude, so $\|A\|_F \approx d^{p/2}/n^{p/2}$, so we would be hoping for $\|A\|_{\text{op}}\approx d^{p/2} / n^{(3p)/4}$.
This would give us $P(w) \leq \eta^{p/2} (d^2 / n)^{p/4}$, slightly better than the bound we are hoping for!

Unfortunately, there is a clear obstacle to the hoped-for logic above: even though $A$ has side-length $n^{p/2}$, it actually has rank at most $d^{p/2}$, so its maximum eigenvalue has to be at least $\|A\|_F / d^{p/4}$---since $d < n$, this is substantially worse than the bound we wanted.
(One can see that the rank is at most $d^{p/2}$ by factoring $A$ as $BB^\top$ where $B$ has rows of the form $X_{i_1} \otimes X_{i_2} \otimes \ldots \otimes X_{i_{p/2}}$.)

\paragraph{Decomposing $P$, step 1: grouping by index-equality pattern.}
We will resuscitate the above approach by breaking $P$ into $p^{O(p)}$ terms.
This happens in two steps.
We first group terms based on the equality patterns among $i_1, \ldots, i_p$ and then apply a version of the Efron-Stein decomposition~\cite{efron1981jackknife} to each group.
The highest-order term within each group acts like we were hoping the matrix $A$ above would; the lower order terms we will be able to bound in a recursive fashion.
In particular, the decomposition will isolate low-rank components of the matrix, which cause too large spectral norms, and relates them to lower-dimensional instance of the problem (i.e., for values $p' < p$).

Given a partition $Q$ of $[p]$, say that $(i_1,\ldots,i_p) \in [n]^p$ \emph{respects} $Q$ if $i_j = i_k$ if and only if $j,k$ are in the same member of $Q$.
Let
\begin{align}
\label{eq:techniques-all-distinct}
P_Q(w) = \frac 1 {n^p} \sum_{\substack{i_1,\ldots,i_p \in [n] \\ \text{respects } Q}} \prod_{j \leq p} w_{i_j} \cdot \prod_{j \leq p} \iprod{X_{i_j},X_{i_{j+1}}} \;.
\end{align}
Then, since $P(w) = \sum_Q P_Q(w)$ and there are at most $p^{O(p)}$ distinct choices of $Q$, it will be enough to show that for each fixed $Q$ we can with high probability certify \new{the inequality} 
$P_Q(w) \leq O(\eta^{p/2}) + \eta^{p/2} \cdot (d^2 / n )^{p/4} \cdot d^{O(1)}$.

\paragraph{Decomposing $P$, step 2: Efron-Stein.}
For each partition $Q$, we further decompose the polynomial $P_Q$ in the style of the Efron-Stein decomposition \cite{efron1981jackknife}.
The important ideas in our proof are already captured by the case that $Q = \Set{\{1\},\{2\},\ldots,\{p\}}$, i.e., all of $i_1,\ldots,i_p$ are distinct, so we begin with this term.
In this case, there is one term in the Efron-Stein decomposition for each $S \subseteq [p]$.
We first define an auxiliary polynomial in $w$, given by
\[
P_Q^{\subseteq S}(w) = \frac 1 {n^p} \sum_{\substack{i_1,\ldots,i_p \in [n]\, \\ \text{all distinct}}} \prod_{j \leq p} w_{i_j} \cdot \E_{ \{X_{i_j} \}_{j \notin S}} \left [ \prod_{j \leq p} \iprod{X_{i_j},X_{i_{j+1}}} \right ] \;.
\]
That is, we take the expectation over those $X_{i_j}$ for which $j \notin S$.
Then, to do an inclusion-exclusion, we let $P_Q^{=S}(w) = \sum_{T \subseteq S} (-1)^{|S| - |T|} P_Q^{\subseteq S}(w)$.
Abusing notation and identifying $Q$ with $[p]$, this gives us
\[
P_Q(w) = P_Q^{=Q}(w) - \sum_{S \subsetneq [p]} (-1)^{|S| - |T|}  P_Q^{\subseteq S}(w) \, .
\]
We claim that $P_Q^{=[p]}(w)$ has a matrix representation that behaves like the optimistic heuristic we used above.
Further, note that for $S \subsetneq [p]$, $P_Q^{\sse S}$ \emph{almost} equals
\begin{align*}
    P_Q^{\sse S} \approx \frac {(\eta n)^{p-\card{S}}} {n^p} \sum_{\substack{i_1,\ldots,i_{\card{S}} \in [n]\, \\ \text{all distinct}}} \prod_{j \leq \card{S}} w_{i_j} \cdot \prod_{j \leq \card{S}} \iprod{X_{i_j},X_{i_{j+1}}} \,,
\end{align*}
which corresponds to a scaled version of Equation \cref{eq:techniques-all-distinct}, but containing only products of at most $\card{S}$ terms.
The almost equality can be turned into an exact one by slightly adjusting the scaling factor (see the ``squared inner product'' example a the end of this section for a similar argument), and we can bound this term by induction.
To show this equality, we crucially exploit the fact that $w_i^2 = w_i$.

When $Q \neq \Set{\{1\},\{2\},\ldots,\{p\}}$ is not the ``all-distinct'' partition, we proceed analogously, but in this case there is one term in the Efron-Stein decomposition for each  subset $S$ of the ``representatives'' of $Q$ (which we denote by $S \sse Q$).
The inductive form will still hold but requires more thought (see \cref{lem:main_induction} for full details).

\paragraph{A matrix for each $P_Q^{=Q}$.}
For every partition $Q$ (not necessarily the ``all-distinct'' one), we will find a matrix $A_{Q}$ whose entries are degree-$p$ polynomials in $X_1,\ldots,X_n$ such that for some $q_{Q} \leq p/2$, we have
\[
P_Q^{=Q}(w) = (w^{\otimes q_{Q}})^\top \cdot A_{Q} \cdot w^{\otimes q_{Q}} \, \, \text{ for every $w \in \{0,1\}^n$ with $\sum_{i\leq n} w_i = \eta n$} \,.
\]
Then, by computing operator norms $\|A_{Q}\|_{\text{op}}$, one obtains the certifiable bound, for any $X_1,\ldots,X_n$,
\begin{align}
\label{eq:techniques-matrix-expansion}
\max_{\substack{w \in \{0,1\}^n \\ \sum_{i \leq n} w_i = \eta n}} \sum_{Q} P_Q^{=Q}(w) \leq \sum_{Q} (\eta n)^{q_{Q}} \cdot \|A_{Q}\|_{\text{op}} \, .
\end{align}
The advantage over our ``first attempt'' is that we have used the fact that we only need a bound on $P(w)$ for those $w$ with $0/1$ entries and $\sum_{i \leq n} w_i = \eta n$ to allow us to bound $P_Q^{=Q}$ using eigenvalues of a larger family of random matrices---in particular, our random matrices now have dimension $n^{q} \times n^{q}$ for various choices of $q \leq p/2$.
Further, we show that these matrix representations of $P_Q^{=Q}$ are spectrally much more well-behaved than the na\"ive ones of $P_Q$.

Now the goal is to show that the right-hand side of \Cref{eq:techniques-matrix-expansion} is at most $O(\eta^{p/2}) + \eta^{p/2} \cdot (d^2 / n )^{p/4} \cdot d^{O(1)}$, with high probability over $X_1,\ldots,X_n$.
We use the \emph{graph matrix} approach to random matrices developed in \cite{AhnMP16} to obtain nearly-sharp bounds on the high-probability values of the operator norms $\|A_{Q}\|_{\text{op}}$, to show the bound we want on \Cref{eq:techniques-matrix-expansion}.
This step is the technical heart of our proof, and by far the most challenging: we use the structure of our decomposition of $P$, and its interplay with the graph matrix formalism, to argue a tradeoff between $q_{Q}$ and $\|A_{Q}\|_{\text{op}}$.
To do so, we develop a careful combinatorial accounting scheme to track the structure of dependencies among the entries of the matrices $A_{Q}$.

While the details are more complex, we offer the following intuition (and the example below) for why $P_Q^{=Q}$ should have a more tame spectral representation than $P_Q$ itself.
First, note that the na\"ive matrix representation of $P_Q$ is not even centered and the $P_Q^{\sse \emptyset}$ term in the Efron Stein decomposition removes exactly this expectation term.
Further, even after centering, there are still shared biases among the entries, causing a low-rank structure (for instance, if $\norm{x_1}$ is abnormally large, all entries containing $x_1$ tend to be slightly larger in magnitude).
A key observation is that this low-rank structure exactly corresponds to parts of the polynomial depending on strictly fewer than $p$ variables, with the exact rank scaling proportionally with the number of variables.
Thus, the Efron-Stein decomposition is grouping these low-rank terms together (and we eventually apply an induction to them).

\paragraph{Example: the Squared-Inner-Product Polynomial.}
To illustrate the ideas above, we present a simple example---the bounds here do not go directly into our main argument, but the technique we use here is a special case of the technique we develop for the Schatten $p$-norm polynomial described above.
Let $X_1,\ldots,X_n \sim \cN(0,I_d)$ be $d$-dimensional Gaussian vectors, with $n \gg d$.
Our goal will be to give an algorithm which can with high probability certify
\[
\max_{w \in \{0,1\}^n, \sum_{i \leq n} w_i = \eta n} \sum_{i, j \leq n} w_i w_j \cdot (\iprod{X_i,X_j}^2 - \E \iprod{X_i, X_j}^2) \leq \tilde{O}(\eta d n^{1.5}) \, .
\]
(We will see below why this is a reasonable bound to hope for.)

First, following the strategy of partitioning the indices $i,j$ into equality patterns, we can split into two terms depending on whether $i=j$ or $i \neq j$.
The $i \neq j$ term is more interesting; we will ignore the $i = j$ term.
So we can look at $\sum_{i \neq j} w_i w_j ( \iprod{X_i, X_j}^2 - \E \iprod{X_i,X_j}^2)$.

\subparagraph{Optimistic heuristic.}
The first thing we might try again is to write
\[
\sum_{i \neq j} w_i w_j (\iprod{X_i,X_j}^2 - \E \iprod{X_i, X_j}^2) = w^\top A w \leq \|w\|^2 \|A\|_{\text{op}} = \eta n \cdot \|A\|_{\text{op}} \, ,
\]
where $A \in \R^{n \times n}$ has entries $A_{ij} = (\iprod{X_i,X_j}^2 - \E \iprod{X_i, X_j}^2)$ for $i \neq j$ and $\text{diag}(A) = 0$.
Now, a typical entry of $A$ is $\approx \pm d$.
The entries are not independent, but if they were, we would expect that the maximum eigenvalue of $A$ is around $d \sqrt{n}$.
Thus, an optimistic bound would be of order $\eta n \cdot d \sqrt{n} = \eta d n^{1.5}$.

\subparagraph{Failure of the optimistic heuristic.}
However, the failure of independence of the entries of $A$ actually leads to an eigenvalue larger than $d \sqrt{n}$, meaning that the above argument is too optimistic.
To see that a too-large eigenvalue must exist, we can examine $\E {\bf 1}^\top A v$, where ${\bf 1}$ is the all-$1$'s vector, and $v_i = \sum_{k \leq d} (X_i(k)^2 - 1) = \|X_i\|^2 - d$.\footnote{This choice is inspired by our earlier heuristic: When $\norm{x_i}$ is large, all entries of the $i$-th row tend to be slightly larger.}
Using that $\E \iprod{X_i,X_j}^2 = d$, we gt
\[
\E {\bf 1}^\top A v = \sum_{i \neq j \leq n} \sum_{k \leq d} \E\Brac{ (\iprod{X_i,X_j}^2 - d) X_j(k)^2} \geq \sum_{i \neq j \leq n} \sum_{k \leq d} \E \Brac{X_i(k)^2 X_j(k)^4 - X_j(k)^2} = \Omega(n^2 d) \, .
\]
Since $\|{\bf 1}\| = \sqrt{n}, \|v\| \approx \sqrt{dn}$, this witnesses an eigenvalue of $A$ of magnitude $\Omega(n\sqrt{d}) \gg d \sqrt{n}$.

There is still hope to recover the optimistic bound, though, as long as the eigenvector(s) associated with this larger eigenvalue are not too aligned with the indicator of any set of $\eta n$ coordinates.
At this point, the standard approach is to expand the matrix $A$ in the basis of \emph{graph matrices}, which roughly corresponds to taking an expansion of each entry $\iprod{X_i,X_j}^2 - \E \iprod{X_i,X_j}^2$ in Hermite polynomials, to try and separate this ``bad'' eigenvector out from the rest.

This is what we will do, but when we apply these techniques to the Schatten-$p$ matrix, the difficulty is a combinatorial explosion of $p^{O(p)}$ different matrices.
The Efron-Stein decomposition gives us a way to group these matrices into like terms.
We illustrate the idea on $A$; for simplicity we avoid the graph matrix terminology and focus on Efron-Stein.

\subparagraph{Efron-Stein isolates the ``bad'' part of $A$.}

So let us apply Efron-Stein to the polynomial $p(w) = w^\top A w$.
We obtain four different polynomials.
\begin{align*}
p^{\sse \Set{i,j}}(w) & = \sum_{i \neq j} w_i w_j (\iprod{X_i,X_j}^2 - d) = p(w)\\
p^{\sse \Set{j}}(w) & = \sum_{i \neq j} w_i w_j \E_{X_i} (\iprod{X_i,X_j}^2 - d) = \sum_{i \neq j} w_i w_j (\|X_j\|^2 - d) \\
p^{\sse \Set{i}}(w) & = \sum_{i \neq j} w_i w_j \E_{X_j} (\iprod{X_i,X_j}^2 - d) = \sum_{i \neq j} w_i w_j (\|X_i\|^2 - d) \\
p^{\sse \emptyset}(w) & = \sum_{i \neq j} w_i w_j \E_{X_i,X_j} (\iprod{X_i,X_j}^2 - d) = 0 \, .
\end{align*}
So we can decompose $p(w) = w^\top A w$ as
\footnote{The polynomials $p^{\sse \Set{i}}(w)$ and $p^{\sse \Set{j}}(w)$ are identical in this case, but might not be for more general polynomial.}
\[
w^\top A w = (p(w) - p^{\sse \Set{i}}(w) - p^{\sse \Set{j}}(w)) + p^{\sse \Set{i}}(w) + p^{\sse \Set{j}}(w) \, .
\]
Observe that $p^{\sse \Set{i}}$ and $p^{\sse \Set{j}}$ exactly correspond to a quadratic form of the $w$ with the ``bad'' matrix ${\bf 1}v^\top$, so the first polynomial $p(w) - p^{\sse \Set{j}}(w) - p^{\sse \Set{i}}(w)$ is removing this ``bad'' piece of $w^\top A w$.
Indeed, it turns out that if we write down a matrix $A'$ such that $w^\top A' w = p(w) - p^{\sse \Set{j}} - p^{\sse \Set{i}}(w)$, then we will be able to show $\|A'\|_{\text{op}} \leq \tilde{O}(d \sqrt{n})$, recovering the optimistic bound from before.\footnote{This bound easily follows using the graph matrix approach.}

But what about $p^{\sse \Set{j}}(w)$ and $p^{\sse \Set{i}}(w)$?
We can certify bounds on them by factoring out $\sum_i w_i$ and $\sum_j w_j$, respectively.
That is, we can bound
\[
p^{\sse \Set{j}}(w) = \sum_{i \neq j} w_i w_j (\|X_j\|^2 - d) = \sum_{i,j} w_i w_j (\|X_j\|^2 - d) - \sum_i w_i^2 (\|X_i\|^2 - d) = \Paren{\sum_{i} w_i - 1} \cdot \iprod{w, v}
\]
and similarly for $p^{\sse \Set{i}}$.
The final expression is at most $\eta n \cdot \|w\| \cdot \|v\| \leq (\eta n)^{1.5} \cdot \|v\|$.
So, by computing $\|A'\|_{\text{op}}$ and $\|v\|$, we can in polynomial time certify the bound
\[
w^\top A w \leq \eta n \cdot \|A\|_{\text{op}} + (\eta n)^{1.5} \cdot \|v\| \approx \eta n^{1.5} d + \eta^2 n^{1.5} d 
\]
which, since $\eta \leq 1$, is exactly what we wanted.

\section{Preliminaries}
\label{sec:prelims}
\paragraph{Linear algebra background.} For a vector $x$,  $\|x\|$, unless specified,  corresponds to the Euclidean norm.
For a matrix $H$ and even $p$,
the $p$-Schatten norm of $H$ is denoted by $\|H\|_p := \trace( (H^\top H)^{p/2})^{1/p}$, which also equals $(\sum_i \lambda_i ^p)^{1/p}$ where $\lambda_i$'s are the singular values of the matrix $H$.
The operator norm is denoted by $\|H\|_\op$ and corresponds to $\|H\|_\infty$.
We use $I_d$ to denote the $d\times d$ identity matrix, omitting the subscript when it is clear from context. 

\paragraph{Sub-gaussian distributions.}
We say that a distribution $D$ over $\R^d$ with mean $\mu$ is $s$-subgaussian, if for all $k \in \N$ and every unit vector $v \in \R^d$ it holds that $\E_{x \sim D} \iprod{x-\mu,v}^k \leq (C s \sqrt{k})^k$ for a constant $C$.
Similarly, a distribution $D$ over $\R^d$ with mean $\mu$ and covariance $\Sigma$ is termed $s$-\emph{hypercontractive}-subgaussian, if for all $k \in \N$ and every unit vector $v \in \R^d$ it holds that $\E_{x \sim D} \iprod{x-\mu,v}^k \leq (C s \sqrt{k} \sqrt{v^\top \Sigma v})^k$ for a constant $C$; observe that $s\gtrsim 1$ for this to hold.
We refer the reader to \cite{Vershynin18} for further details.

\subsection{Sum-of-Squares Proofs and Pseudo-Expectations}

In this section, we will introduce sum-of-squares proofs and pseudo-expectations.
We will introduce what is necessary for this work, we refer to~\cite{barak2014sum,barak2016proofs} for more background.
We say a polynomial $h$ is a \emph{sum of squares}, if it can be written as a sum of squared polynomials.
All polynomials are over the real numbers in this work.

\paragraph{Sum-of-squares proofs.}

Consider a system of $m \in \N$ polynomial inequalities in $N$ formal variables $X = (X_1, \ldots, X_N)$:
\[
    \cA = \Set{q_1(X) \geq 0, \ldots, q_m(X) \geq 0} \,.
\]
Let $p$ be another polynomial in $X$.
We say that there is a \emph{sum-of-squares proof} (SoS proof) that $\cA$ implies that $p(X) \geq 0$, if we can write $p(X) = \sum_{S \sse [m]} b_S(X) \prod_{i \in S} q_i(X)$, where each $b_S$ is a polynomial that is a sum of squares.
We say that the sum-of-squares proof has degree $t$, if each summand has degree at most $t$.
We denote this by $\cA \proves_t^{X} p \geq 0$.
Whenever the set of variables is clear from context, we may omit the superscript.
We write $\cA \proves_t^{X} p \geq p'$ if $\cA \proves_t^{X} p - p' \geq 0$.
We frequently use (most often implicitly) that sum-of-squares proof can naturally be composed, that is, if $\cA \proves_t^{X} p_1 \geq p_2$ and $\cA \proves_{t'}^{X} p_2 \geq p_3$ then it also holds that $\cA \proves_{\max\Set{t,t'}}^{X} p_1 \geq p_3$.

\paragraph{Pseudo-expectations.}

Pseudo-expectations are the convex duals of sum-of-squares proofs.
We give the following definition.
\begin{definition}%
\label{def:pE}
For every even $d \in \N$, a \emph{degree-$d$ pseudo-expectation} $\pE$ is a linear operator on the space of degree-$d$ polynomials such that $\pE 1 = 1$ and $\pE p^2(X) \geq 0$ for all polynomials $p$ of degree at most $d/2$.
\end{definition}
Given a system of polynomial inequalities $\cA = \Set{q_1 \geq 0, \ldots, q_m \geq 0}$, we say that a pseudo-expectation $\pE$ \emph{satisfies} $\cA$ if for all $S \sse [m]$ and all polynomials $h$ that are sum of squares and such that $\deg(h \cdot \prod_{i \in S} q_i) \leq d$ it holds that $\pE h \cdot \prod_{i \in S} q_i \geq 0$.
We say that $\pE$ \emph{approximately satisfies} $\cA$ if under the same conditions on h and $S$,
it instead holds that $\pE h \cdot \prod_{i \in S} q_i \geq -2^{-N^{\Omega(d)}} \norm{h}_2 \prod_{i \in S} \norm{q_i}_2$, where the norms are the norms of the coefficient vectors of the polynomials (in the monomial basis).
The normalization and positivity constraint are satisfied exactly also in the approximate case.

Pseudo-expectations and sum-of-squares proofs interact in the following way (see \Cref{fact:sos-duality,fact:strong_duality_alternate} for further connections).
\begin{fact}
\label{fact:pE_and_sos_proofs}
Let $\pE$ be a degree-$d$ pseudo-expectation that satisfies $\cA$ and suppose that $\cA \proves_t p \geq 0$.
Let $h$ be an arbitrary sum of squares polynomial such that $\deg(h) + t \leq d$, then it holds that $\pE h \cdot p \geq 0$. 
In particular, it holds that $\pE p \geq 0$ if $t \leq d$.
If $\pE$ only approximately satisfies $\cA$ it holds that $\pE p \geq - 2^{-N^{\Omega(d)}} \norm{h}_2$.\footnote{Strictly speaking, we need the additional condition that the bit-complexity of the sum-of-squares proof is suitably bounded. This will be the case in all proofs we consider and we will not mention this explicitly.
See ~\cite{o2017sos,raghavendra2017bit} for a more in-depth discussion.}
\end{fact}

As long as the bit-complexity of the constraints in $\cA$ is polynomially bounded (in $m+N$), we can efficiently find pseudo-expectations approximately satisfying $\cA$.
We call a system $\cA$ feasible, if there exists an instantiation of the formal variables $X$ that satisfies all constraints.
\begin{fact}[\cite{parrilo2000structured,lasserre2001new}]
\label{fact:pE_efficient_optimization}
Let $\cA = \Set{q_1 \geq 0, \ldots, q_m \geq 0}$  over formal variables $X = (X_1, \ldots, X_N)$ such that the bit-complexity of all $q_i$ is at most $(m+N)^{O(1)}$.
Then, we can find a degree-$d$ pseudo-distribution that approximately satisfies $\cA$ in time $(m+N)^{O (d)}$.\footnote{Technically, we also need that $\cA$ contains a constraint of the form $\norm{X}_2^2 \leq B$ for some large $B$, e.g., $B = (m+N)^{\Omega(1)}$. In all our applications we can add this without changing the proofs and will not mention it explicitly.}
\end{fact}
In all of our applications the additional slack factor of $2^{-N^{\Omega(d)}}$ times the norm of the coefficient vectors introduced by approximately satisfying pseudo-expectations will be negligible.
We thus assume instead that we can find a pseudo-expectation that \emph{exactly} satisfies $\cA$ (in the same time).

We will also use the following closely-related fact
\begin{fact}[\cite{barak2014sum}]
    \label{fact:check_sos_proof}
    Let $\e \geq 2^{-n^{100}}$.
    Let $p$ be a polynomial and $\cA$ be as above.
    We can decide in time $n^{O(t)}$ whether $\cA \proves_{O(t)}^X p(X) \geq -\e$.%
    \footnote{Again, formally we require that there exists an SoS proof of polynomial bit complexity. We will also pretend $\e = 0$ throughout this paper, all error terms introduced by this will be negligible.}
\end{fact}

We will need the following version of the Cauchy-Schwarz inequality for pseudo-expectations.
\begin{fact}
\label{fact:pE_cauchy_schwarz}
Let $\pE$ a degree-$d$ pseudo-distribution and $p$ be a polynomial of degree at most $d/2$.
Then, it holds that $(\pE p(X ))^2 \leq \pE [p^2(X)]$.
\end{fact}

We shall also use the following duality between SoS proofs and Pseudoexpectations.
\begin{fact}[Duality Between SoS Proofs and Pseudoexpectations~\cite{JosHen16-archimedian}]
\label{fact:sos-duality}
We say that a system of polynomial inequalities $\cA$ over formal variables $X=(X_1,\dots,X_N)$ is Archimedean if for some real $B > 0$ it contains the inequality $\sum_i X_i^2 \leq B$.
For every Archimedian system $\cA$ and every polynomial $p$ and every degree $t$, exactly one of the following holds.
\begin{enumerate}[leftmargin=*]
\item For every $\eps > 0$, there is an SoS proof of  $\cA
\proves_{m}^X
p(X) + \eps \geq 0$.
\item There is a degree-$m$ pseudoexpectation satisfying $\cA$ but $\pE [p(X)] < 0$.
\end{enumerate}
\end{fact}

We also need the following strong duality fact that does not impose Archimedean property over the system of polynomial inequalities, it is used in~\cref{sec:transfer-lemma}.
In the following, the formal variables are $(w_1,\dots,w_n, v_1,\dots,v_d)$ and $\mathcal{B}$ denotes the system of polynomial inequalities. 
\begin{fact}[Strong Duality for Our Constraint System]
    \label{fact:strong_duality_alternate}
    Let $n,t \in \N, \eta \in [0,1]$ be such that $\eta n$ is an integer.
    Let $\cB = \{\forall i \in [n] \colon w_i^2 = w_i \,, \sum_{i=1}^n w_i = \eta n\}$ be a system of polynomial (in)equalities over $w$ and $v$; observe that there are no constraints on the variable $v$.
    Then for every $t\in \N$ and every polynomial $q$ in variables $w$ and $v$, exactly one of the following two statements is true:
    \begin{enumerate}[label=(\roman*)]
        \item For every $\e > 0$, there exists an SoS proof that $\cB \proves_t^{w,v} q(w, v) + \e \geq 0$.
        \item There is a degree-$t$ pseudo-expectation $\pE$ satisfying $\cB$ such that $\pE q(w, v) < 0$.
    \end{enumerate}
\end{fact}

\begin{proof}
    The fact is implied by the following statement
    \[
        S \coloneqq \sup \Set{C \mid \cB \proves_t^{w,v} q(w, v)\geq C} = \inf \Set{\pE q \mid \pE \text{ of degree-$t$, satisfying $\cB$}} \eqqcolon I \,.
    \]
    Indeed, suppose this statement holds.
    If (i) is true then $S \geq 0$ and hence $I \geq 0$, implying that (ii) cannot hold.
    On the other hand, if (ii) holds, then $I < 0$ and hence $S < 0$ so (i) cannot hold (since 1.\ would imply $S \geq 0$).

    The optimization problems implicit in this statement are dual semi-definite programs and we always have $S \leq I$  by weak duality (see, e.g., \cite{laurent2009sums}).
    We will show that in fact strong duality holds, implying $S = I$.
    Note that both programs are feasible, it is thus sufficient to exhibit a strictly feasible point for the program on the right-hand side, i.e., find a small ball in the relative interior of the feasible region.\footnote{For a domain $\mathcal{X}$, the relative interior of  $\mathcal{X}$ is defined to be the set of all points $x \in \cX$ such that the intersection of a small enough ball around $x$ and the affine hull of $\cX$ belongs to $\cX$.}

    We represent pseudo-expectations by their moment-matrices.
    The condition that they satisfy $w_i^2 = w_i$ and $\sum_{i=1}^n w_i = \eta n$ at degree-$t$ correspond to various affine constraints on these matrices (they also need to be positive semi-definite and the entry corresponding to the constant 1 polynomial needs to be 1).
    For a moment-matrix $M$ and any vector $\vec{p}$, $\vec{p}^\top M \vec{p} = \pE p^2$ for the polynomial $p$ with coefficient vector $\vec{p}$.
    Note that the condition that $M$ satisfies our constraints imply that there exists several vectors $\vec{p}$ such that $\vec{p}^\top M \vec{p} = 0$, i.e., the matrix $M$ might have 0 eigenvalues.
    In particular, this holds for all polynomials $p$, such that $p(w,v) = \sum_{i=1}^n (w_i^2 - w_i) \lambda_i(w,v) + (\sum_{i=1}^n w_i - \eta n) \lambda_{n+1}(w_i,v)$ for polynomials $\lambda_1, \ldots, \lambda_{n+1}$ such that  all summands have degree at most $\ell$.
    It can be checked that a matrix is in the relative interior, if and only if all eigenvectors corresponding to zero eigenvalues are of this form \cite[Theorem 6.5]{laurent2009sums} (this result first appeared in \cite{marshall2003optimization}).

    We claim that the pseudo-expectation, denoted by $\pE_{w,v}$, which chooses $w_1, \ldots, w_n$ such that $\sum_{i=1}^n w_i = \eta n$ uniformly at random and $v \sim \cN(0,I_d)$ independent of $w$, is strictly feasible.
    
    We denote by $\pE_{w,v}$ the pseudo-expectation described above and by $\pE_{w}$ its restriction to only the $w$ variables.
    We first argue that $\pE_{w}$ is a strictly feasible solution for the setting when the vectors/polynomials that only depend on the $w$ variables.
    Indeed, this follows directly from \cite[Lemma 6.1]{raghavendra2017bit}.
    Their notion of $d$-completeness is exactly the eigenvalue property described above.

    Next, consider a polynomial $p$ in variables both $w$ and $v$ such that $\pE_{w,v} p^2(w,v) = 0$.
    Consider the Hermite decomposition of $p$, where $\alpha$ indicates a multi-index over $[d]$:
    \[
        p(w,v) = \sum_{\card{\alpha} \leq t} \hat{p}_\alpha(w) h_\alpha(v) \,,
    \]
    where $\hat{p}_\alpha(w) = \E_{z \sim \cN(0,I_d)} p(w,z) h_\alpha(z)$.
    Note that since $p$ has degree at most $d$, $\hat{p}$ has degree at most $t - \card{\alpha}$ by orthogonality of Hermite polynomials.
    Using orthogonality of the Hermite polynomials again, it follows that
    \[
        0 = \pE_{w,v} p^2(w,v) = \sum_{\card{\alpha} \leq t} \pE_{w} \hat{p}^2_\alpha(w) = 0\,.
    \]
    Thus, for all $\alpha$, we must have that $\pE_{w} \hat{p}^2_\alpha(w_i) = 0$.
    Leveraging the previously established fact for $\pE_w$,
    there must exist $\lambda_1^{(\alpha)}, \ldots, \lambda_n^{(\alpha)}, \lambda_{n+1}^{(\alpha)}$ such that
    \[
      \hat{p}_\alpha(w) = \sum_{i=1}^n (w_i^2 - w_i) \lambda^{(\alpha)}_i(w) +  \Paren{\sum_{i=1}^n w_i - \eta n} \lambda^{(\alpha)}_{n+1}(w) 
    \]
    and all summands have degree at most $t - \card{\alpha}$.
    It follows that
    \[
        p(w,v) = \sum_{i=1}^n (w_i^2 - w_i) \cdot \Paren{\sum_{\card{\alpha} \leq t} h_\alpha(v) \lambda_i(w)^{(\alpha)}} +  \Paren{\sum_{i=1}^n w_i - \eta n} \cdot \Paren{ \sum_{\card{\alpha} \leq t} h_\alpha(v)\lambda_{n+1}^{(\alpha)}(w) }
    \]
    and all summands have degree at most $t$.
\end{proof}

\subsection{Moments and Polynomials}

\paragraph{Wick's Theorem.}
Let $g \sim \cN(0,I_d)$ be a $d$-dimensional standard Gaussian.
\begin{theorem}[Wick's theorem]
\label{thm:wick_moment}
  Let $u_1,\ldots,u_n \in \R^d$ with $n$ even. Let $\cM_n$ be the set of matchings on $[n]$.
  Then
  \[
  \E \prod_{i \leq n} \iprod{g,u_i} = \sum_{M \in \cM_n} \prod_{\{i,j\} \in M} \iprod{u_i,u_j} \, .
  \]
\end{theorem}
Wick's theorem admits the following useful corollary:

\begin{corollary}
    \label{cor:wicks-with-norm}
    Let $u_1,\ldots,u_n \in \R^d$ with $n$ even. Let $\cM_n$ be the set of matchings on $[n]$.
    Let $k$ be even.
    Then
    \[
    \E \|g\|^k \cdot \prod_{i \leq n} \iprod{g,u_i} = O(d + k)^{k/2} \cdot \sum_{M \in \cM_n} \prod_{\{i,j\} \in M} \iprod{u_i,u_j} \, .
    \]
\end{corollary}
\begin{proof}
We will exploit that the norm of $g$ is independent of $\tfrac g {\norm{g}}$.
In particular, it holds that
$$
    \frac{\E \norm{g}^n}{\E \norm{g}^{n+k}} \cdot \E \|g\|^k \cdot \prod_{i \leq n} \iprod{g,u_i} = \frac{\E \norm{g}^n}{\E \norm{g}^{n+k}} \cdot \E \|g\|^{n+k} \cdot \E \prod_{i \leq n} \iprod{\tfrac g {\norm{g}},u_i} = \E \prod_{i \leq n} \iprod{g,u_i} \,.
$$
Using Wick's theorem and that $\E \norm{g}^{2l} = d \cdot (d+2) \cdot \ldots \cdot (d+2l - 2)$ it follows by rearranging that
\begin{align*}
\E \|g\|^k \cdot \prod_{i \leq n} \iprod{g,u_i} &= \frac{\E \norm{g}^{n+k}}{\E \norm{g}^{n}} \cdot  \sum_{M \in \cM_n} \prod_{\{i,j\} \in M} \iprod{u_i,u_j} \\
&= \frac{d \cdot (d+2) \cdot \ldots \cdot (d+n+k - 2)}{d \cdot (d+2) \cdot \ldots \cdot (d+n - 2)} \cdot  \sum_{M \in \cM_n} \prod_{\{i,j\} \in M} \iprod{u_i,u_j} \\
&= O(d+k)^{k/2} \cdot \sum_{M \in \cM_n} \prod_{\{i,j\} \in M} \iprod{u_i,u_j} \,.
\end{align*}
\end{proof}

\paragraph{Hermite Polynomials.} Hermite polynomials form a complete orthogonal basis of the vector space $L^2(\R,\cN(0,1))$ of all functions $f:\R \to \R$ such that $\E_{X\sim \cN(0,1)}[f^2(X)]< \infty$. The \emph{probabilist's} Hermite polynomials $H_{e_k}$ for $k\in \N$ satisfy $\int_\R H_{e_k}(x) H_{e_m}(x) e^{-x^2/2} dx/\sqrt{2\pi} = k!   \mathbf{1}(k=m)$. 
We will use the \emph{normalized probabilist's} Hermite polynomials, $h_k(x) := H_{e_k}(x)/\sqrt{k!}$, $k\in \N$ for which $\int_\R h_k(x) h_{m}(x) e^{-x^2/2} dx/\sqrt{2\pi} = \mathbf{1}(k=m)$.

\begin{fact}[Properties of Hermite Polynomials] 
\label{fact:hermite-poly}
We have the following properties:
\begin{enumerate}[label=(H.\arabic*),ref=\thefact\,(H.\arabic*)]
    \item \label[subfact]{item:her-unit-variance} For any $x$, $\E_{X \sim \cN(x, 1)} [h_j(X)] = x^i/\sqrt{i!}$.
    \item \label[subfact]{item:her-ornstein} For $\sigma^2 \leq 1$, let $\rho^2 := 1 - \sigma^2$.
    Then $\E_{X \sim \cN(x, \sigma^2)} [h_j(X)]
    =\rho^i h_j(x/\rho)$.
    \item \label[subfact]{item:her-near-zero} There exists a constant $c$, such that for any $|x|^2 \leq 2i$, $h_i^2(x) \leq c e^{x^2/2} / \sqrt{i}$. In particular, for all $|x| \leq 1$, $h_i^2(x) \leq 3C/\sqrt{i}$~\cite[Theorem 1 (i)]{BonCla90}.
    \item     \label[subfact]{fact:hermite_monomial_expansion}
    Let $\ell \geq 1$, then
    \[
        x^\ell = \sum_{m=0}^{\floor{\frac{\ell}{2}}} \frac{\ell!}{2^m m! \sqrt{(\ell-m)!}} h_{\ell - 2m}(x) \,.
    \]
\end{enumerate}    
\end{fact}

\paragraph{Moment matching.} We shall use the following consequence of the duality between moments and non-negative polynomials:
\begin{fact}[{\cite[Lemma 3.148]{BlePT12}}]
 \label{fact:moment-matching}
 There    exists a distribution $F$ over $[-B,B]$ with the first three moments equal to $x_1,x_2,x_3$ if and only if the following matrices are PSD\footnote{The result in \cite[Lemma 3.4.18]{BlePT12} is stated for $B=1$ but the general case follows by appropriate normalization.}:
\begin{align}
    \begin{bmatrix}
        1 + \frac{x_1}{B} & \frac{x_1}{B}+ \frac{x_2}{B^2} \\
        \frac{x_1}{B}+ \frac{x_2}{B^2} &  \frac{x_2}{B^2} + \frac{x_3}{B^3} 
    \end{bmatrix} \,\, \text{and}\,\,
    \begin{bmatrix}
        1 - \frac{x_1}{B} & \frac{x_1}{B}- \frac{x_2}{B^2} \\
        \frac{x_1}{B}- \frac{x_2}{B^2} &  \frac{x_2}{B^2} - \frac{x_3}{B^3} 
    \end{bmatrix} \,.
\end{align}
Moreover, if such a distribution $F$ exists, then one can assume that $F$ is supported on a set of 4 elements~\cite[Lemma 4.5]{DiaKPPS21}. 
\end{fact}

\section{Certificates for Sparse Singular Values and Operator Norm Resilience}
\label{sec:deriving_final_certificates}
In this section we will derive certificates for sparse singular values of random matrices (cf.~\cref{thm:main-intro}).
For many of the applications,
a slightly different (but equivalent) formulation will be more convenient:
operator norm induced by a small subset of the rows of the input matrix (see below for the exact definition).
We refer to this as \emph{operator norm resilience}.
We will derive our certificates for Gaussian data in this section and, using the tools developed in~\cite{diakonikolas2024sos}, transfer these to subgaussian data in~\cref{sec:transfer-lemma}.

As discussed before, we obtain our certificates by certifying upper bounds on the Schatten-$p$ norm of matrices of the form $\sum_{i=1}^n w_i x_i x_i^\top$, where the $w_i$ select a small subset of the indices.
More specifically, we will use the following theorem that we will prove in~\cref{sec:sos_certificate_schatten_p}.
\begin{theorem}[SoS bound on Schatten-$p$ norm]
\torestate{
\label{thm:sos-schatten-p}
Let $n,d,p \in \N$  be such that $p \geq 2, d \leq n \leq d^2$ and let $\eta \in [0,1]$.
Further, let $x_1,\ldots,x_n \sim \cN(0,I_d)$ and let $w_1,\ldots,w_n$ be indeterminates.
Then, 
\begin{align*}
\E_{x_1,\ldots,x_n} &\inf \left \{  B \in \R \text{ s.t. } \Set{\{w_i^2 = w_i\}_{i \in [n]}, \sum_{i \leq n} w_i \leq \eta n }\proves_{O(p)}^{w} \left \| \frac 1 n \sum_{i \leq n} w_i x_i x_i^\top \right \|_p^p \leq B \right \} \\
&\leq \Paren{ p^p \cdot \log^p(n)}^{O(1)} \cdot d^{5/2} \cdot \max\Set{\frac d n, \frac{\eta d}{\sqrt{n}}}^{p/2} \, .
\end{align*}
}
\end{theorem}
In particular, 
since $n \leq d^2$, 
when $n = (p \cdot \log d)^{O(1)} \cdot\max\Set{\tfrac d {\alpha^2},\tfrac{\eta^2 d^2}{\alpha^4}}$ the above is at most $d^{3/2}n^{1/2} \alpha^p$.\footnote{The condition that $n \geq d$ could potentially be removed by including another term in the maximum.}

From \Cref{thm:sos-schatten-p}, the operator norm resilience certificates and the sparse singular values certificate follow rather easily.
For ease of notation, we define the following quantity which will appear in our bounds:
\begin{align}
    \label{eq:sosBound}
    \sosBound :=\max\Set{(p \cdot \log n)^{O(1)} \cdot d^{5/(2p)} \cdot \max\Set{\sqrt{\frac d n}, \frac{\sqrt{\eta d}}{n^{1/4}}}, O(1)\sqrt{\eta}} \,.
\end{align}

For technical reasons related to the application to covariance estimation, it turns out that the following form of certificate is helpful.
In particular, it will be helpful to not require the unit norm constraint on the vector-valued variable $v$.\footnote{Ideally, we would like to remove the $p$-th powers also, but it is unclear to us how to do it within SoS; observe that \Cref{fact:sos_square_root} is inapplicable towards that goal.}

\begin{theorem}
  \label{thm:resilience_Gauss_full_alternate}
  Let $p,n,d \in \N$ be such that $p \geq 2$ is a power of 2 and $d \leq n$.
  Let $x_1, \ldots, x_n$ be \iid samples from $\cN(0,I_d)$.
  Then for $\sosBound$ defined in \Cref{eq:sosBound},
  \begin{align*}
    \E_{x_1,\ldots,x_n} &\inf \left \{  B \in \R \text{ s.t. } \Set{\{w_i^2 = w_i\}_{i \in [n]}\,, \sum_{i \leq n} w_i \leq \eta n }\proves_{O(p)}^{w,v} \Paren{\tfrac 1 n \sum_{i=1}^n w_i \iprod{x_i,v}^2}^p \leq B^p \norm{v}^{2p} \right \} \\
    &\leq \sosBound\, .
    \end{align*}
  Further, with probability at least $1-\delta$, the following sum-of-squares proof exists
  \begin{align*}
    \Biggl\{w_i^2 &= w_i \,, \sum_{i=1}^n w_i \leq \eta n \Biggr\} \proves_{O(p)}^{w_i, v} \Paren{\tfrac 1 n \sum_{i=1}^n w_i \cdot \iprod{x_i, v}^2}^p 
    \leq \biggl[ \sosBound^p +\Paren{\frac {O(1)\log(1/\delta)} n}^p \biggr] \cdot \norm{v}^{2p}\,.
  \end{align*}
\end{theorem}

We can easily derive a cleaner version with the unit norm constraint from this, which we state in \cref{thm:resilience_Gauss_full} below.
\begin{theorem}
  \label{thm:resilience_Gauss_full}
  Let $p,n,d \in \N$ be such that $p \geq 2$ and $d \leq n$.
  Let $x_1, \ldots, x_n$ be \iid samples from $\cN(0,I_d)$.
  Then for $\sosBound$ defined in \Cref{eq:sosBound},
  \begin{align*}
    \E_{x_1,\ldots,x_n} &\inf \left \{  B \in \R \text{ s.t. } \Set{\{w_i^2 = w_i\}_{i \in [n]}\,, \sum_{i \leq n} w_i \leq \eta n \,, \norm{v}^2 = 1}\proves_{O(p)}^{w,v} \tfrac 1 n \sum_{i=1}^n w_i \iprod{x_i,v}^2 \leq B \right \} \\
    &\leq \sosBound,
    \end{align*}
  Furthermore, with probability at least $1-\delta$, the following sum-of-squares proof exists
  \begin{align*}
    \Biggl\{w_i^2 &= w_i \,, \sum_{i=1}^n w_i \leq \eta n \,, \Norm{v}^2 = 1 \Biggr\} \proves_{O(p)}^{w_i, v} \tfrac 1 n \sum_{i=1}^n w_i \cdot \iprod{x_i, v}^2 \leq \sosBound +  \frac{O(1)\log(1/\delta)}{n}\,.
  \end{align*}
\end{theorem}
The above version follows by using the constraint that $\norm{v}^2 = 1$ and noting that for every constant $C > 0$, it holds that $\Set{X \geq 0 \,, X^p \leq C^p} \proves_{p}^X X \leq C$ (\Cref{fact:sos_square_root}).
Note that we do not require $p$ to be a power of 2 anymore, since we can invoke~\cref{thm:resilience_Gauss_full_alternate} with the next largest power of 2.
Technically, this leads to an additional constant in the resulting bound, but for the ease of presentation, we incorporate that constant in the original definition of $\sosBound$ in \Cref{eq:sosBound} itself. 

We now give the proof of~\cref{thm:resilience_Gauss_full_alternate} from \Cref{thm:sos-schatten-p}.
\begin{proof}[Proof of~\cref{thm:resilience_Gauss_full_alternate}]
  Let $p, n, d$ be as in the theorem statement.
    We prove only the first statement, the second one follows by the methodology of \cite{diakonikolas2024sos} and is deferred to \Cref{thm:op_resilience_alternate_subgaussian}.

  First consider the case when $n\geq d^2$. In this regime, \cite{barak2012hypercontractivity} already gives a gives a degree-four SoS bound of $O(1)\sqrt{\eta}$ on $\sum_iw_i \langle x_i,v\rangle^2$; see the discussion in \Cref{sec:techniques}.
  Since $\sosBound$ is sufficiently larger than $\sqrt{\eta}$, the setting of $n\geq d^2$ is handled.

  Consider now the regime of $n \leq d^2$, which was the focus of \cref{thm:sos-schatten-p}. It implies that 
  \begin{align*}
    \E_{x_1,\ldots,x_n} &\inf \left \{  B \in \R \text{ s.t. } \Set{\{w_i^2 = w_i\}_{i \in [n]}, \sum_{i \leq n} w_i \leq \eta n }\proves_{O(p)}^{w} \left \| \frac 1 n \sum_{i \leq n} w_i x_i x_i^\top \right \|_p^p \leq B \right \} \\
    &\leq \Paren{ p^p \cdot \log^p(n)}^{O(1)} \cdot d^{5/2} \cdot \max\Set{\frac d n, \frac{\eta d}{\sqrt{n}}}^{p/2} \, .
    \end{align*}
  Applying \cref{clm:SoS_op_to_schatten_p} with $M = \tfrac 1 n \sum_{i=1}^n w_i x_i x_i^\top$ (which is symmetric), it follows that
  \[
    \proves_{O(p)}^{w,v} \Paren{\tfrac 1 n \sum_{i=1}^n w_i \iprod{x_i, v}^2}^p \leq \Norm{ \tfrac 1 n \sum_{i=1}^n w_i x_i (x_i)^\top }_p^p \norm{v}^{2p}\,.
  \]
  This finishes the proof.
\end{proof}

\paragraph{Certificates for sparse singular values.}

We can now derive our results for certifying upper bounds on sparse singular values of a matrix.
To this end, consider the following constraint system $\cA$ in vector-valued variables $u,w$ (in dimension $n$) and vector-valued variable $v$ (in dimension $d$).
Let $\eta$ be such that $\eta n$ is an integer.
Note that $\cA$ encodes that $u$ is an $\eta n$-sparse $n$-dimensional unit vector and $v$ is a (dense) $d$-dimensional unit vector.
\[
  \cA \coloneqq \Set{w_i^2 = w_i \,, \sum_{i=1}^n w_i \leq \eta n \,, w_iu_i = u_i \,, \norm{u}^2 = 1, \norm{v}^2 = 1} \,.
\]
\begin{theorem}[Formal version of \Cref{thm:main-intro}]
  \label{thm:sparse_sing_val_full}
  Let $p,n,d \in \N$ be such that $p \geq 2$ and $d \leq n$.
  Let $M \in \R^{n \times d}$ be a matrix with \iid entries drawn from $\cN(0,1)$ and let $\cA$ be as above.
  Then for $\sosBound$ defined in \Cref{eq:sosBound},
  \begin{align*}
    \E_{M} \inf \left \{  B \in \R \text{ s.t. } \cA \proves_{O(p)}^{w,u,v} \tfrac 1 {\sqrt{n}} u^\top M v\leq B \right \} \leq \sqrt{\sosBound}
  \end{align*}
  Moreover, with probability at least $1-\delta$, the following sum-of-squares proof exists:
  \begin{align*}
    \cA \proves_{O(p)}^{w_i,u,v} \tfrac 1 {\sqrt{n}} u^\top M v \leq \sqrt{\sosBound}  +  \sqrt{O(1)\frac{\log(1/\delta)}{n}}\,.
  \end{align*}
\end{theorem}
\begin{proof}
  Since $\Set{X^2 \leq C^2} \proves_2^X X \leq C$, it is enough to show that (after taking expectations over $M$)
  \begin{align*}
    \cA \proves_{O(p)}^{w_i,u,v} \tfrac 1 {n} \Paren{u^\top M v}^2 \leq \sosBound\,.
  \end{align*}
  Let $M_i \in \R^d$ be the rows of $M$.
  By the constraint that $w_i u_i = u_i$ and SoS Cauchy-Schwarz (cf.\ \cref{fact:sos_cs}), it holds that
  \begin{align*}
    \cA \proves_{O(1)}^{w_i, u, v} \tfrac 1 n \Paren{u^\top M v}^2 &= \tfrac 1 n \Paren{\sum_{i=1}^n u_i \iprod{M_i, v}}^2 = \tfrac 1 n \Paren{\sum_{i=1}^n u_i \cdot w_i \iprod{M_i, v}}^2 \\
    &\leq \norm{u}^2 \Paren{\tfrac 1 n \sum_{i=1}^n w_i \iprod{M_i, v}^2} = \tfrac 1 n \sum_{i=1}^n w_i \iprod{M_i, v}^2 \,.
  \end{align*}
    The result then follows from~\cref{thm:resilience_Gauss_full}.

\end{proof}

\subsection{Sum-of-Squares Certificate for Schatten-\texorpdfstring{$p$}{p} Norm}
\label{sec:sos_certificate_schatten_p}

In this section, we will show the following theorem.
\restatetheorem{thm:sos-schatten-p}
In particular, since $n \leq d^2$, when $n = (p \cdot \log d)^{O(1)} \cdot\max\Set{\tfrac d {\alpha^2},\tfrac{\eta^2 d^2}{\alpha^4}}$ the above is at most $d^{3/2}n^{1/2} \alpha^p$.
The condition that $d \leq n$ could potentially be avoided at the cost of having another term in the maximum.
The condition that $n \leq d^2$ seems inherent to our techniques, and we refer the reader to the later subsections for further intuition.

Our strategy for showing~\cref{thm:sos-schatten-p} was outlined in~\cref{sec:techniques}. Specifically, 
we break the polynomial $\| \sum_{i \leq n} w_i x_i x_i^\top \|_p^p$ into several parts by expanding the $p$-th power and grouping terms based on equalities among the index patterns in the resulting sum.
We then construct SoS upper bounds for each part.
In this subsection (\cref{sec:sos_certificate_schatten_p}), we will formally define the parts, state the upper bounds that we will prove on them (cf.~\cref{lem:sos-graph-poly}), and see how these imply~\cref{thm:sos-schatten-p}.
The upper bounds are then proved using the Efron Stein decomposition and graph matrices in~\cref{sec:high_level_graph_poly,sec:induction_subset_terms,sec:bound_equals_term,sec:graph_matrix_combo}.

\paragraph{Graph polynomials and vertex merges.}

We use graphs to describe the individual parts.
In particular, we will eventually express the Schatten-$p$ norm as a sum of what we call \emph{graph polynomials} defined as follows:
\begin{definition}[Graph polynomial $P_G$]
    Let $G = (V,E)$ be an undirected graph.
    We define a polynomial associated to $G$ as follows:
    \[
    P_G(w,x) = \sum_{\substack{f \, : \, V \rightarrow [n] \\ \text{ injective}}} \prod_{v \in V} w_{f(v)} \cdot \prod_{\{v,w\} \in E} \iprod{x_{f(v)}, x_{f(w)}} \, .
    \]
\end{definition}

We remark that the graph $G$ may contain self-loops and the edges may have multiplicities.
In particular, when $G$ consists of a single self-loop (of multiplicity 1), it holds that $P_G(w,x) = \sum_{i=1}^n w_i \norm{x_i}^2$.
We use the convention that a loop contributes 2 to the degree of a vertex, i.e., each ``endpoint'' contributes 1.
Most often we will not write the arguments ($x$ and $w$) of the polynomial.
Further, we also allow $G$ to be the empty graph, in which case we use the convention that $P_G(w,x) = 1$.

The following types of graphs with arise in our decomposition into graph polynomials.
\begin{definition}[Vertex merge]
  Let $G = (V,E)$ be a graph.
  We say a (multi-)graph $H$ is a \emph{vertex merge} of $G$ if $H$ can be obtained by sequentially identifying pairs of vertices of $G$ while maintaining the set of edges.
  Note that $H$ may be a multigraph, and it may contain self-loops.
\end{definition}
\begin{definition}[Graphs obtained by vertex merging (disjoint) cycles]
  For a $p\in \N$ and $r \in \N$, denote by $\Gcycmerge{p}$ and (respectively, $\Gcycmergedisjoint{p,r}$) the class of multigraphs obtained by starting with a $p$-cycle (respectively, a disjoint union of at most $r$ cycles whose lengths sum to $p$) and performing an arbitrary number of vertex merges.
\end{definition}
Clearly, $\Gcycmerge{p} = \Gcycmergedisjoint{p,1}$.
For $p = 1$, we define the $1$-cycle to be a single vertex that has a self-loop.
See \Cref{fig:graph_polys} for an illustration of the vertex merging process.
The graph polynomial corresponding to the four-cycle (\Cref{fig:graph_poly_1}) is
\[
  P_{C_4} = \sum_{\substack{i_1,i_2,i_3,i_4 \in [n]\,, \\ \text{all distinct}}} w_{i_1}w_{i_2}w_{i_3}w_{i_4} \iprod{x_{i_1},x_{i_2}} \iprod{x_{i_2},x_{i_3}} \iprod{x_{i_3},x_{i_4}} \iprod{x_{i_4},x_{i_1}} \,,
\]
and the graph polynomial corresponding to the graph $G$ in \cref{fig:graph_poly_2} is 
\[
  P_G = \sum_{\substack{i_1,i_2,i_3 \in [n]\,, \\ \text{all distinct}}} w_{i_1}w_{i_2}w_{i_3} \iprod{x_{i_1},x_{i_2}}^2 \iprod{x_{i_2},x_{i_3}}^2 \,. 
\]

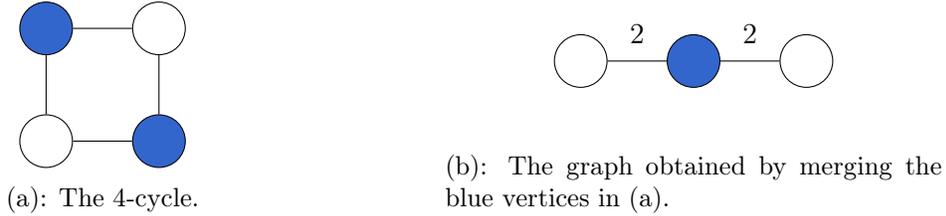
\begin{figure}[ht]
    \centering
    \begin{subfigure}[b]{0.4\textwidth}
        \label{fig:four_cycle}
        \centering
        \begin{tikzpicture}[baseline,every node/.style={circle,draw,minimum size=0.7cm}]
            \node (A) at (0,-1) [fill=navyblue] {};
            \node (B) at (1.5,-1) {};
            \node (C) at (1.5,-2.5) [fill=navyblue] {};
            \node (D) at (0,-2.5) {};

            \draw (A) -- (B);
            \draw (B) -- (C);
            \draw (C) -- (D);
            \draw (D) -- (A);

        \end{tikzpicture}
        \caption{(a): The 4-cycle.}
        \label{fig:graph_poly_1}
    \end{subfigure}
    \hspace{1cm}
    \begin{subfigure}[b]{0.4\textwidth}
    \label{fig:four_cycle_contracted}
        \centering
        \begin{tikzpicture}[baseline,every node/.style={circle,draw,minimum size=0.7cm}]
            \node (A) at (-0.5,1) {};
            \node (B) at (1,1) [fill=navyblue] {};
            \node (C) at (2.5,1) {};

            \draw (A) -- node[above,draw=none] {2} (B);
            \draw (B) -- node[above,draw=none] {2} (C);
        \end{tikzpicture}
        \caption{(b): The graph obtained by merging the blue vertices in (a).}
        \label{fig:graph_poly_2}
    \end{subfigure}
    \caption{Illustration of vertex merging process.}
    \label{fig:graph_polys}
\end{figure}

\paragraph{Breaking the Schatten-\texorpdfstring{$p$}{p} norm into graph polynomials.}

With this in place, we start by expanding out the Schatten-$p$ norm.
In particular, let $\cQ$ be the set of all partitions of $[p]$, and for any $Q \in \cQ$, we say that $(i_1, \ldots, i_p) \in [n]^p$ \emph{respects} $Q$ if $i_j = i_k$ if and only if $j$ and $k$ are in the same part of $Q$.
Using the convention that $i_{p+1} = i_1$ it follows that
\begin{align*}
      \left \| \frac 1 n \sum_{i \leq n} w_i x_i x_i^\top \right \|_p^p = \frac 1 {n^p} \sum_{i_1, \ldots, i_p = 1}^n \prod_{j=1}^p w_{i_j} \cdot \prod_{j=1}^p \iprod{x_{i_j},x_{i_{j+1}}} = \frac 1 {n^p} \sum_{Q \in \cQ} \sum_{\substack{i_1, \ldots, i_p \\ \text{respecting }Q}}^n \prod_{j=1}^p w_{i_j} \cdot \prod_{j=1}^p \iprod{x_{i_j},x_{i_{j+1}}} \,.
\end{align*}
Note that we can represent each $Q$ as a graph whose vertex set are the parts of $Q$ and the edges are given by $(i_j, i_{j+1})$ (abusing notation to denote by $i_j$ not a number in $[n]$, but the part of $Q$ it belongs to).
Note that there is a natural bijection between $\cQ$ and $\Gcycmerge{p}$.
Further, for a given $Q$ and its associated graph $G$, we can express all $(i_1, \ldots, i_p)$ respecting $Q$ as injective functions from the vertex set of $G$ to $[n]$.
Thus, 
  \[
  \frac 1 {n^p} \sum_{i_1, \ldots, i_p = 1}^n \prod_{j=1}^p w_{i_j} \prod_{j=1}^p \iprod{x_{i_j},x_{i_{j+1}}}  = \frac 1 {n^p} \sum_{G \in \Gcycmerge{p}} P_G(w,x)\,.
  \]

Constructing SoS upper bounds on the individual graph polynomials is the heart of our argument.
In particular, we will prove the following lemma in the subsequent sections by combining the Efron-Stein decomposition of our graph polynomials, bounds on the operator norms of certain graph matrices, and an induction on $p$.
\begin{lemma}
\torestate{
\label{lem:sos-graph-poly}
Let $n,d,p \in \N$ be such that $p \geq 2, d \leq n \leq d^2$ and let $\eta \in [0,1]$.
Let $x_1,\ldots,x_n \sim \cN(0,I_d)$ and let $w_1,\ldots,w_n$ be indeterminates.
Let $G \in \Gcycmergedisjoint{p,r}$, then
\begin{align*}
&\E_{x_1,\ldots,x_n} \inf \left \{  B \in \R \text{ s.t. } \Set{\{w_i^2 = w_i\}_{i \in [n]}, \sum_{i \leq n} w_i \leq \eta n} \proves_{O(p)}^{w}
P_G(w,x) \leq B \right \} \\
&\leq\qquad (p \cdot \log n)^{O(p)} \cdot dn^{r/4} \cdot (dn)^{p/2} \cdot (\eta^2 n)^{\card{V} / 4} \cdot d^{r/2}  \, .
\end{align*}
}
\end{lemma}
We defer its proof to later subsections, but at this point offer the following intuition for the dependence of the bound on the various parameters.
In particular, consider the graph polynomial corresponding to $C_p$, the cycle on $p$ vertices.
In this case, $r = 1$ and $\card{V} = p$ and (using the convention that $i_{p+1} = i_1$)
\[
  P_{C_p} = \sum_{\substack{i_1, \ldots, i_p \in [n]\,, \\ \text{all distinct}}} \prod_{j=1}^p w_{i_j} \cdot \prod_{j=1}^p \iprod{x_{i_j}, x_{i_{j+1}}} \,.
\]
Each summand, ignoring the $w_{i_j}$ variables, is roughly equal to $\pm d^{p/2}$.
Thus, heuristically we expect this polynomial of the order $(\eta n d)^{p/2}$ for a concrete instantiation of the $w_{i_j}$ variables.
Note that this is smaller by roughly a factor of $n^{p/4}$ than the bound claimed above.
Yet, some extra slack factor is indeed necessary as illustrated by the following example:
Consider the case when $G$ corresponds to a union of $p$ self-loops.
In this case $\card{V} = p$ and $r = p$.
It holds that
\[
  P_G = \sum_{\substack{i_1, \ldots, i_p \in [n]\,, \\ \text{all distinct}}} \prod_{j=1}^p w_{i_j} \cdot \prod_{j=1}^p \norm{x_{i_j}}^2 \,,
\]
which is of the order $(\eta n d)^p$ and is only slightly smaller than the bound given by \cref{lem:sos-graph-poly}, which would be $\poly(n,d) \cdot (\sqrt{\eta} n d)^p$.
Lastly, consider the graph $G$ obtained by starting with $C_p$ and merging all the vertices.
We obtain a single self-loop of multiplicity $p$.
In this case $\card{V} = r = 1$.
And $P_G = \sum_{i=1}^n w_i \norm{x_i}^{2p}$ which is on the order of $(\eta n) d^p$.
This is only a bit smaller than the bound $\poly(n,d) \cdot (dn)^{p/2}$ given by \cref{lem:sos-graph-poly}.

With~\cref{lem:sos-graph-poly} in hand, we can readily proof~\cref{thm:sos-schatten-p}.
\begin{proof}[Proof of~\cref{thm:sos-schatten-p}]
  By the preceding discussion, it is enough to bound
  \[
    \frac 1 {n^p} \sum_{G \in \Gcycmerge{p}} P_G(w,x) = \frac 1 {n^p} \sum_{t=1}^p \sum_{\substack{G \in \Gcycmerge{p}\,,\\ \card{V(G)} = t}} P_G(w,x)\,.
  \]
  Using that there are at most $p^{O(p)}$ graphs in $\Gcycmerge{p}$, it follows from~\cref{lem:sos-graph-poly} that in expectation over $x_1, \ldots, x_n$ there is an SoS proof that the above is at most
  \begin{align*}
    &\frac{p^{O(p)}}{n^p} \sum_{t = 1}^p (p \cdot \log n)^{O(p)} \cdot dn^{1/4} \cdot (dn)^{p/2} \cdot (\eta^2 n)^{t / 4} \cdot d^{1/2} \\
    &\qquad \leq (p \cdot \log n)^{O(p)} \cdot d^{3/2}n^{1/4} \cdot \Paren{\frac d n}^{p/2} \cdot \sum_{t = 1}^p (\eta^2 n)^{t / 4} \,.
  \end{align*}
  We make a case distinction.
  Suppose that $\eta^2 n \geq 1$, then since $t \leq p$ it holds that $\sum_{t = 1}^p (\eta^2 n)^{t / 4} \leq p \cdot (\eta \sqrt{n})^{p/2}$.
  If $\eta^2 n \leq 1$, then the same sum is at most $p \cdot \sqrt{\eta}n^{1/4}$.
  Thus, the total expression is at most
  \[
    (p \cdot \log n)^{O(p)} \cdot d^{3/2}n^{1/2} \cdot \Paren{\frac d n}^{p/2} \cdot \max\Set{1, (\eta \sqrt{n})^{p/2}} = (p \cdot \log n)^{O(p)} \cdot d^{3/2}n^{1/2} \cdot \max\Set{\frac d n, \frac{\eta d}{\sqrt{n}}}^{p/2} \,,
  \]
  which implies the claim since $n \leq d^2$ and hence $d^{3/2} n^{1/2} \leq d^{5/2}$.
  
\end{proof}

In the remaining subsections, we shall give the proof of \Cref{lem:sos-graph-poly}.

\subsection{Sum-of-Squares Bound on Graph Polynomials}
\label{sec:high_level_graph_poly}

In this section, we will start the proof~\cref{lem:sos-graph-poly} to give SoS upper bounds on graph polynomials.
We will use the Efron Stein decomposition (defined below) to decompose them into several parts.
For one of these parts, we will be able to construct a good-enough spectral certificate,  and for the remaining parts, we will show their bounds are a function of smaller graph polynomials, which we will show to be small via induction.
We will describe the decomposition and induction in this section and defer the proof of the spectral upper bound to~\cref{sec:bound_equals_term,sec:graph_matrix_combo} and the proof of some auxiliary properties used in the induction to~\cref{sec:induction_subset_terms}.

We start by recalling the Efron Stein decomposition and show how it induces the decomposition of graph polynomials.

\paragraph{Efron-Stein decomposition and its interaction with graph polynomials.}

We first recall the Efron-Stein decomposition of a function $f$ on a product space.

\begin{definition}[Efron-Stein decomposition]
Let $X_1,\ldots,X_n$ be independent random variables assuming values in a set $\Omega$ and suppose $f \, : \, \Omega^n \rightarrow \R$ is a real-valued function.
For $S \subseteq [n]$, let $f^{\subseteq S} \, : \, \Omega^{|S|} \rightarrow \R$ be given by
\[
f^{\subseteq S}(y) = \E_{X_{\overline{S}}} f(X_{\overline{S}},y) \, .
\]
We also define
\[
f^{=S}(y) = \sum_{T \subseteq S} (-1)^{|S \setminus T|} f^{\subseteq T}(y) \, .
\]
\end{definition}
For a thorough treatment of the Efron-Stein decomposition, see, e.g.,~\cite{Odo14}.
We will use the following facts: The first one just uses the definition of $f^{=}$ and that $f^{\sse [n]} = f$.
\begin{fact}
  $f = f^{=[n]} - \sum_{T \subsetneq [n]} (-1)^{\card{[n] \setminus T}} f^{\subseteq T}$ \,.
\end{fact}
\begin{fact}
  \label{fact:ES_remove_strict_subset_func}
  Suppose $f \colon \Omega^n \rightarrow \R$ is a function that does not depend on $X_i$.
  Then $f^{=[n]} = 0$.
\end{fact}
\begin{proof}
  Pair up the sets $T \sse [n]$ in such a way that one set in the pair does not contain $i$ and the other set is the first union $\Set{i}$.
  Note that this partitions all sets.
  Consider one such pair $(T, T\cup \Set{i})$.
  Since $f$ does not depend on $X_i$, it holds that $f^{\sse T} = f^{\sse T \cup \Set{i}}$.
  The fact now follows by noting that $f^{\sse T}$ and $f^{\sse T\cup \Set{i}}$ appear with opposite signs in the definition of $f^{=[n]}$.
\end{proof}

Fix $G \in \Gcycmergedisjoint{p,r}$.
We use the Efron-Stein decomposition to define a decomposition of the graph polynomial $P_G$. 
Note that this is not quite the same as the vanilla Efron-Stein decomposition of $P_G$ itself -- if $G = (V,E)$, our decomposition has $2^{|V|}$ pieces, even though $P_G$ depends on the $n$ random variables $x_1,\ldots,x_n$ so its Efron-Stein decomposition would have $2^n$ pieces.
In particular, for every $S \subseteq V$ that is non-empty, the function $P_G^{=S}$ below \emph{does} depend on all of $x_1, \ldots, x_n$.
Note that in this sense the definition below abuses the notation of the Efron Stein decomposition.
As we will never make reference to any ``standard'' Efron Stein decomposition, this should not lead to confusion.

\begin{definition}
  Let $G = (V,E)$ be a graph and $P_G(x,w)$ be the associated graph polynomial.
  Suppose $x_1,\ldots,x_n \sim \cN(0,I_d)$.
  For $S \subseteq V$, we define the polynomials
  \[
  P_G^{\subseteq S}(w,x) = \sum_{\substack{f \, : \, V \rightarrow [n] \\ \text{ injective}}} \prod_{v \in V} w_{f(v)} \cdot \Paren{\prod_{\{v,w\} \in E} \iprod{x_{f(v)}, x_{f(w)}}}^{\subseteq f(S)}
  \]
  and
  \[
    P_G^{=S}(w,x) = \sum_{\substack{f \, : \, V \rightarrow [n] \\ \text{ injective}}} \prod_{v \in V} w_{f(v)} \cdot \Paren{\prod_{\{v,w\} \in E} \iprod{x_{f(v)}, x_{f(w)}}}^{=f(S)} \, .
  \]
\end{definition}
Clearly, $P_G = \sum_{S \subseteq V} P_G^{=S}$ and $P_G = P_G^{= V} - \sum_{S \sse V} (-1)^{\card{V \setminus S}} P_G^{\subsetneq S}$.
Note that all parts of the decomposition are polynomials in all of the $w_i$'s and $x_i$'s.

\paragraph{Example of Efron-Stein Decomposition for simple graph polynomials.}

Performing Efron-Stein decomposition of graph polynomials naturally gives rise to other graph polynomials, corresponding to graphs on fewer vertices.
For intuition, let us consider again the graph polynomial corresponding to $C_4$:
\[
  P_{C_4} = \sum_{\substack{i_1,i_2,i_3,i_4 \in [n]\,, \\ \text{all distinct}}} w_{i_1}w_{i_2}w_{i_3}w_{i_4} \iprod{x_{i_1},x_{i_2}} \iprod{x_{i_2},x_{i_3}} \iprod{x_{i_3},x_{i_4}} \iprod{x_{i_4},x_{i_1}} \,.
\]
Let $S$ be the set of two non-adjacent vertices.
Then it, e.g., holds that our constraints (on $w$) imply that 
\begin{align*}
  P_{C_4}^{\sse S} &= \sum_{\substack{i_1,i_2,i_3,i_4 \in [n]\,, \\ \text{all distinct}}} w_{i_1}w_{i_2}w_{i_3}w_{i_4} \E_{x_{i_1}, x_{i_3}} \iprod{x_{i_1},x_{i_2}} \iprod{x_{i_2},x_{i_3}} \iprod{x_{i_3},x_{i_4}} \iprod{x_{i_4},x_{i_1}} \\
  &= \sum_{\substack{i_1,i_2,i_3,i_4 \in [n]\,, \\ \text{all distinct}}} w_{i_1}w_{i_2}w_{i_3}w_{i_4} \iprod{x_{i_2},x_{i_4}}^2 = (\eta n)^2 P_{G'} \,,
\end{align*}
where $G'$ is the graph consisting of two nodes connected by a single edge of multiplicity 2.

\paragraph{Proof of~\cref{lem:sos-graph-poly}}

We will use an induction on the number of vertices in the graph $G$ to show~\cref{lem:sos-graph-poly} (restated below).
\restatelemma{lem:sos-graph-poly}
In particular, in the decomposition $P_G = P_G^{= V} - \sum_{S \sse V} (-1)^{\card{V \setminus S}} P_G^{\subsetneq S}$, $P_G^{= V}$ is the part that we will bound using a spectral certificate in~\cref{sec:bound_equals_term,sec:graph_matrix_combo}, and for each $P_G^{\subsetneq S}$ we will show that it corresponds to a function of graph polynomials on graphs on fewer vertices and use our inductive hypothesis.
We defer the proof of some auxiliary lemmas used in the induction to~\cref{sec:induction_subset_terms}.

\begin{proof}[Proof of~\cref{lem:sos-graph-poly}]
For technical reasons, we will show instead that
\begin{align*}
  &\E_{x_1,\ldots,x_n} \inf \left \{  B \in \R \text{ s.t. } \Set{\{w_i^2 = w_i\}_{i \in [n]}, \sum_{i \leq n} w_i \leq \eta n} \proves_{O(p)}^{w}
  P_G(w,x)^2 \leq B^2 \right \} \\
  &\leq (p \cdot \log n)^{O(p)} \cdot dn^{r/4} \cdot (dn)^{p/2} \cdot (\eta^2 n)^{\card{V} / 4} \cdot d^{r/2}  \, .
\end{align*}
Note that the second statement implies the first using that $\Set{X^2 \leq C^2} \proves_2 X \leq C$ (cf.~\cref{fact:sos_square_root}).
We use an induction on $p$.

\paragraph{Base case.}

For $p = 1$, we start with a 1-cycle, i.e., a single vertex with a self-loop.
The bound we want to show is $(\log n)^{O(1)} \cdot dn^{1/4} \cdot (dn)^{1/2} \cdot (\eta^2 n)^{1/ 4} d^{1/2} = (\log n)^{O(1)} \sqrt{\eta} n d^2$.
Clearly, there are no vertex merges that we can perform.
Thus, $\Gcycmergedisjoint{1}$ only contains the 1-cycle.
In this case, $P_G$ takes the form $P_G(w,x) = \sum_{i=1}^n w_i \norm{x_i}^2$.
Since in this case our constraints imply that $P_G(w,x) \geq 0$, it suffices to certify an upper bound on $P_G$ (as this will imply an upper bound on $P_G^2$).
Note that
\[
  \Set{w_i^2 = w_i \,, \sum_{i=1}^n w_i \leq \eta n} \proves_{O(1)}^{w} P_G(w,x) = \sum_{i=1}^n w_i \norm{x_i}^2 \leq \eta n \cdot \max_{i \in [n]} \;\norm{x_i}^2 \,.
\]
Thus,
\begin{align*}
  \E_{x_1,\ldots,x_n} &\inf \left \{  B \in \R \text{ s.t. } \Set{\{w_i^2 = w_i\}_{i \in [n]}, \sum_{i \leq n} w_i \leq \eta n} \proves_{O(1)}^{w_i}
P_G(w,x) \leq B \right \} \\
&\leq \eta n \cdot \E_{x_1, \ldots, x_n} \max_{i \in [n]}\; \norm{x_i}^2\,.
\end{align*}
Let $Z = \max_{i \in [n]}\; \norm{x_i}$.
Note that $Z$ is a 1-Lipschitz function of the standard Gaussian vector $(x_1, \ldots, x_n) \sim N(0,I_{d\cdot n})$ and hence, $Z - \E Z$ is 1-subgaussian (see, e.g.,~\cite[Theorem 2.26]{Wainwright19}).
Since further, all $\norm{x_i}$ are 1-subgaussian with mean $\E \norm{x_1}$, it holds that $\E Z \leq \E \norm{x_1} + \sqrt{2\log n} \leq 5 \sqrt{d}$, as $n \leq d^2$.
And thus,
\[
  \E_{x_1, \ldots, x_n} \max_{i \in [n]}\; \norm{x_i}^2 = \E \Paren{Z - \E Z}^2 + \Paren{\E Z}^2 \leq 100 d \,.
\]

It follows that the whole expression is at most $O(\eta n d) \leq (\log n)^{O(1)} \sqrt{\eta} n d^2$.

\paragraph{Inductive step.}

Let $p \geq 2$ and assume the statement holds for all $G' \in \Gcycmergedisjoint{p',r}$ for all $1 \leq p' < p$ and $r \leq p'$ simultaneously.
Let $G = (V,E)$.
We distinguish the cases $\card{V} = 1$ and $\card{V} \geq 2$.
Note that in the first case, $G$ is a single vertex with $p$ self-loops.
Thus, $P_G = \sum_{i=1}^n w_i \norm{x_i}^{2p}$.
Again, our constraints imply that $P_G \geq 0$, so it is sufficient to certify an upper bound on $P_G$.
Since $\card{V} = 1, r \geq 1$, the bound we aim for is at least $(p \cdot \log n)^{O(p)} \cdot d^{3/2} \cdot \sqrt{n} \cdot\sqrt{\eta} \cdot (dn)^{p/2}$.
Analogously to the base case it follows that
\[
  \Set{w_i^2 = w_i \,, \sum_{i=1}^n w_i \leq \eta n} \proves_{O(1)}^{w} P_G(w,x) = \sum_{i=1}^n w_i \norm{x_i}^{2p} \leq \eta n \cdot \max_{i \in [n]} \;\norm{x_i}^{2p} \,.
\]
It remains to take expectations over $x_1, \ldots, x_n$.

As before, it follows that
\[
  \E \max_{i \in [n]} \,\norm{x_i}^{2p} = \E Z^{2p} \leq 4^p \E(Z - \E Z)^{2p} + 4^p \Paren{\E Z}^{2p} \leq O(d)^p \cdot p^{O(p)} \,.
\]

Thus, the final bound is $p^{O(p)} \cdot \eta n \cdot d^p$.
This compares favorably to the bound we aimed for since $d \leq n \leq d^2$.

\paragraph{$G$ contains at least two vertices.}

Next, we consider the case that $\card{V} \geq 2$.
We use the decomposition
\[
  P_G = P_G^{= V} - \sum_{S \subsetneq V} (-1)^{\card{V \setminus S}} P_G^{\subseteq S} \,.
\]
Suppose our constraints on $w_i$ imply an SoS upper bound of $\tilde{B}$ on each of the summands squared in expectation over $x_1, \ldots, x_n$ (at degree $O(p)$).
Since there are at most $2^p$ summands, using SoS triangle inequality (cf.~\cref{fact:sos_triangle}) this implies that
\[
  \Set{w_i^2 = w_i \,, \sum_{i=1}^n w_i \leq \eta n} \proves_{O(p)}^{w} P_G^2 \leq 4^p \cdot \Paren{P_G^{=V}}^2 + 4^p \cdot \sum_{S \subsetneq V} \Paren{P_G^{\subseteq S}}^2 \leq 8^p \cdot \tilde{B}^2 \,.
\]

We will bound the $P_G^{=V}$ term using a spectral certificate obtained via graph matrices and the remaining terms via our inductive hypothesis.
Indeed, we will show the following bound in~\cref{sec:bound_equals_term,sec:graph_matrix_combo} which shows that the first term is bounded in the way we would like.
Note that the bound obtained via the lemma below is better by a factor of $d^r$ than what we actually need.
\begin{lemma}
  \torestate{
  \label{lem:bound-on-equals-term}
      Let $p \geq 2, d \leq n \leq d^2$, $G = (V,E) \in \Gcycmergedisjoint{p,r}$ such that $\card{V} \geq 2$, and let $P_G(x,w)$ be its associated graph polynomial.
      Let $x_1,\ldots,x_n \sim \cN(0,I_d)$.
      Then,
      \begin{align*}
      &\E_{x_1,\ldots,x_n} \inf \left \{  B \in \R \text{ s.t. } \Set{\{w_i^2 = w_i\}_{i \in [n]}, \sum_{i \leq n} w_i = \eta n }\proves_{O(p)}^{w} \Paren{P_G^{=V}(w,x)}^2 \leq B^2 \right \} \\
      &\leq (p \cdot \log n)^{O(p)} \cdot dn^{r/4} \cdot (dn)^{p/2} \cdot (\eta^2 n)^{\card{V} / 4}\, .
      \end{align*}
  }
\end{lemma}

For our induction, the following lemma is the key ingredient.
We prove it in~\cref{sec:induction_subset_terms}.
We use the following notation.
For a vertex $v$, we use $\deg(v)$ to denote its degree, $\degdagger(v)$ to denote its degree counting only edges that are \emph{not} loops, and $s(v)$ to denote the number of self-loops.
Note that $\deg(v) = \degdagger(v) + 2s(v)$.
\begin{lemma}
  \torestate{
  \label{lem:main_induction}
  Let $p \in \N$ and let $G = (V,E) \in \Gcycmergedisjoint{p,r}$, i.e., a vertex merge of a union of $r$ cycles whose lengths sum to $p$.
  Let $S \subsetneq V$ be a strict subset of the vertices of $G$ and $u$ be a vertex in $V \setminus S$.
  There exists a family $\cF$ of at most $\deg(u)^{O(\deg(u))}$ graphs $G'$, each on vertex set $V' = V \setminus \{u\}$, and each a vertex merge of a union of at most $r +  \tfrac{\degdagger(u)} 2 -1$ cycles with lengths summing to $p - \deg(u) / 2$, such that
  \[
    \Set{w_i^2 = w_i \,, \sum_{i=1}^n w_i = \eta n} \proves_{O(1)}^{w} P_{G}^{\subseteq S} = O(d)^{s(u)} \cdot (\eta n - |V'|) \cdot \sum_{G' \in \cF} P_{G'}^{\subseteq S}\,.
  \]
  Further, for each such graph $G'$, each vertex $v$ has the same degree and ``dagger-degree'' as in $G$.
  }
\end{lemma}

Let $S \subsetneq V$.
Let $u_1, \ldots, u_t$ be the vertices in $V \setminus S$.
We apply~\cref{lem:main_induction} iteratively for $u_1, \ldots, u_t$.
Let $r' = r -t + \tfrac 1 2 \sum_{i=1}^t \degdagger(u_i)$, where the degrees are taken to be in the original graph $G$ -- note that this does not matter, since~\cref*{lem:main_induction} guarantees that the degrees do not change except for the vertex we contract.
Let $p' = p - \tfrac 1 2 \sum_{i=1}^t \deg(u_i)$ and $s = \sum_{i=1}^t s(u_i)$.
We obtain a family $\cH$ of at most $\prod_{i=1}^t \deg(u_i)^{\deg(u_i)} \leq p^{O(\sum_{i=1}^t \deg(u_i))} = p^{O(p)}$ graphs, each on vertex set $S$, and each a vertex merge of a union of at most $r'$ cycles with lengths summing to $p'$, i.e., in $\Gcycmergedisjoint{p',r'}$, such that\footnote{The increase in degree of the SoS proof to at most $O(t)$ follows by standard composition of SoS proofs.
Note that $t \leq p$.}
\begin{equation*}
  \Set{w_i^2 = w_i \,, \sum_{i=1}^n w_i = \eta n} \proves_{O(t)}^{w_i} P_G^{\sse S} = O(d)^{s} \cdot \prod_{i=1}^t(\eta n - \card{V} - i)\cdot \sum_{G' \in \cH} P_{G'}^{\sse S} \,.
\end{equation*}
Since $\mathcal{H}$ contains at most $p^{O(p)}$ terms, it again follows by SoS triangle inequality that
\begin{equation}
  \Set{w_i^2 = w_i \,, \sum_{i=1}^n w_i = \eta n} \proves_{O(t)}^{w_i} \Paren{P_G^{\sse S}}^2 \leq p^{O(p)} \cdot O(d)^{2s} \cdot \prod_{i=1}^t(\eta n - \card{V} - i)^2\cdot \sum_{G' \in \cH} \Paren{P_{G'}^{\sse S}}^2 \,. \label{eq:induction}
\end{equation}
We can apply our inductive hypothesis to each term in the sum.
In particular, we can certify an upper bound of (note that when $S = \emptyset$, $G'$ will be the empty graph and the upper bound is trivially 1)
\[
  (p \cdot \log n)^{O(p)} \cdot dn^{r'/2} \cdot (dn)^{p'} \cdot (\eta^2 n)^{\card{S} / 2} \cdot d^{r'} = \frac{(p \cdot \log n)^{O(p)} \cdot dn^{r/2} \cdot (dn)^{p} \cdot (\eta^2 n)^{\card{V} / 2} \cdot d^{r} }{n^{(r-r')/2} \cdot (dn)^{(p-p')} \cdot (\eta^2 n)^{(\card{V} - \card{S})/2} \cdot d^{(r - r')}} \,.
\]
Recall that the numerator of the right-hand side is the bound we aim to show.
Plugging this back into~\cref{eq:induction}, the bound we can certify on $(P_G^{\sse S})^2$ is our target bound, multiplied by (this also uses that $\cH$ contains at most $p^{O(p)}$ graphs, $O(1)^s = p^{O(p)}$, $\prod_{i=1}^t(\eta n - \card{V} - i) \leq (\eta n)^t$, and $r' - r \geq - t$)
\begin{align*}
  &\frac{d^{2s} \cdot (\eta n)^{2t}}{n^{(r-r')/2} \cdot (dn)^{(p-p')} \cdot (\eta^2 n)^{(\card{V} - \card{S})/2} d^{(r - r')}} \\
  &= d^{2s + (p'-p) + (r'-r)} \times n^{3t/2 + (r'-r)/2 + (p'-p)} \times \eta^{2t} \\
  &\leq d^{2s + (p'-p) + (r'-r)} \times n^{3t/2 + (r'-r)/2 + (p'-p)}\,.
\end{align*}
We start with some simplifying observations about how the changes in $p$ and $r$ balance out, using that $\deg(u_i) = 2 s(u_i) + \degdagger(u_i)$.
\begin{align*}
  (p' - p) + (r' - r) &= -\tfrac 1 2 \sum_{i=1}^t \deg(u_i) - t + \tfrac 1 2 \sum_{i=1}^t \degdagger(u_i) = -s - t \,.
\end{align*}
Thus, the exponent of $d$ is equal to $2s + \Brac{(p'-p) + (r' -r)} = s-t$.
And the exponent of $n$ is equal to 
\begin{align*}
  \tfrac 1 2 \Brac{3t + 2(p'-p) + (r' -r)} &= \tfrac 1 2 \Brac{2t - s +  (p'-p)} = \tfrac 1 2 \Brac{2t - 2s - \frac 1 2 \sum_{i=1}^t \degdagger(u_i) }\notag \\
  &= \Paren{t-s} - \frac{1}{4} \sum_{i=1}^t \degdagger(u_i)\,.
\end{align*}
We make a case distinction based on whether $s \leq t$ or $s \geq t$.
We start with the case when $s \geq t$.
In this case, we upper bound the exponent of $n$ by $t-s$.
Then the expression involving $n$ and $d$ is at most $(d/n)^{(s-t)} \leq 1$ since $n \geq d$ and $s \geq t$.
If $s \leq t$, note that there are at least $t-s$ vertices such that $\degdagger(u_i) \geq 2$.
Thus, we can upper bound the exponent of $n$ by $\tfrac{t-s}{2}$.
Thus, the whole expression involving $n$ and $d$ is at most $(n/d^2)^{(t-s)/2} \leq 1$ since $n \leq d^2$.

\end{proof}

\subsection{Leading up to a Bound on \texorpdfstring{$P_G^{=V}$}{Equals Term} via Graph Matrices}
\label{sec:bound_equals_term}

We prove~\cref{lem:bound-on-equals-term} (restated below) in this subsection and the next.
\restatelemma{lem:bound-on-equals-term}
Its proof is based on showing an upper bound on the spectral norm of an appropriate matrix representation for the polynomial $P_G^{=V}$.
The proof is split into two parts.
We will first define the matrix representation and show that it can be expressed as a sum of not too many graph matrices (defined below).
Each graph matrix in the sum can (a) be obtained in a specific way from the graph $G$ underlying the polynomial $P_G$, we call such graph matrices \emph{admissible} (see below for a precise definition), or (b) we can roughly speaking relate its spectral norm to that of an admissible graph matrix.
The crucial property of these admissible graph matrices is that they satisfy appropriate norm bounds (cf.~\cref{lem:norm_bound_admissible_g_m}).\footnote{Why this is true will unfortunately only become clear in~\cref{sec:graph_matrix_combo}. We nevertheless choose this way of presenting our proof since the arguments in showing that the decomposition exists and in showing the norm bounds are very different, and can be read largely independently.}
It follows that each term in the sum has not too large spectral norm.
We encapsulate the facts we need about the graph matrix decomposition in~\cref{lem:graph_matrix_decomposition}.
Assuming both of these lemmas, we first finish the proof of~\cref{lem:bound-on-equals-term}.
We will then prove \cref{lem:graph_matrix_decomposition} in the remainder of this subsection.
We prove \cref{lem:norm_bound_admissible_g_m} in \cref{sec:graph_matrix_combo}.

\paragraph{Matrix representations of graph polynomials.}

We represent graph polynomials as quadratic forms as follows.
For sets $L, R$ and tuples $\alpha \in [n]^{\card{L}},\beta \in [n]^{\card{R}}$, we say that $(\alpha, \beta)$ is a \emph{valid labelling} if all entries of $\alpha$ and $\beta$ (combined) are distinct.
Further, for $(\alpha,\beta)$ that is a valid labelling, we define $f_{\alpha,\beta} \colon L \cup R \rightarrow [n]$ as $f(v) = \alpha_v$, if $v \in L$ and $f(v) = \beta_v$ if $v \in R$.
Recall that for a graph $G = (V,E)$, the corresponding graph polynomial $P_G$ was defined as
\[
    P_G(w,x) = \sum_{\substack{f \, : \, V \rightarrow [n] \\ \text{ injective}}} \prod_{v \in V} w_{f(v)} \cdot \prod_{\{v,w\} \in E} \iprod{x_{f(v)}, x_{f(w)}} \, .
\]
We use the following definition.
\begin{definition}[Matrix representation of graph polynomials]
  Let $G = (V,E)$ be a graph and let $L \cup R = V$ be a partition of the vertices of $G$ such that both $L$ and $R$ are non-empty, i.e., in particular $\card{V} \geq 2$.
  Consider the matrix $M_G \in \R^{n^{\card{L}} \times n^{\card{R}}}$ indexed by tuples $\alpha \in [n]^{\card{L}},\beta \in [n]^{\card{R}}$ that is defined as
  \[
    (M_G)_{(\alpha,\beta)} =
    \begin{dcases}
      \hfill \prod_{\{v,w\} \in E} \iprod{x_{f_{\alpha,\beta}(v)}, x_{f_{\alpha,\beta}(w)}} \,, &\quad \text{if $(\alpha,\beta)$ is a valid labelling,} \\
      \hfill 0 \,, &\quad \text{otherwise.}
    \end{dcases}
  \]
\end{definition}
Note that for a graph polynomial $G$ we have $P_G(w,x) = \Paren{w^{\otimes \card{L}}}^\top M_G w^{\otimes \card{R}}$.
Further, if $M_G$ is a matrix representation of $P_G$ (for some choice of $L$ and $R$), let $M_G^{=S}$ be the matrix with entries $(M_G)_{\alpha,\beta}^{=f_{\alpha,\beta}(S)}$.
Then $P_G(w,x)^{=S} = \Paren{w^{\otimes \card{L}}}^\top M_G^{=S}w^{\otimes \card{R}}$.

Note that $M_G$ depends on the choice of $L$ and $R$.
We suppress this dependence in the notation as the choice will be clear from context.

\paragraph{Reducing~\cref{lem:bound-on-equals-term} to a graph matrix decomposition and norm bounds.}

Let $M_G^{=V}$ be a matrix representation of $P_G^{=V}$.
We describe below how to choose $L$ and $R$, but recall that $\card{L} + \card{R} = \card{V}$.
We will use the following SoS upper bound on $(P_G^{=V}(w,x))^2$, where the first inequality uses the fact that $\proves_4^{x,y} (x^\top A y)^2 \leq \norm{A}^2 \norm{x}^2 \norm{y}^2$ (cf.~\cref{fact:spectral_upper_bound}):
  \begin{align*}
    \proves_{O(p)}^{w} (P_G^{=V})^2 = \Paren{(w^{\otimes\card{L}})^\top M_G^{=V} w^{\otimes\card{R}}}^2 &\leq \Norm{M_G^{=V}}^2 \norm{w^{\otimes \card{L}}}^2 \norm{w^{\otimes \card{L}}}^2 = \Norm{M_G^{=V}}^2 \norm{w}^{2(\card{L} + \card{R})} \\
    &= \Norm{M_G^{=V}}^2 \cdot (\eta n)^{\card{V}} \,.
  \end{align*}
  Taking expectations over $x_1, \ldots, x_n$, this implies that
  \begin{align*}
    &\E_{x_1,\ldots,x_n} \inf \left \{  B \in \R \text{ s.t. } \Set{\{w_i^2 = w_i\}_{i \in [n]}, \sum_{i \leq n} w_i = \eta n }\proves_{O(p)}^{w} \Paren{P_G^{=V}(w,x)}^2 \leq B^2 \right \} \\
    &\leq (\eta n)^{\card{V}/2} \cdot \E_{x_1, \ldots, x_n} \Norm{M_G^{=V}} \, .
    \end{align*}
  Thus, it is enough to show that
  \[
    \E_{x_1, \ldots, x_n} \Norm{M_G^{=V}} \leq (p \cdot \log n)^{O(p)} \cdot d\cdot d^{p/2} \cdot n^{r/4} \cdot n^{\tfrac{2p - \card{V}} 4}
  \]
  In the remainder of this section, we drop the superscript $V$ and write $M_G^{=V} = M_G^=$.

  This directly follows by triangle inequality from the following lemma.
  We will show most of its proof in this subsection and the remaining part (containing everything related to graph matrix norm bounds) in~\cref{sec:graph_matrix_combo}.
  \begin{lemma}[Graph Matrix Decomposition]
    \label{lem:graph_matrix_decomposition}
      There exists a choice of $L$ and $R$ such that the following holds:
      There exists a sequence $M_1, \ldots, M_N \in \R^{n^{\card{L}} \times n^{\card{R}}}$ of graph matrices (cf.~\cref{eq:graph_matrix_def}) and coefficients $c_1, \ldots, c_N \in \R$ such that
      \begin{enumerate}
          \item $M_G^= =\sum_{i=1}^N c_i M_i$.
          \item $N \leq p^{O(p)}$ and $\forall i \in [N] \colon \abs{c_i} \leq p^{O(p)}$.
          \item \label{test} For all $M_i$ it holds that
          \[
            \E_{x_1, \ldots, x_n} \norm{M_i} \leq (p \cdot \log n)^{O(p)} \cdot d \cdot d^{p/2} \cdot n^{r/4} \cdot n^{\tfrac{2p - (\card{V})} 4} \,.
          \]
      \end{enumerate}
    \end{lemma}

\subsubsection{Graph Matrices and Admissible Graph Matrices}
\label{sec:graph_matrices_and_admissible_graph_matrices}
  
In order to proceed, we introduce graph matrices and what we call \emph{admissible} graph matrices.
The former is a by now standard notion (cf.~\cite{AhnMP16}), the latter is a structured subset of them useful for our analysis, see \cref{fig:admissible_gm} for an example.
For the former, we do not introduce the full framework but only what is necessary for us.

\paragraph{Graph matrices.}

Let $H = (V(H), E(H))$ be a graph of two types of vertices, ``circle vertices'' and ``diamond vertices''.
We assign a weight to each vertex.
All circle vertices are assigned weight 1 and all diamond vertices are assigned weight $\log_n (d)$.
For a set of vertices $S$, we denote by $w(S)$ the sum of weights of all vertices in $S$.
Every edge $e$ has one circle and one diamond endpoint, denoted by $e_\circ$ and $e_\diamond$.
Denote by $V_\diamond$ the set of all diamond vertices.
Let $L,R$ be a partition of all the circle vertices.
Edges can have multiplicity larger than 1.
Let $x_1, \ldots, x_n \in \R^d$ and let $h_i \colon \R \rightarrow \R$ be the normalized $i$-th probabilists' Hermite polynomial.
We call the following matrix $M \in \R^{n^{\card{L}} \times n^{\card{R}}}$ (based on the graph $H$) a \emph{graph matrix}\footnote{Note that the dimension of our definition is slightly larger than the one used in~\cite{AhnMP16}. However, all additional entries are 0, so the spectral norm bounds we use still apply.}
  \begin{equation}
    \label{eq:graph_matrix_def}
    M_{(\alpha,\beta)} =
    \begin{dcases}
      \hfill \sum_{\substack{g \colon V_\diamond \rightarrow [d] \\ \text{injective}}} \prod_{e \in E(H)} h_{\ell(e)} \Paren{x_{f_{\alpha,\beta}(e_\circ)}(g(e_\diamond)) } \,, &\quad \text{if $(\alpha,\beta)$ is a valid labelling,} \\
      \hfill 0 \,, &\quad \text{otherwise.}
    \end{dcases} \,.
  \end{equation}
  We call the graph $H$ its graph representation.
  Note that graph representations of graph matrices and graph polynomials are \emph{not} the same (but will be related).
  
  \paragraph{Admissible graph matrices.}

  \emph{Admissible} graph matrices are graph matrices derived from less complex ones via the process outlined below.
  The key property we will use from this is that we can derive norm bounds for the more complex ones from norm bounds for the simpler ones.
  We will describe how to derive the graph representations of the more complex ones from the graph representation of the simpler one.

  We start describing the starting point:
  Let $p \in \N$.
  Let $G_0$ be the following graph on $2p$ vertices.
  It has $p$ ``circle'' vertices and $p$ ``diamond'' vertices (subdivided into two subsets $L,R$).
  The graph $G_0$ is a union of $r$ cycles which alternate between circle and diamond vertices, and where the circle vertices alternate between $L$ and $R$ vertices.
  If a cycle contains an odd number of circle vertices, we put the last two circles on the same side.
  The cycles contain $\ell_1, \ldots, \ell_r$ circle vertices each (such that $\sum_{i=1}^r \ell_i = p$).
  There could be cycles that contain only one circle and diamond vertex.
  In this case, we identify the two parallel edges with a single edge.
  They are arranged such that $|L| = \lceil p/2 \rceil$ and $|R| = \lfloor p/2 \rfloor$. 
  
  We use the following definition (see \cref{fig:admissible_gm} for an illustration of the described process).
  \begin{definition}[Admissible Circle- and Diamond-Merged Graphs]
    \torestate{
      \label{def:circle_and_diamond_merge}
      Let $G_0$ be as above and denote by $(\text{circles})$ and $(\text{diamonds})$ the set of circle and diamond vertices, respectively.
      Let $S$ be a partition of $(\text{circles})$ and $T$ be a partition of $(\text{diamonds})$.
      Let $G(S,T)$ be as defined below.
      We call the graph matrices represented by $G(S,T)$ \emph{admissible}, where we additionally allow each edge to have multiplicity larger than 1, as long as the total sum of multiplicities is at most $O(p)$.

      \textbf{Admissible Circle-Merged Graph:} Let $S$ be such that in every part $s$ of $S$, all circles have distance larger than 2, i.e., they do not share a diamond neighbor.
      We call the graph $G(S)$ that is obtained from $G_0$ via the following procedure a \emph{circle-merged graph}.
      \begin{quote}
          For each $s \in S$, replace the circles in $s$ with a single ``super-circle'' vertex, whose neighbors are the union of the neighbors of the circle vertices in $s$.
          If there are edges of multiplicity larger than 1, replace them with an edge of multiplicity 1.
          If $s \cap L = \emptyset$, add the super-circle vertex to $R$, otherwise add it to $L$.
          The vertex set of the new graph consists of super-circles and diamonds.
      \end{quote}
  
      \textbf{Admissible Circle- and Diamond-Merged Graph:}
      Let $S$ be such that in every part $s$ of $S$, all circles have distance larger than 2, i.e., they do not share a diamond neighbor.
      We call a graph $G(S,T)$ that is obtained from $G_0$ via the following procedure an \emph{admissible circle- and diamond-merged graph}.
      \begin{enumerate}
          \item \emph{(Circle Merging)} Obtain $G(S)$ from $G_0$ as above. 
          \item \emph{(Diamond Merging)} For each $t \in T$, replace the diamonds in $t$ with a single ``super-diamond'' vertex.
              The edge set stays the same up to relabelling endpoints.
              If there are edges of multiplicity larger than 1, replace them with an edge of multiplicity 1.
              The vertex set of the new graph consists of super-circles and super-diamonds.
          \item \emph{(Edge Removal)} Remove edges arbitrarily from $G(S,T)$ subject to the following constraints:
          \begin{enumerate}
              \item We only remove edges that when forming $G(S,T)$ first had multiplicity at least 2.
              \item No super-circle becomes isolated.
          \end{enumerate}
      \end{enumerate}
    }
  \end{definition}

\begin{figure}[ht]
    \centering
    
    \begin{subfigure}[t]{0.4\textwidth}
        \centering
        \begin{tikzpicture}[baseline,every node/.style={draw, minimum size=0.7cm}]
            \node[circle, minimum size=0.6cm] (B) at (0,-1) [fill=navyblue] {};
            \node[circle, minimum size=0.6cm] (A) at (0,1) [fill=navyblue] {};
            \node[circle, minimum size=0.6cm] (D) at (4,-1) [fill=red] {};
            \node[circle, minimum size=0.6cm] (C) at (4,1) [fill=red] {};

            \node[draw, fit=(A)(B), inner sep=0.3cm] {};
            \node[draw, fit=(C)(D), inner sep=0.3cm] {};
            
            \node[draw, diamond, aspect=2](A1) at (2,1.5) {};
            \node[draw, diamond, aspect=2](B1) at (2,0.5) {};
            \node[draw, diamond, aspect=2](C1) at (2,-0.5) {};
            \node[draw, diamond, aspect=2](D1) at (2,-1.5) {};

            \draw (A) -- (A1);
            \draw (A1) -- (C);
            \draw (C) -- (B1);
            \draw (B1) -- (B);
            \draw (B) -- (C1);
            \draw (C1) -- (D);
            \draw (D) -- (D1);
            \draw (D1) -- (A);
        \end{tikzpicture}
        \caption{Initial graph matrix representation.}
    \end{subfigure}
    \hspace{1cm}
    \begin{subfigure}[t]{0.4\textwidth}
        \centering
        \begin{tikzpicture}[baseline,every node/.style={draw, minimum size=0.7cm}]
            \node[circle, minimum size=0.6cm] (B) at (0,0) {};
            \node[circle, minimum size=0.6cm] (D) at (4,0) {};

            \node[draw, fit=(B), inner sep=0.3cm] {};
            \node[draw, fit=(D), inner sep=0.3cm] {};
            
            \node[draw, diamond, aspect=2](A1) at (2,1.5) [fill=navyblue] {};
            \node[draw, diamond, aspect=2](B1) at (2,0.5) [fill=navyblue] {};
            \node[draw, diamond, aspect=2](C1) at (2,-0.5) [fill=red] {};
            \node[draw, diamond, aspect=2](D1) at (2,-1.5) [fill=green] {};

            \draw (B) -- (A1);
            \draw (A1) -- (D);
            \draw (D) -- (B1);
            \draw (B1) -- (B);
            \draw (B) -- (C1);
            \draw (C1) -- (D);
            \draw (D) -- (D1);
            \draw (D1) -- (B);
        \end{tikzpicture}
        \caption{Graph matrix representation after merging circles according to the partition given by the colors.}
    \end{subfigure}
    
    \vspace{1cm}
    
    \begin{subfigure}[t]{0.4\textwidth}
        \centering
        \begin{tikzpicture}[baseline,every node/.style={draw, minimum size=0.7cm}]
            \node[circle, minimum size=0.6cm] (B) at (0,0) {};
            \node[circle, minimum size=0.6cm] (D) at (4,0) {};

            \node[draw, fit=(B), inner sep=0.3cm] {};
            \node[draw, fit=(D), inner sep=0.3cm] {};
            
            \node[draw, diamond, aspect=2](A1) at (2,1) {};
            \node[draw, diamond, aspect=2](C1) at (2,0) {};
            \node[draw, diamond, aspect=2](D1) at (2,-1) {};

            \draw (B) -- (A1) [color=red];
            \draw (A1) -- (D) [color=red];
            \draw (B) -- (C1);
            \draw (C1) -- (D);
            \draw (D) -- (D1);
            \draw (D1) -- (B);
        \end{tikzpicture}
        \caption{Graph matrix representation after merging the diamonds according to the partition given by the colors. The red edges originally had multiplicity 2 after merging diamonds.}
    \end{subfigure}
    \hspace{1cm}
    \begin{subfigure}[t]{0.4\textwidth}
        \centering
        \begin{tikzpicture}[baseline,every node/.style={draw, minimum size=0.7cm}]
            \node[circle, minimum size=0.6cm] (B) at (0,0) {};
            \node[circle, minimum size=0.6cm] (D) at (4,0) {};

            \node[draw, fit=(B), inner sep=0.3cm] {};
            \node[draw, fit=(D), inner sep=0.3cm] {};
            
            \node[draw, diamond, aspect=2](A1) at (2,1) {};
            \node[draw, diamond, aspect=2](C1) at (2,0) {};
            \node[draw, diamond, aspect=2](D1) at (2,-1) {};

            \draw (A1) -- (D);
            \draw (B) -- (C1);
            \draw (C1) -- (D);
            \draw (D) -- (D1);
            \draw (D1) -- (B);
        \end{tikzpicture}
        \caption{Final graph matrix representation after removing one of the red edges.}
    \end{subfigure}
    
    \caption{Example of a sequence of circle and diamond merges and edge removals to obtain an admissible graph matrix.}
    \label{fig:admissible_gm}
\end{figure}

In~\cref{sec:graph_matrix_combo} we will show the following norm bound on admissible graph matrices.
\begin{lemma}[Norm Bound for Admissible Graph Matrices]
  \torestate{
  \label{lem:norm_bound_admissible_g_m}
  Let $d \leq n \leq d^2$ and let $M = M(G(S,T))\in \R^{n^{\card{L}}\times n^{\card{R}}}$ be an admissible graph matrix.
  Then,
  \[
    \E_{x_1, \ldots, x_n} \norm{M_i} \leq (p \cdot \log n)^{O(p)} \cdot d^{p/2} \cdot n^{r/4} \cdot n^{\tfrac{2p - (\card{L} + \card{R})} 4} \,.
  \]
  }
\end{lemma}
We will assume it for the remainder of this subsection.

\subsubsection{Proof of Graph Matrix Decomposition Lemma (\cref{lem:graph_matrix_decomposition})}

Recall that we obtain $G$, the graph representing the graph polynomial, by starting from a disjoint union of cycles and then merging vertices.
We first show \cref{lem:graph_matrix_decomposition} assuming that we never merge adjacent vertices.
We reduce to this case in~\cref{sec:reduction_to_ind_set}.

\begin{proof}[Proof of \cref{lem:graph_matrix_decomposition}, independent set case]
\looseness=-1
  Let $M_G$ be a matrix representation of $P_G^{=V}$.
  We first describe how to choose $L$ and $R$ (the goal is to obtain a "balanced" matrix that is reasonably close to being square):
  Recall that we obtained $G$ by performing vertex merges on $r$ disjoint cycles $C_1, \ldots, C_r$.
  We first partition the vertices in $C_1, \ldots, C_r$ into sets $L_0$ and $R_0$ such that $\card{L_0} = \card{R_0} \pm 1$ and build $L$ and $R$ from these.
  In particular, for each $C_i$, we can partition its vertices into two (almost) independent sets of sizes that differ by at most 1, by assigning vertices to one of them alternatingly (in case $C_i$ has odd length, we assign the last two to the same set).
  Let $L_0$ and $R_0$ be a union of the parts of these independent sets (i.e., such that $L_0$ and $R_0$ form an independent set in $G$ as well) such that $\card{L_0} = \card{R_0} \pm 1$.
  Note that this implies that $\card{L_0},\card{R_0} \geq \floor{\tfrac p 2}$ which will be important later on.
  We construct $L$ and $R$ from $L_0$ and $R_0$ as follows:
  Every time we perform a vertex merge and identify $v_1$ and $v_2$ we arbitrarily delete one of the two from the set it is contained in.

  Note that although our independent set condition implies that performing the vertex merges does not introduce new self-loops, some of the initial cycles might be of length 1, i.e., a self-loop, there might still be some.
  We will need the following auxiliary graph $H$, from which we construct the graph matrices in the decomposition: Its vertices have two types.
  It contains the vertices in $V$ as "circle vertices" and the edges in $E$ as "diamond vertices".
  The circle vertices are partitioned into $L$ and $R$ as in $G$.
  Further, for every edge in $E(G)$, let $e$ be the corresponding diamond vertex in $H$ and add the edges $\Set{u,e}$ and $\Set{e,v}$ to $E(H)$.
  Note that since the only self-loops in $G$ correspond to initial 1-cycles, the only edges in $H$ that appear with multiplicity larger than 1 are for $\Set{u,e}$ such that $u$ and $e$ do not have any other neighbors. 
  Note that the edges of $H$ both have exactly one endpoint among circle vertices and diamond vertices.
  Further, by construction $H$ has $2 \card{E(G)} = 2p$ edges.

  \paragraph{Constructing the graph matrix decomposition.}

  We will first prove the first two points of the decomposition theorem and show the last one at the end.
  We show that, in the independent set case, all graph matrices in the decomposition are admissible, and thus, the norm bound follows from~\cref{lem:norm_bound_admissible_g_m} (here we also use that $\card{L} + \card{R} = \card{V}$).
  Recall that $M_G$ is the matrix representation of $P_G$ and $M_G^=$ that of $P_G^{=V}$.
  We will construct a decomposition of $M_G$ which will induce the decomposition of $M_G^{=V}$.
  As we will see, in this later decomposition, the terms that do not correspond to admissible graph matrices are removed.

  \paragraph{Step 1: Obtaining few summands with small coefficients.}
  Let $(\alpha,\beta)$ be a valid labelling and consider the corresponding entry in $M_G$.
  We know that $G$ contains $p$ edges (counting multiplicity), label them $e^{(1)}, \ldots, e^{(p)}$.
  We can expand this as follows, where for an edge $e$, we denote its endpoints as $e_1, e_2$ (which could be the same)
  \begin{align*}
    \prod_{e \in E} \iprod{x_{f_{\alpha,\beta}(e_1)},x_{f_{\alpha,\beta}(e_2)}} &= \prod_{e \in E} \sum_{j=1}^d x_{f_{\alpha,\beta}(e_1)}(j) \cdot x_{f_{\alpha,\beta}(e_2)} (j) \\
    &=\sum_{j_1, \ldots, j_p = 1}^d \prod_{\ell = 1}^p x_{f_{\alpha,\beta}(e_1^{(\ell)})}(j_\ell) \cdot x_{f_{\alpha,\beta}(e_2^{(\ell)})} (j_\ell) \,. 
  \end{align*}
  We will group the above sum based on equality patterns among the $j_\ell$.
  We will later show that each equality pattern correspond to a linear combination of roughly $p^p$ graph matrices.
  In particular, let $T$ be a partition of the elements $\Set{1,\ldots, p}$ and let $(j_1,\ldots, j_p) \in [d]^p$.
  We denote $(j_1,\ldots, j_p) \sim T$ if $j_\ell = j_{\ell'}$ if and only if $\ell$ and $\ell'$ belong to the same part in $T$.
  Let $\cT$ be the set of all possible partitions, then each entry of $M_G$ can be written as
  \[
    \sum_{T \in \cT}\sum_{(j_1, \ldots, j_p) \sim T} \prod_{\ell = 1}^p x_{f_{\alpha,\beta}(e_1^{(\ell)})}(j_\ell) \cdot x_{f_{\alpha,\beta}(e_2^{(\ell)})} (j_\ell) \,.
  \]
  For $T\in \cT$, let $M_{G(T)}$ be the matrix whose entries are given by
  \begin{align}
    (M_{G(T)})_{(\alpha,\beta)} =
    \begin{dcases}
      \hfill \sum_{(j_1, \ldots, j_p) \sim T} \prod_{\ell = 1}^p x_{f_{\alpha,\beta}(e_1^{(\ell)})}(j_\ell) \cdot x_{f_{\alpha,\beta}(e_2^{(\ell)})} (j_\ell) \,, &\quad \text{if $(\alpha,\beta)$ is a valid labelling,} \\
      \hfill 0 \,, &\quad \text{otherwise.} \label{eq:almost_graph_matrix}
    \end{dcases} \,.
  \end{align}
  The above implies that
  \[
    M_G = \sum_{T \in \cT} M_{G(T)}\,.
  \]
  For each $T \in \cT$, consider the following auxiliary (multi-)graph $H(T)$ defined as follows:
  Start with $H$ and for a part $t \in T$, identify all diamond vertices $e^{(\ell)}$ such that $\ell \in t$ in $H$ with a single "super-diamond" (and keep all the edges).
  Note that each $\vec{j} = (j_1, \ldots, j_p) \sim T$ assigns a value to each super-diamond.
  Let $\tilde{e}$ denote the edges in $H(T)$ and $\ell_{\tilde{e}}$ their multiplicities.
  Denote by $\tilde{e}_\circ$ and $\tilde{e}_\diamond$ its circle and diamond endpoint, respectively.
  For $\vec{j} = (j_1, \ldots, j_k) \sim T$, let $\tilde{e}_\diamond(\vec{j})$ be the value that $\vec{j}$ assigns to $\tilde{e}_\diamond$.
  With these definitions in place, note that we can write the first case in~\cref{eq:almost_graph_matrix} as
  \[
    \sum_{\vec{j} = (j_1, \ldots, j_p) \sim T} \prod_{\tilde{e} \in E(H(T))} \Paren{x_{f_{\alpha,\beta}(\tilde{e}_\circ)}(\tilde{e}_\diamond(\vec{j}))}^{\ell_{\tilde{e}}} \,.
  \]
  Note that this \emph{almost} corresponds to the definition of graph matrices.
  The only difference is that we have monomials instead of Hermite polynomials.
  We will fix this shortly using that the Hermite polynomials form a basis of all polynomials.
  In particular, if $h_\ell$ is the $\ell$-th Hermite polynomial, graph matrices have entries of the form 
  \[
    \sum_{\vec{j} = (j_1, \ldots, j_p) \sim T} \prod_{\tilde{e} \in E(H(T))} h_{\ell_{\tilde{e}}}\Paren{x_{f_{\alpha,\beta}(\tilde{e}_\circ)}(\tilde{e}_\diamond(\vec{j}))} \,.
  \]
  We next argue that we can represent each $M_{G(T)}$ as a linear combination of at most $p^{O(p)}$ graph matrices in which the absolute value of each coefficient is at most $p^{O(p)}$.
  That is, that we can find a sequence $\Set{(M_i^{(T)}, c_i^{(T)})}_{i=1}^{N^{(T)}}$ of matrix-scalar pairs such that $N^{(T)} \leq p^{O(p)}$, for all $i$, $\abs{c_i^{(T)}} \leq p^{O(p)}$ and
  \[
    M_{G(T)} = \sum_{i=1}^{N^{(T)}}c_i^{(T)} M_i^{(T)} \,.
  \]
  Note that since $\card{\cT} \leq p^{O(p)}$ this directly induces a decomposition of $M_G$ satisfying the first and second property (not too many summands and bounded coefficients):
  \begin{equation}
    \label{eq:decomp_M_G}
    M_G = \sum_{T \in \cT} \sum_{i=1}^{N^{(T)}}c_i^{(T)} M_i^{(T)} \,.
  \end{equation}
  The decomposition of $M_{G(T)}$ directly follows from the following observation and the fact that $H$, and hence also $H(T)$, has $2p$ edges, i.e., $\sum_{\tilde{e} \in E(H(T))}\ell_{\tilde{e}} = 2p$.
  Note that some of the multiplicities of the new edges might be 0, in which case we remove the edge.
  \begin{observation}
  \label{obs:hermite_decomposition}
    Let $0 \leq k \leq 2p$ and $\ell_1, \ldots, \ell_k$ be positive integers such that $\sum_{j=1}^k \ell_j = 2p$, then we can express the polynomial (in formal variables $X$) $\prod_{j=1}^k X_j^{\ell_j}$ as a linear combination of at most $p^{O(p)}$ polynomials of the form $\prod_{j=1}^k h_{\ell'_j}(X_j)$ such that $\sum_{j=1}^k \ell'_j \leq 2p$ and the coefficient of each polynomial is at most $p^{O(p)}$ in absolute value.
    Further, if $\ell_j = 1$, then also $\ell_j' = 1$.
  \end{observation}
  \begin{proof}
    By~\cref{fact:hermite_monomial_expansion} we can write $X_j^{\ell_j} = \sum_{r = 0}^{\ell_j} c_r h_r(X_j)$ where $\card{c_r} \leq\ell_l! \leq \ell_j^{\ell_j}$.
    Substituting this for every $X_j^{\ell_j}$ and multiplying out terms, we get at most $\prod_{j = 1}^k (\ell_j  + 1) \leq (4p)^k = p^{O(p)}$ terms of the desired form.
    Note that for each term, the coefficient is at most $\prod_{j = 1}^k \ell_j^{\ell_j} \leq (2p)^{\sum_{j=1}^k \ell_j} = p^{O(p)}$ in absolute value.
    Lastly, if $\ell_j = 1$, then $X_j^{\ell_j} = h_{\ell_j}(X_j)$.
    Since this is the only occurence of $X_j$, the final claim follows.
  \end{proof}

  \paragraph{Step 2: Efron-Stein decomposition removes terms with isolated circles.}

  To obtain the final decomposition of $M_G^{=}$, we consider
  \begin{equation}
    \label{eq:decomp_M_G_equals}
    M_G^= = \sum_{T \in \cT} \sum_{i=1}^{N^{(T)}}c_i^{(T)} \Paren{M_i^{(T)}}^{=V} \,,
  \end{equation}
  Where $\Paren{M}^{=V}$ is defined analogously as before (applying the transformation $(\cdot)^{=f_{\alpha,\beta}(V)}$ entry-wise).
  To obtain a spectral norm bound on all of these terms, we make the following observation about how the Efron-Stein Decomposition interacts with our graph matrix decomposition.
  In particular, it keeps exactly those terms in which no circle vertex is isolated.
  \begin{observation}
    \label{obs:efron_stein_isolated_circles}
    Let $M$ be one of the terms in the graph matrix expansion of $M_G$ (cf.~\cref{eq:decomp_M_G}).
    Consider the graph representation of $M$ based on $H(T)$.
    If it contains an isolated circle vertex, then $(M)^{=V} = 0$.
    If it does not contain an isolated circle vertex, then $(M)^{=V} = M$.
  \end{observation}
  \begin{proof}
    We begin with the first case.
    We prove that every entry of $(M)^{=V}$ equals 0.
    To this end, fix $\alpha, \beta$ that form a valid labeling.
    For simplicity, assume that the assigned labels to circle vertices are $1, \ldots, k$ and that vertex 1 is isolated.
    Fix $\vec{j} = (j_1, \ldots, j_p) \sim T$.
    Note that by linearity of the Efron-Stein Decomposition, it is enough to show that (using that $f_{\alpha,\beta}(V) = [k]$)
    \[
        \Paren{\prod_{\tilde{e} \in E(H(T))} h_{\ell_{\tilde{e}}}\Paren{x_{f_{\alpha,\beta}(\tilde{e}_\circ)}(\tilde{e}_\diamond(\vec{j}))}}^{=[k]} = 0 \,.
    \]
    However, since vertex 1 is isolated, the product does not depend on $x_1$ and hence this follows by~\cref{fact:ES_remove_strict_subset_func}.

    For the second case, we again proceed entry-wise.
    Again, fix $\alpha, \beta$ that form a valid labeling and assume the assigned labels to circle vertices are $1, \ldots, k$.
    Again fix $\vec{j} = (j_1, \ldots, j_p) \sim T$.
    We will show that 
    \[
        \Paren{\prod_{\tilde{e} \in E(H(T))} h_{\ell_{\tilde{e}}}\Paren{x_{f_{\alpha,\beta}(\tilde{e}_\circ)}(\tilde{e}_\diamond(\vec{j}))}}^{\subseteq S} = 0
    \]
    whenever $S \subsetneq [k]$.
    The claim than follows since for any function $f$ on $n$ variables, $f = f^{\subseteq [n]}$.
    Fix $S \subsetneq [k]$ and let $i \in [k] \setminus S$.
    Since the circle vertex corresponding to $i$ is not isolated, we now that there exists at least one edge in $H(T)$ incident to it.
    Using that the random variables $x_i(j)$ are independent and that $\E h_\ell(x_i(j)) = 0$ for each $i,j, \ell$ it follows that
    \begin{align*}
        &\Paren{\prod_{\tilde{e} \in E(H(T))} h_{\ell_{\tilde{e}}}\Paren{x_{f_{\alpha,\beta}(\tilde{e}_\circ)}(\tilde{e}_\diamond(\vec{j}))}}^{\subseteq S} = \E_{S^c} \prod_{\tilde{e} \in E(H(T))} h_{\ell_{\tilde{e}}}\Paren{x_{f_{\alpha,\beta}(\tilde{e}_\circ)}(\tilde{e}_\diamond(\vec{j}))} \\
        &= \Paren{\E_{\Set{i}} \prod_{\tilde{e} \in E(H(T))\,, \tilde{e}_\circ = i} h_{\ell_{\tilde{e}}}\Paren{x_{f_{\alpha,\beta}(\tilde{e}_\circ)}(\tilde{e}_\diamond(\vec{j}))} } \Paren{ \E_{S^c \setminus \Set{i}} \prod_{\tilde{e} \in E(H(T))\,, \tilde{e}_\circ \neq i} h_{\ell_{\tilde{e}}}\Paren{x_{f_{\alpha,\beta}(\tilde{e}_\circ)}(\tilde{e}_\diamond(\vec{j}))} } \\
        &= \Paren{ \prod_{\tilde{e} \in E(H(T))\,, \tilde{e}_\circ = i} \E_{\Set{i}}h_{\ell_{\tilde{e}}}\Paren{x_{f_{\alpha,\beta}(\tilde{e}_\circ)}(\tilde{e}_\diamond(\vec{j}))} } \Paren{ \E_{S^c \setminus \Set{i}} \prod_{\tilde{e} \in E(H(T))\,, \tilde{e}_\circ \neq i} h_{\ell_{\tilde{e}}}\Paren{x_{f_{\alpha,\beta}(\tilde{e}_\circ)}(\tilde{e}_\diamond(\vec{j}))} } \\
        &= 0\, .
    \end{align*}
  \end{proof}

  \paragraph{Step 3: Terms without isolated circles are admissible graph matrices.}

  Let $M$ be a matrix in the graph matrix decomposition of $M^=$.
  Recall that we want to show the bound 
  \[
        \E_{x_1, \ldots, x_n} \norm{M} \leq (p \cdot \log n)^{O(p)} \cdot d \cdot d^{p/2} \cdot n^{r/4} \cdot n^{\tfrac{2p - \card{V}} 4} \,.
  \]
  We will show that $M$ is in fact an admissible graph matrix (cf.~\cref{def:circle_and_diamond_merge}) which implies the norm bound via~\cref{lem:norm_bound_admissible_g_m} (recall that we are already assuming that $d \leq n \leq d^2$).

  Let $\tilde{H}(T)$ be the graph representing $M$.
  By~\cref{obs:efron_stein_isolated_circles}, we know that $M^= = 0$ if it contains an isolated circle.
  So we can assume without loss of generality that it does not.
  We show that it satisfies the process outlined in~\cref{def:circle_and_diamond_merge}.
  Indeed, we can first do circle and diamond merges to obtain $H(T)$ and then obtain $\tilde{H}(T)$ from $H(T)$ by the ``Hermite decomposition'' using~\cref{obs:hermite_decomposition}.
  The circle and diamond merging are exactly as described, in particular, by our ``independent set'' condition, we never merge adjacent circle vertices.
  We only have to verify the edge removal step.
  That is, we have to show that we only remove edges that had multiplicity at least 2 in $H(T)$ and that no circle vertices become isolated.
  The second property follows by assumption.
  The first property follows directly by the last part of~\cref{obs:hermite_decomposition}.
  We also remark that clearly the total sum of edge multiplicities is at most $O(p)$.

  This finishes the proof of~\cref{lem:graph_matrix_decomposition} and thus also the one of~\cref{lem:bound-on-equals-term} in the independent set case.
  We show how to handle the general case in the next section.

\end{proof}

  \subsubsection{Reduction to Independent Set Case}
  \label{sec:reduction_to_ind_set}

  Next, we treat the case when the vertices we merged in $C_1, \ldots, C_r$ do not form an independent set.
  This will finish the proof of~\cref{lem:graph_matrix_decomposition}.

  \begin{proof}[Proof of~\cref{lem:graph_matrix_decomposition}, general case]
  Let $S$ be the partition of vertices according to which we merge.
  Let $z$ be the sum of the sizes of "connected components of size $\geq 2$" within any part of $S$.
  When there is no such component, define $z$ to be 1.
  We use an induction on $z$ to show that for all $z$ the same norm bound still holds.
  More specifically, we will show that it holds for all graph matrices that arise in the graph matrix expansion of $M_G^{=V}$.
  Note that at every level of the induction, we show the statement \emph{for all} $p \geq 2$ (more specifically, we only need to consider $p \geq z$) and all graph matrices arising in the graph matrix expansion.

  $z = 1$ is the independent set case shown above.
  Assume we have shown the statement for some $z \geq 1$ we will show it for $z' = z+1$.
  Let $p \geq z'$ be arbitrary.
  Note that since we require $\card{V} \geq 2$, $z' = 2$ implies that $p \geq 3$ (since otherwise everything would be merged to a single vertex).
  Hence, in any case we can assume $p \geq 3$.
  Let $v$ be a vertex in a connected component of size at least 2 in a part of $S$.
  Consider the graph $G'$ consisting of a disjoint union of $r$ cycles $C_1', \ldots, C_r'$ on $p-1 \geq 2$ vertices obtained by removing $v$ from $C_1, \ldots, C_r$ and connecting its two neighbors via an edge, if $v$ only has a single neighbor, we add a self-loop (of multiplicity 1) to this neighbor.
  Let $S'$ be a partition of the vertices of $C_1', \ldots, C_r'$ which is the same as $S$, except that $v$ is removed.
  Let $H$ and $H'$ be the auxiliary graphs corresponding to the two sets of disjoint cycles.
  Note that $H$ is the same as $H'$, except that it as an additional diamond vertex $D$ that has an edge of multiplicity 2 going to a circle vertex.
  Note that our inductive hypothesis applies to the graph matrix expansion of $M_H$.
  As before, let $\cT$ consist of all partitions of the elements $\Set{1,\ldots,p}$ and let $T \in \cT$.
  Similarly, let $\cT'$ consist of all partitions of the elements $\Set{1,\ldots,p-1}$ and let $T' \in \cT'$.
  Let $H(T), H'(T')$ and $M_{G(T)}, M_{G'(T')}$ be defined as before.
  We give a spectral upper bound on every term in the graph matrix expansion of $M_{G(T)}$.
  By induction, we have an upper bound on all graph matrices in the Hermite expansion of $M_{G'(T')}$.
  In what follows, we will identify $D$ with the element $p$.
  We make a case distinction.

  \paragraph{$D$ is a singleton in $T$.}
  Let $T' \in \cT'$ be such that $T' = T \setminus \Set{\Set{D}}$.
  Consider entry $(\alpha,\beta)$ in $M_{G(T)}$.
  Denote by $s$ the circle adjacent to $D$ in $H(T)$.
  Note that $H(T)$ is the same as $H'(T')$ except that $H(T)$ has an additional edge of multiplicity 2 from $s$ to $D$.
  Since $D$ is a singleton in $T$, it holds that
  \begin{align*}
    &\sum_{\vec{j} = (j_1, \ldots, j_p) \sim T} \prod_{\tilde{e} \in E(H(T))} \Paren{x_{f_{\alpha,\beta}(\tilde{e}_\circ)}(\tilde{e}_\diamond(\vec{j}'))}^{\ell_{\tilde{e}}} \\
    &= \sum_{\vec{j}' = (j_1, \ldots, j_{p-1}) \sim T', j_p \in [d]} \Paren{\prod_{\tilde{e} \in E(H'(T'))} \Paren{x_{f_{\alpha,\beta}(\tilde{e}_\circ)}(\tilde{e}_\diamond(\vec{j}))}^{\ell_{\tilde{e}}} } \cdot \Paren{x_{f_{\alpha,\beta}(s)}(j_p)}^2 \,.
  \end{align*}
  Using that $h_2(X) = \tfrac 1 {\sqrt{2}}(X^2 - 1)$, the above is equal to
  \[
    \sum_{\vec{j}' = (j_1, \ldots, j_{p-1}) \sim T', j_p \in [d]} \Paren{\prod_{\tilde{e} \in E(H'(T'))} \Paren{x_{f_{\alpha,\beta}(\tilde{e}_\circ)}(\tilde{e}_\diamond(\vec{j}))}^{\ell_{\tilde{e}}} } \cdot \Paren{\sqrt{2} \cdot h_2\Paren{x_{f_{\alpha,\beta}(s)}(j_p)} + 1} \,.
  \]

  Note that the expansion of the first product in terms of Hermite polynomials, corresponds exactly to the Hermite expansion of $M_{G'(T')}$.
  Thus, we can group the graph matrices in the Hermite expansion of $M_{G(T)}$ as follows: For every term appearing in the expansion of $M_{G'(T')}$, we associate two terms in the expansion of $M_{G(T)}$. One corresponds to the graph matrix that has an extra isolated diamond (this takes care of the 1-term) and one in which this diamond is connected to its circle neighbor (this takes care of the $h_2(\cdot)$-term).
  This enumerates all terms in the expansion of $M_{G(T)}$ and does not count any term twice.
  Fix one such pair.
  Since the second graph matrix has an extra edge, its spectral norm is at most a $(p \log n)^{O(1)}$ factor larger than that of the first one by~\cref{fact:graph_matrix_bound}.
  For the first one, since its graph matrix representation is exactly the same as the one in the expansion of $M_{G(T)}$ plus an additional isolated diamond, its spectral norm is larger by a factor (up to log-factors) of at most $d$.
  By induction, the norm bound we can show on the term in the expansion of $M_{G(T)}$ is
  \[
    (p \cdot \log n)^{O(p)} \cdot d \cdot d^{(p-1)/2} \cdot n^{r/4} \cdot n^{\tfrac{2p -2- \card{V}} 4} \cdot \eta^{\card{V} / 2} = \frac{(p \cdot \log n)^{O(p)} \cdot d \cdot d^{p/2} \cdot n^{r/4} \cdot n^{\tfrac{2p- \card{V}} 4} \cdot \eta^{\card{V} / 2}}{\sqrt{d} \sqrt{n}} \,.
  \]
  Since $\tfrac d {\sqrt{d} \sqrt{n}} \leq 1$, this finishes this part.

  \paragraph{$D$ is not a singleton in $T$.}
  
  Without loss of generality, assume that $p-1$ and $D$ are in the same part, say $t$ and let $s$ be the circle neighbor of $t$, which in $H$ was the circle neighbor of $D$.
  Let $T' \in \cT'$ be such that $T' = (T \setminus t) \cup (t \setminus \Set{D})$, i.e., $T$ restricted to elements $1, \ldots, p-1$.
  Note that there is a natural bijection between $\vec{j}' \sim T'$ and $\vec{j} \sim T$ given as follows:
  We can express every $\vec{j} = (j_1, \ldots, j_p) \sim T$ as $(\vec{j}' = (j_1, \ldots, j_{p-1}), j_{p-1})$, where $\vec{j}' \sim T'$.
  Note that $H(T)$ is the same as $H'(T')$ except that the edge from $s$ to $t$ has multiplicity two larger.
  In the following, let $\ell_{\tilde{e}}$ denote the edge multiplicities in $H'(T')$.
  Using the same identity as above, it follows that we can express entry $(\alpha,\beta)$ of $M_{G(T)}$ as
  \begin{align*}
    &\sum_{\vec{j} = (j_1, \ldots, j_p) \sim T} \Brac{ \prod_{\tilde{e} \in E(H(T)) \setminus \Set{s,t}} \Paren{x_{f_{\alpha,\beta}(\tilde{e}_\circ)}(\tilde{e}_\diamond(\vec{j}))}^{\ell_{\tilde{e}}} } \cdot \Paren{x_{f_{\alpha,\beta}(s)}(t(\vec{j}))}^{\ell_{\tilde{e}} + 2} \\
    &= \sum_{\vec{j}' = (j_1, \ldots, j_{p-1}) \sim T} \Brac{ \prod_{\tilde{e} \in E(H'(T')) \setminus \Set{s,t}} \Paren{x_{f_{\alpha,\beta}(\tilde{e}_\circ)}(\tilde{e}_\diamond(\vec{j}))}^{\ell_{\tilde{e}}} } \cdot \Paren{x_{f_{\alpha,\beta}(s)}(t(\vec{j}'))}^{\ell_{\tilde{e}} + 2} \,.
  \end{align*}
  Note that by~\cref{fact:hermite_monomial_expansion} the Hermite expansion of $x^r$ contains all Hermite polynomials of order $r-2\floor{\tfrac r 2}, r-2\floor{\tfrac r 2} -2, \ldots$.
  Thus, the Hermite expansion of $\Paren{x_{f_{\alpha,\beta}(s)}(t(\vec{j}'))}^{\ell_{\tilde{e}} + 2}$ is the same as that of $\Paren{x_{f_{\alpha,\beta}(s)}(t(\vec{j}'))}^{\ell_{\tilde{e}}}$ except that it contains one more term whose order is of the same parity as the others (and minor changes in coefficients).
  Splitting of this final term an expanding the product in Hermite basis, we obtain the graph matrix expansion of $M'(G(T'))$ up to minor changes in coefficients.
  Thus, all of these terms have bounded spectral norm by induction.
  The graph matrices arising from the final term in the Hermite expansion of $\Paren{x_{f_{\alpha,\beta}(s)}(t(\vec{j}'))}^{\ell_{\tilde{e}} + 2}$ have a correspondant in the Hermite expansion of $M'_{G(T')}$ as well, except that the labels might be larger by 1.
  Using~\cref{fact:graph_matrix_bound} again, we obtain the same norm bound for these terms as well.

\end{proof}

\subsection{Norm Bounds for \emph{Admissible} Graph Matrices}
\label{sec:graph_matrix_combo}

This subsection is devoted to proving \cref{lem:norm_bound_admissible_g_m}.
We first recall the definition of admissible graph matrices.
Note that in this section, we use $G$ for the graph of the graph matrix representation, previously $G$ was used for representing graph polynomials and $H$ for graph matrix representations.
This should not lead to confusion.
Recall that circle vertices have weight 1 and diamond vertices have weight $\log_n d$.

\paragraph{Admissible graph matrices and main lemma.}

Let $p \in \N$.
Let $G_0$ be the following graph on $2p$ vertices.
We start by describing the vertex set $V$, which consists two types of vertices: $p$ ``circle'' vertices and $p$ ``diamond'' vertices.
The circle vertices are further subdivided into two subsets $L,R$.
The graph $G_0$ is a union of $r$ cycles which alternate between circle and diamond vertices, and where the circle vertices alternate between $L$ and $R$ vertices.
If a cycle contains an odd number of circle vertices, we put the last two circles on the same side.
The cycles contain $\ell_1, \ldots, \ell_r$ circle vertices each such that $\sum_{i=1}^r \ell_i = p$.
There could be cycles that contain only one circle and diamond vertex.
In this case, we identify the two parallel edges with a single edge.
They are arranged such that $|L| = \lceil p/2 \rceil$ and $|R| = \lfloor p/2 \rfloor$. 

\restatedefinition{def:circle_and_diamond_merge}

We recall \cref{lem:norm_bound_admissible_g_m}.
\restatelemma{lem:norm_bound_admissible_g_m}
Note that $\card{L} + \card{R} = \card{S}$.
We use the following fact\footnote{The bound in this fact is an in-expectation bound, while Corollary 8.16 gives a high-probability bound. It can readily be obtained by combining the estimates in Corollary 8.16 with Lemma 8.7 which gives bounds on the expectation of higher-order trace powers of $M M^\top$, again choosing $q$ to be roughly $w(S_{\min}) \log n$.}
  \begin{fact}[Adaptation of Corollary 8.16 in~\cite{AhnMP16}]
    \label{fact:graph_matrix_bound}
    Let $M$ be a graph matrix with graph representation $G$.
    Let $S_{\min}$ be a minimum weight vertex separator of $L$ and $R$ and $V_{\mathrm{iso}}$ be the set of isolated vertices in $G$ (either circle or diamond vertices).
    Let $\ell = \sum_{e \in E(G)} \ell(e)$.
    Then, it holds that
    \[
      \E_{x_1, \ldots, x_n} \norm{M} \leq \Paren{w(V) \cdot \log n}^{O(\ell)} n^{\frac 1 2 \Paren{w(V) - w(S_{\min}) + w(V_{\mathrm{iso}})}} \,.
    \]
\end{fact}

Instead of working with minimum weight separators, it is more convenient to work with maximum vertex-capacitated flows:
We denote by $f(G(S,T))$ the maximum vertex-capacitated $(L,R)$-flow in $G(S,T)$.
By the generalized max-flow min-cut theorem, this is the same as the minimum-weight $(L,R)$ vertex separator in $G(S,T)$.
Define $f(G(S))$ analogously.
Further, for a vertex (either circle or diamond) $s$, we define $f(s)$ as the total incoming flow into this vertex.
The heart of this section is showing the following lower bound on flow in $G(S,T)$.
\begin{lemma}[Large flow in $G(S,T)$]
    \label{lem:graph-matrix-tradeoff}
    Let $G(S,T)$ be an admissible circle- and diamond-merged graph.
    Then, $G(S,T)$ satisfies the following inequality:
    \[
    f(G(S,T)) \geq \sum_{i=1}^r \floor{\ell_i/2} + w  (\text{number of isolated super-diamonds in $G(S,T)$}) - w  (p - |T|) - \frac 32  (p - |S|) \, .
    \]
\end{lemma}

With these two in hand, we can prove~\cref{lem:norm_bound_admissible_g_m}
\begin{proof}[Proof of~\cref{lem:norm_bound_admissible_g_m}]
  First, observe that by definition we have $w(V) \leq O(p)$ and by assumption $\ell = O(p)$.
  Further, since $d \leq n \leq d^2$, the weight of diamond vertices $w = \log_nd$ satisfies $1/2 \leq w \leq 1$.

  We will also use the following fact proven below.
  By the generalized max-flow-min-cut theorem, we know that if the maximum vertex-capacitated $(L,R)$-flow in $G(S,T)$ is at least $f$, then the minimum weight vertex separator has weight at least $f$ as well.\footnote{We can assign capacity $\infty$ to all edges.}
  Using this, \cref{lem:graph-matrix-tradeoff} and the bound in~\cref{fact:graph_matrix_bound} it follows that $\E_{x_1, \ldots, x_n} \norm{M(G(S,T))}$ is at most $(p \cdot \log n)^{O(p)}$ times $n$ to the power of the following expression
  \begin{align*}
    &\frac{\card{S} + w \cdot \card{T} - f + w \cdot (\text{number of isolated super-diamonds in $G(S,T)$})} 2 \\
    &\leq \frac{\card{S} + w \cdot \card{S} - \sum_{i=1}^r \floor{\ell_i/2} + w \cdot (p - \card{T}) + \tfrac 3 2 (p - \card{S})} 2 \\
    &= \frac{- \tfrac 1 2 \card{S} - \sum_{i=1}^r \floor{\ell_i/2} + w \cdot p + \tfrac 3 2 p } 2 \leq \frac{p + \tfrac r 2 + w \cdot p - \tfrac 1 2 \card{S}} 2\,.
  \end{align*}
  Thus, we can conclude that:
  \[
    \E_{x_1, \ldots, x_n} \Norm{M} \leq (p \cdot \log n)^{O(p)} \cdot d^{p/2} \cdot n^{r/4} \cdot n^{\tfrac{2p - \card{S}} 4} \,.
  \]
\end{proof}

\subsubsection{Finding Large Flows in \texorpdfstring{$G(S,T)$}{G(S,T)}}

In this section, we will prove~\cref{lem:graph-matrix-tradeoff}.
It will follow from \cref{clm:G-to-G(S),clm:G(S)-to-G(S_T)}, which we state and prove below.

\paragraph{Merging circles.}
We establishing the following claims. 
\begin{claim}
    \label{clm:G-to-G(S)}
    Let $G(S)$ be an admissible circle-merge graph.
    Then, $G(S)$ satisfies
    \[
    f(G(S)) \geq \sum_{i=1}^r \floor{\ell_i/2} - (p - |S|) \, .
    \]
\end{claim}
\begin{proof}
    We first argue that $f(G_0) = \sum_{i=1}^r \floor{\ell_i/2}$.
    For $\ell_i = 1$, i.e., self-loops, the argument is trivial.
    Without loss of generality assume that there are less circle vertices in $L$ than in $R$.
    Then, every circle vertex in $L$ can send 1/2-units of flow via each of its diamond neighbors to $R$.
    Since every circle vertex (in both $L$ and $R$) has degree 2, we never over-saturate circle vertices.
    Further, we route at most 1/2-units of flow via each diamond vertex, thus, we also respect the capacity constraints there.
    Since $L$ has size at least $\floor{\ell_i/2}$, we are done.

    Next, let $\mathbf{f}$ be an $(L,R)$ vertex-capacitated flow in $G_0$ of value $\sum_{i=1}^r \floor{\ell_i/2}$.
    The process of obtaining $G(S)$ from $G_0$ can be viewed as a sequence of $p - |S|$ merges of pairs of circle vertices.
    After each merge, $\mathbf{f}$ still respects the capacity constraints on all but at most one one circle vertex (the newly-merged one), where $\mathbf{f}$ may now try to route as many as $2$ units of flow where the capacity is $1$.
    Reducing the flow by at most $1$ 
    fixes the capacity constraints.
\end{proof}

\begin{claim}
  \label{claim:no_iso_diamonds_in_G(S)}
  Let $G(S)$ be an admissible circle-merge graph.
  Then, the following properties hold.
  \begin{enumerate}
      \item Every diamond in $G(S)$ has degree exactly 2 (in particular, there are no isolated diamonds in $G(S)$).
      \item Every super-circle has even degree.
      \item There are no edges of multiplicity larger than 1.
  \end{enumerate}
\end{claim}
\begin{proof}
  These properties follow since they hold true in $G_0$ and we do not delete any edges when forming $G(S)$ and we never merge circles that have the same diamond neighbor.
\end{proof}

\paragraph{Merging diamonds.} 

Let $f(G(S))$ be a flow in $G(S)$ satisfying the conclusion of~\cref{clm:G-to-G(S)}.
We first make a simplifying observation that will help us later on.
\begin{claim}
  \label{claim:only_path_diamonds_matter}
  Let $f(G(S))$ be as above.
  Let $t$ be a diamond vertex that has both its circle neighbors on the same side.
  Then, we can without loss of generality assume that $f(t) = 0$.
\end{claim}
\begin{proof}
  Suppose $t$ has both its neighbors $v_1, v_2$ on the left side and the flow is going from $v_1$ to $v_2$.
  So there is a "flow-path" transporting $f(t)$ units of flow from (say) $v_1$ to $v_2$ (and then further) via the path $v_1 \rightarrow t \rightarrow v_2$.
  We can remove this path without change the total amount of flow going from $L$ to $R$ and all capacity constraints are still observed.
  A symmetric argument works when $t$ has both its neighbors in $R$.
\end{proof}

\begin{claim}
    \label{clm:G(S)-to-G(S_T)}
    Let $G(S,T)$ be an admissible circle- and diamond-merged graph.
    Then, the graph $G(S,T)$ satisfies
    \[
    f(G(S,T)) - w \cdot (\text{number of isolated super-diamonds in $G(S,T)$}) \geq f(G(S)) - w \cdot (p - |T|) - \frac{p - |S|}{2} \,.
    \]
\end{claim}
\begin{proof}

    Our goal will be to construct a large flow in $G(S,T)$, starting from a maximal flow in $G(S)$:

    \paragraph{Constructing a large flow in $G(S,T)$.}

    Let $\mathbf{f}$ be a maximal vertex-capacitated $(L,R)$-flow in $G(S)$.
    We will construct a vertex-capacitated $(L,R)$-flow $\mathbf{f}'$ in $G(S,T)$ by modifying $\mathbf{f}$ as follows (arguing equivalence class by equivalence class).
    Fix an equivalence class of diamonds $t \in T$ and call its associated super-diamond $x_t$.
    Let $E(t)$ be the edges incident to $t$ in $G(S)$ and $E_L(t), E_R(t)$ be the edges connecting to $L$ and $R$, respectively.
    Note that $E(t)$ is non-empty by~\cref{claim:no_iso_diamonds_in_G(S)} (and hence at least one of $E_L(t)$ and $E_R(t)$ is non-empty as well).
    Let $E(x_t)$ be the edges incident to $x_t$ in $G(S,T)$.
    To define $\mathbf{f'}$ on the edges $E(x_t)$ we split into cases.

    For convenience, we say that the edges of an equivalence class are \emph{preserved}, if both 1.) $E_L(t) \neq \emptyset$ implies that there is an edge from $L$ to $x_t$ in $G(S,T)$ and 2.) $E_R(t) \neq \emptyset$ implies that there is an edge from $R$ to $x_t$ in $G(S,T)$ are true.
    Note that this does not imply anything about whether $E_L(t)$ or $E_R(t)$ are empty or not.

    \paragraph{Case 1: No edges lost (implies $x_t$ is not isolated)}
    
    \begin{quote}
      The edges of $t$ are preserved.
    \end{quote}
    Let $v \in t$ be a vertex for which $f(v)$ is minimal (among vertices in $t$). 
    If $\mathbf{f}(v) = 0$, set $\mathbf{f}'$ to be zero for all edges in $E(x_t)$.
    Note that this implies that all vertices in $t$ have one neighbor in $L$ and one in $R$.
    Further, each edges in $G(S)$ carries at least $f(v)$ units of flow.
    Since edge preservation holds, there is an edge from $x_t$ to both $L$ and $R$ and by the above, we can assign these flow $f(v)$ without violating capacity constraints.
    If $\mathbf{f}(v) \neq 0$ a
    
    Note that in each of the two scenarios above we have $\mathbf{f'}(x_t) = \mathbf{f}(v)$.
    Note that by~\cref{claim:no_iso_diamonds_in_G(S)} not both of $E_L(t)$ and $E_R(t)$ are empty.
    Thus, since no edges are lost, $x_t$ is connected to at least one of $L$ and $R$ by assumption in either case, so it is not isolated.

    \paragraph{Case 2: Edges lost but $x_t$ is not isolated}

    \begin{quote}
      The edges of $t$ are not preserved, but $x_t$ is not isolated in $G(S,T)$.
    \end{quote}

    In this case, let $\mathbf{f}$' route zero flow along edges in $E(x_t)$.
    By assumption $x_t$ is not isolated in $G(S,T)$.

    \paragraph{Case 3: Edges lost and $x_t$ is isolated}
    \begin{quote}
      The edges of $t$ are not preserved and $x_t$ is isolated in $G(S,T)$.
    \end{quote}

    Also in this case, let $\mathbf{f}$' route zero flow along edges in $E(x_t)$.
    Furthermore, $x_t$ is isolated in $G(S,T)$ by assumption.

    \paragraph{Auxiliary lemmas about Cases 2 and 3.}

    We record the following observation about edge preservation.
    \begin{observation}[Properties of Edge Deletions in $G(S,T)$]
        \label{obs:edge_preservation}
        Let $G(S,T)$ be an admissible circle- and diamond-merged graph.
        The following two properties are true.

        \textbf{Charging Edge Deletions to Edges in $G(S)$:} For every edge that we remove from $G(S,T)$, we can ``charge'' this to two distinct edges in $G(S)$ that are \emph{not} present in $G(S,T)$ (after relabelling the diamond endpoint) and that have the same super-circle endpoint.
        Further, aggregated over all removed edges in $G(S,T)$, we never charge an edge in $G(S)$ twice.

        \textbf{No Edge Deletions Incident to Singleton Diamonds:} Additionally, whenever we remove an edge from $G(S,T)$, it must be the case that its super-diamond was obtained by a partition of diamonds of size at least 2.
    \end{observation}
    \begin{proof}
        Let $x_t$ be a super-diamond in $G(S,T)$ and suppose we remove the edge $e = (e_\circ, x_t)$ from $G(S,T)$, where $e_\circ$ is the super-circle endpoint.
        Recall that all edges in $G(S)$ have multiplicity 1 (cf.~\cref{claim:no_iso_diamonds_in_G(S)}).
        Thus, since it first had multiplicity at least 2 in $G(S,T)$, it must be the case that the equivalence class $t$ of diamonds that formed $x_t$ must have had size at least 2.
        Further, $e_\circ$ must have been incident to two edges in $t$.
        This directly implies the second claim.
        For the first, note that since these two edges are not present in $G(S,T)$ (since we remove the edge $e$), we claim that we can charge the removal of $e$ to them.
        Indeed, all the charged edges will be distinct, since for every super-circle and super-diamond/equivalence class of diamonds pair, we only ever remove one edge. 
    \end{proof}

    To proceed, we need the following two claims.
    \begin{claim}
    \label{clm:combinatorics-1}
    \begin{align}
    \label{eq:combinatorics-2}
    (\text{number of case-2 $t$s}) + 2 \cdot (\text{number of case-3 $t$s}) \leq p - |S| \, .
    \end{align}
    \end{claim}
    \begin{proof}[Proof of \cref{clm:combinatorics-1}]
    Our approach to prove this claim starts with a lower bound on the number of edges removed in the ``edge removal step'' of obtaining $G(S,T)$.
    
    For every case-$2$ $t$, by definition at least one of $E_L(t)$ or $E_R(t)$ is non-empty in $G(S)$, but there is no edge from $x_t$ to $L$ or $R$ respectively.
    By~\cref{obs:edge_preservation} we know we can find a super-circle in $G(S)$ such that we charged two of the edges incident to it for this missing edge.

    For every case-$3$ $t$, since $x_t$ is isolated, we must have lost all edges incident to  $t$.
    In this case $t$ must have contained at least 2 diamonds (cf.\ the first property in~\cref{obs:edge_preservation}).
    Further, since every diamond has degree 2 in $G(S)$ and has distinct neighbors by~\cref{claim:no_iso_diamonds_in_G(S)}, $x_t$ must have had degree at least 2 in $G(S,T)$ before the ``edge-removal'' step.
    For each of these edges, by~\cref{obs:edge_preservation} we can find a super-circle in $G(S)$ such that we charged two of the edges incident to it for this missing edge.

    For a super-circle $y_s$ in $G(S)$, let $\alpha_s \in \N$ be the number of edges incident to it that have been charged.
    Note that $\alpha_s$ is even by definition.
    Since $y_s$ is not isolated in $G(S,T)$ and all the charged edges in $G(S)$ are not present in $G(S,T)$, it must have degree at least $\alpha_s + 1$ in $G(S)$.
    Since every circle vertex in $G(S)$ has even degree by~\cref{claim:no_iso_diamonds_in_G(S)}, this implies that it in fact must have degree at least $\alpha_s + 2$ (again, in $G(S)$).
    Further, since every circle has degree 2 in $G_0$ and no edges were lost when forming $G(S)$, the subset $s \in S$ of circles which collapsed to form $y_s$ has size $|s| \geq \tfrac {\alpha_s} 2 + 1$.
    So we have
    \[
    p = \sum_{s \in S} |s| \geq \sum_{s \in S} \Paren{\frac {\alpha_s} 2 + 1} \geq      (\text{number of case-2 $t$s}) + 2 \cdot (\text{number of case-3 $t$s}) + |S| \, ,
    \]
    which proves \eqref{eq:combinatorics-2}.
    \end{proof}

    Moving on, we next establish the following claim. 
    \begin{claim}
        \label{clm:combinatorics-2}
        Let $t \in T$ be an equivalence class whose edges are not preserved (i.e., $t$ belongs to cases $2$ or $3$). Then there exist distinct vertices $v, v' \in t$ such that $\mathbf{f}(v_t) + \mathbf{f}(v_t') \leq 1$.
    \end{claim}
    \begin{proof}[Proof of \cref{clm:combinatorics-2}]
    Since the edges of $t$ are not preserved, there exists a super circle $s$ that has at least two neighbors in $t$, call them $v$ and $v'$.
    Assume that $s \in L$, the case that $s \in R$ is symmetric.
    If either of them has both neighbors on the same side we know that $\mathbf{f}(v) = 0$ (or $\mathbf{f}(v') = 0$), so assume without loss of generality that this is the case for neither of them.
    Then for both of them the only edge from $L$ is from $s$.
    Since $s$ has capacity $1$ it follows that $\mathbf{f}(v) + \mathbf{f}(v') \leq \mathbf{f}(s) \leq 1$.
    \end{proof}

  \paragraph{Putting everything together.}
    Now we will put together \cref{clm:combinatorics-1,clm:combinatorics-2} and \cref{clm:G(S)-to-G(S_T)}.
    We aim to lower-bound amount of $(L,R)$ flow routed by $\mathbf{f}$'.
    If $x_1,\ldots,x_{|T|}$ are the super-diamonds in $G(S,T)$, then
    \[
    f(G(S,T)) \geq \sum_{t \in T} \mathbf{f}'(x_t) \, .
    \]
    Furthermore, $f(G(S)) \leq \sum_{t \in T} \mathbf{f}(t)$, where $\mathbf{f}(t) = \sum_{v \in t} \mathbf{f}(v)$.
    For each $t$ in case $1$ there exists $v_t \in t$ such that $\mathbf{f}'(x_t) = \mathbf{f}(v_t)$ and hence
    \[
      \mathbf{f}'(x_t) - \mathbf{f}(t) \geq - \sum_{v \in t, v \neq v_t} \mathbf{f}(v) \, ,
    \]
    while for each $t$ in case $2$ and case $3$,
    \[
    \mathbf{f}'(x_t) - \mathbf{f}(t) = - \mathbf{f}(t) \, .
    \]
    Hence, since $\mathbf{f}(v) \leq w$ for every $v$.
    \begin{align}
    f(G(S,T)) - f(G(S)) & \geq - \sum_{t \in \text{ case 1}, v \in T, v \neq v_t} \mathbf{f}(v) - \sum_{t \in \text{ cases 2,3}, v \in t} \mathbf{f}(v) \nonumber \\
    & \geq - w \cdot (|\{v \, : \, v \in t, t \in \text{ case 1} \}| - (\text{number of $t \in$ case 1})) - \;\;\sum_{\mathclap{\substack{t \in \text{ cases 2,3}\,,\\ v \in t}}} \mathbf{f}(v) \,.\label{eq:combinatorics-3}
    \end{align}
    Now, by \cref{clm:combinatorics-2}.
    \begin{align}
    &- \sum_{t \in \text{cases 2,3}, v \in t} \mathbf{f}(v) \\
    &\geq - w \cdot (|\{v \, : \, v \in t, t \in \text{ cases 2,3} \}| - 2 \cdot (\text{number of $t \in$ cases 2,3})) - (\text{number of $t \in$ cases 2,3}) \,. \label{eq:combinatorics-4}
    \end{align}
    Putting together \eqref{eq:combinatorics-3} and \eqref{eq:combinatorics-4}, we get
    \[
    f(G(S,T)) - f(G(S)) \geq -w \cdot (p - |T|) + (w - 1) \cdot (\text{number of $t \in $ cases 2,3}) \, .
    \]
    Furthermore, the number of isolated super-diamonds in $G(S,T)$ is at most the number of $t$ in case $3$.
    So, using also $w-1 \geq -1/2$,
    \begin{align*}
    & f(G(S,T)) - f(G(S)) - w \cdot (\text{number of isolated super-diamonds}) \\
    & \quad \geq -w \cdot (p - |T|) - \frac 12  \cdot (\text{number of $t \in $ case 2}) - (\text{number of $t \in $ case 3})\, .
    \end{align*}
    Applying \cref{clm:combinatorics-1} finishes the proof.
\end{proof}

\subsection{Auxiliary Lemmas used in Induction}
\label{sec:induction_subset_terms}

This subsection proves~\cref{lem:main_induction}.
Recall that we use the following notation.
For a vertex $v$, we use $\deg(v)$ to denote its degree, $\degdagger(v)$ to denote its degree counting only edges that are \emph{not} loops, and $s(v)$ to denote the number of self-loops.
Note that $\deg(v) = \degdagger(v) + 2s(v)$ and that the lemma also applies when $u$ only has self-loops.

We start with the following helper lemma; see \cref{fig:expecting_vertices} for an illustration.
\begin{lemma}
\label{lem:union-of-cycles-merge}
    Let $G$ be a vertex merge of a union of $r$ cycles whose lengths sum to $p$.
    Let $u$ be a vertex in $G$ and $G'$ be a graph obtained by removing $u$ from $G$, along with all its incident edges, and adding a matching $M$ on the neighbors of $u$ (counted with multiplicity).
    Then $G'$ is a vertex merge of a union of at most $r'$ cycles whose lengths sum to $p - \deg(u)/2$, where $r' = r + \tfrac{\degdagger(u)}{2} -1$.
    Further, all other vertices in $G'$ have the same degree as in $G$.
\end{lemma}
\begin{proof}
    Let $\cC = C_1,\ldots,C_r$ be the union of cycles of which $G$ is a vertex merge.
    Let $S_1,\ldots,S_N$ be a partition of the vertices of $C_1 \cup \ldots \cup C_r$ so that collapsing the vertices in each equivalence class yields $G$.
    Let $S_u$ be the member of the partition whose vertices collapse to $u$ in the vertex-merged graph $G$.
    The matching $M$ specifies a matching on the vertices not in $S_u$ which are incident to edges across the cut induced by $S_u$.
    Create a new union of cycles $\cC'$ by removing all the vertices in $S_u$ and adding edges between the matched vertices.
    In the resulting graph all vertex degrees are still $2$, so it must be a union of disjoint cycles, and $G'$ is clearly a vertex merge of $\cC'$.
    If $u$ had no other neighbors in $G$, i.e., $\degdagger(u) = 0$, there are no edges across the cut.
    Thus, the vertices in $S_u$ must have been a union of some of the cycles in $\cC$.
    It follows, that $\cC'$ is a union of at least one cycle less.
    If $u$ had only two edges to other vertices in $G$, i.e., $\degdagger(u) = 2$, the number of edges across the cut is 2.
    Thus, both of these edges must belong to the same cycle.
    It follows, that the number of cycles in $\cC'$ is the same as in $\cC$.
    Else, since there are at most $\degdagger(u)$ vertices in the matching and we can partition them into pairs belonging to the same cycle, this accounts for $\degdagger(u)/2$ cycles.
    Before removing $u$ at least one cycle was intersecting $S_u$, thus, the new number of cycles is at most $r + \tfrac{\degdagger(u)}{2} - 1$.

    For every vertex $u'$ other than $u$, if it previously was adjacent to $u$ (with multiplicity $m$), it is adjacent to $m$ additional vertices in $G'$ (in addition to other vertices).
    Thus, its degree remains unchanged.
    Further, all these vertices will be distinct from $u'$ and hence also the ``dagger-degree'' stays the same. 
    
    Next, we count how many edges are in $\cC'$.
    Note that every edge in $\cC$ with both endpoints in $S_u$ has no counterpart in $\cC'$; every pair of edges with a single endpoint in $S_u$ account for one edge in $\cC$.
    Denote by $t_1,t_2$ the number of edges that have one, respectively, two endpoints in $S_u$.
    Thus, $\cC'$ contains $p - t_2 - \tfrac{t_1}{2}$ edges.
    Since $t_1 + 2t_2 = 2\card{S_u}$, $\cC'$ has $p - |S_u|$ edges.
    The total degree of vertices in $S_u$ is $2|S_u|$, hence the degree of $u$ in $G$ is also $2|S_u|$, proving the lemma.
\end{proof}

\begin{figure}[ht]
    \centering
    \begin{subfigure}[t]{0.45\textwidth}
        \centering
        \begin{tikzpicture}[scale=1.5, every node/.style={draw, circle, minimum size=5mm}]
            \node (A) at (0,0) {$u$};
            \node (B) at (-0.7,-0.4) {};
            \node (C) at (-0.7,0.4) {};
            \node (D) at (0.7,-0.4) {};
            \node (E) at (0.7,0.4) {};

            \draw (A) -- (B);
            \draw (A) -- (C);
            \draw (A) -- (D);
            \draw (A) -- (E);
            \draw (B) -- (C);
            \draw (D) -- (E);
        \end{tikzpicture}
        \caption{An example graph $G$.}
        \label{fig:subfig1}
    \end{subfigure}
    \hfill
    \begin{subfigure}[t]{0.45\textwidth}
        \centering
        \begin{tikzpicture}[scale=1, every node/.style={draw, circle, minimum size=5mm}]
            \node (F1A) at (-1,0) {};
            \node (F1B) at (0,0) {};
            \draw (F1A) -- (F1B);

            \node (F1C) at (-1,-1) {};
            \node (F1D) at (0,-1) {};
            \draw (F1C) -- (F1D);

            \node (F2A) at (1,0) {};
            \node (F2B) at (2,0) {};
            \node (F2C) at (2,-1) {};
            \node (F2D) at (1,-1) {};
            \draw (F2A) -- (F2B) -- (F2C) -- (F2D) -- (F2A);

            \node (F2A) at (3,0) {};
            \node (F2B) at (4,0) {};
            \node (F2C) at (4,-1) {};
            \node (F2D) at (3,-1) {};
            \draw (F2A) -- (F2B) -- (F2C) -- (F2D) -- (F2A);
        \end{tikzpicture}
        \caption{A family of new graphs, note that the four cycle appears with multiplicity 2.}
        \label{fig:subfig2}
    \end{subfigure}
    \caption{An illustration of \cref{lem:union-of-cycles-merge}.}
    \label{fig:expecting_vertices}
\end{figure}

We can now prove~\cref{lem:main_induction}.
\restatelemma{lem:main_induction}

\begin{proof}
    Since $S$ is a strict subset of $V$, let $u \in V \setminus S$.
    Let $s(u)$ be the number of self-loops on $u$ and let $N(u)$ be the neighborhood of $u$, counted with multiplicity if $G$ is a multigraph.
    Note that $N(u)$ could be empty in the case when $u$ only has self-loops.
    Let $E(u)$ be the edges incident to $u$.
    Let $\cM(u)$ be the set of matchings on the neighbors of $u$ (excluding $u$ itself, but counted with multiplicity if $u$ participates in multi-edges).
    For $M \in \cM(u)$, let $G'(M) = (V',E'(M))$ be the graph given by setting $V' = V \setminus \{u\}$ and $E'(M) = (E \setminus E(u)) \cup M$ to be $E$, with edges incident to $u$ removed, and $M$ added.
    Note that when initially $V = \Set{u}$, $G'$ will be the empty graph.
    In this case we use the convention that $\mathcal{M}(u)$ only contains the empty set, such that e.g., $\sum_{M \in \mathcal{M}(u)} \prod_{\Set{v,w} \in E'(M)} \iprod{x_{f(v)},x_{f(w)}} = 1$.
    
    We start by expanding $P_G^{\subseteq S}$ as follows, where~\cref{eq:one} follows by linearity of $(\cdot)^{\subseteq f(S \cup \{u\})}$, \cref{eq:two} follows by \Cref{cor:wicks-with-norm}, using that the power of the norm is at most $p \leq d$, and \cref*{eq:three} follows by linearity of $(\cdot)^{\subseteq f(S \cup \{u\})}$
    \begin{align}
    P_G^{\subseteq S}(w,x)
    & = \sum_{\substack{f \, : \, V \rightarrow [n] \\ \text{ injective}}} \prod_{v \in V} w_{f(v)} \cdot \Paren{\prod_{\{v,w\} \in E} \iprod{x_{f(v)}, x_{f(w)}}}^{\subseteq f(S)} \nonumber \\
    & = \sum_{\substack{f \, : \, V \rightarrow [n] \\ \text{ injective}}} \prod_{v \in V} w_{f(v)} \cdot \E_{x_{f(u)}} \Paren{\prod_{\{v,w\} \in E} \iprod{x_{f(v)}, x_{f(w)}}}^{\subseteq f(S \cup \{u\})} \nonumber \\
    & = \sum_{\substack{f \, : \, V \rightarrow [n] \\ \text{ injective}}} \prod_{v \in V} w_{f(v)} \cdot  \Paren{\E_{x_{f(u)}} \prod_{\{v,w\} \in E} \iprod{x_{f(v)}, x_{f(w)}}}^{\subseteq f(S \cup \{u\})} \label{eq:one} \\
    & = \sum_{\substack{f \, : \, V \rightarrow [n] \\ \text{ injective}}} \prod_{v \in V} w_{f(v)} \cdot  \Paren{\sum_{M \in \cM(u)} O(d)^{s(u)} \prod_{\{v,w\} \in E'(M) } \iprod{x_{f(v)}, x_{f(w)}}}^{\subseteq f(S \cup \{u\})} \label{eq:two} \\
    & = O(d)^{s(u)} \cdot \sum_{M \in \cM(u)} \sum_{i \in [n]} w_i \cdot \sum_{\substack{f \, : \, V' \rightarrow  ([n] \setminus \{i\}) \\ \text{injective}}} \prod_{v \in V'} w_{f(v)} \cdot \Paren{\prod_{\{v,w\} \in E'(M)} \iprod{x_{f(v)}, x_{f(w)}}}^{\subseteq f(S)} \label{eq:three} \, .
    \end{align}
    Next, using that $w_i^2 = w_i$ we obtain that $\Set{w_i^2 = w_i \,, \sum_{i=1}^n w_i = \eta n}$ proves at constant degree that
    \begin{align*}
      &w_i \cdot \sum_{\substack{f \, : \, V' \rightarrow  ([n] \setminus \{i\}) \\ \text{injective}}} \prod_{v \in V'} w_{f(v)} \cdot \Paren{\prod_{\{v,w\} \in E'(M)} \iprod{x_{f(v)}, x_{f(w)}}}^{\subseteq f(S)} \\
      &= w_i \cdot \sum_{\substack{f \, : \, V' \rightarrow  [n] \\ \text{injective}}} \prod_{v \in V'} w_{f(v)} \cdot \Paren{\prod_{\{v,w\} \in E'(M)} \iprod{x_{f(v)}, x_{f(w)}}}^{\subseteq f(S)} \\
      &- w_i \cdot \sum_{\substack{f \, : \, V' \rightarrow  [n]\,,i \in \mathrm{Im}(f) \\ \text{injective}}} \prod_{v \in V'} w_{f(v)} \cdot \Paren{\prod_{\{v,w\} \in E'(M)} \iprod{x_{f(v)}, x_{f(w)}}}^{\subseteq f(S)} \\
      &= w_i \cdot P_{G'(M)}^{\sse f(S)} - \sum_{\substack{f \, : \, V' \rightarrow  [n]\,,i \in \mathrm{Im}(f) \\ \text{injective}}} \prod_{v \in V'} w_{f(v)} \cdot \Paren{\prod_{\{v,w\} \in E'(M)} \iprod{x_{f(v)}, x_{f(w)}}}^{\subseteq f(S)} \,.
    \end{align*}
    Next, using that in the double sum $\sum_{i \in [n]} \sum_{\substack{f \, : \, V' \rightarrow  [n]\,,i \in \mathrm{Im}(f) \\ \text{injective}}}$ every mapping $f$ is counted exactly $\card{V'}$ times, we obtain
    \begin{align*}
      &\sum_{i \in [n]} \sum_{\substack{f \, : \, V' \rightarrow  [n]\,,i \in \mathrm{Im}(f) \\ \text{injective}}} \prod_{v \in V'} w_{f(v)} \cdot \Paren{\prod_{\{v,w\} \in E'(M)} \iprod{x_{f(v)}, x_{f(w)}}}^{\subseteq f(S)} \\
      &= \card{V'} \cdot \sum_{\substack{f \, : \, V' \rightarrow  [n] \\ \text{injective}}} \prod_{v \in V'} w_{f(v)} \cdot \Paren{\prod_{\{v,w\} \in E'(M)} \iprod{x_{f(v)}, x_{f(w)}}}^{\subseteq f(S)} \\
      &= \card{V'} \cdot P_{G'(M)}^{\sse f(S)} \,.
    \end{align*}
    Putting everything together, we obtain
    \begin{align*}
      \Set{w_i^2 = w_i \,, \sum_{i=1}^n w_i = \eta n} \proves_{O(1)}^{w} P_G^{\subseteq S}(w,x) = O(d)^{s(u)} \cdot (\eta n - |V'|) \cdot \sum_{M \in \cM(u)} P_{G'(M)}^{\subseteq S} \,.
    \end{align*}
    Finally, noting that there are at most $\deg(u)! \leq \deg(u)^{O(\deg(u))}$ matchings and applying \cref{lem:union-of-cycles-merge} finishes the proof.
\end{proof}

\section{Transfer Results: Certificates for Subgaussian Data}
\label{sec:transfer-lemma}

In this section, we show how the methodology from \cite{diakonikolas2024sos} can be used to transfer the certificates on Gaussian data to subgaussian data.
All of the ideas in the present section are already present in \cite{diakonikolas2024sos}.
At the same time, we will also be able to derive high-probability concentration bounds for Gaussian data.
We first transfer the operator norm resilience and sparse eigenvalue certificates (cf.\ \cref{thm:resilience_Gauss_full,thm:sparse_sing_val_full}).
At the end of this section, we also show how to transfer the alternate form of the operator norm resilience certificates (cf~\cref{thm:resilience_Gauss_full_alternate}) since this will be necessary for our covariance estimation application.
\begin{theorem}[Adapted from \cite{diakonikolas2024sos}]
\label{thm:Gaussian-cov-resilience}
Let $M$ be $\R^{n\times d}$ random matrix such that $\E[M] = 0$ and the entries of $M$ are jointly $s$-subgaussian, i.e., for all $A \in \R^{n\times d}$, $\P( \sum_{i,j} A_{i,j}M_{i,j} > s t \|A\|_F) \leq 2\exp(-t^2)$.\footnote{In particular, $M$ could be composed of independent rows/columns of (centered) $s$-subgaussian vectors.}  
Let $\cC$ be a fixed set of pseudoexpectations over variables $u$, $v$, and $w$ that satisfy $\sum_i w_i^2u_i^2=1$ and $\|v\|_2^2=1$.

Then, with probability $1-\delta$,
\begin{align*}
 \max_{\pE \in \cC} \pE\left[ \frac{1}{\sqrt{n}}\sum_{i,j} M_{i,j} w_iu_iv_j  \right] \lesssim s \cdot \E_{G} \left[ \frac{1}{\sqrt{n}} \max_{\pE \in \cC} \pE\left[ \sum_{i,j} G_{i,j} w_iu_iv_j  \right]\right] + s \sqrt{\frac{\log(1/\delta)}{n}}, 
\end{align*}
where $G$ is a matrix with i.i.d.\ $\cN(0,1)$ entries.

\end{theorem}
\begin{proof}
Let $R$ denote the quantity $\E_{G} \left[ \frac{1}{\sqrt{n}} \max_{\pE \in \cC} \pE\left[ \sum_{i,j} G_{i,j} u_iv_j  \right]\right]$.
Let $\cT$ denote the set of matrices
\[
    \left\{A \in \R^{n \times d}: A_{i,j} = \frac{1}{\sqrt{n}}\pE[w_iu_i v_j] \text{ for } \pE \in \cC \right\} \,.
\]
The left-hand side is thus equivalent to $\sup_{A \in \cT} \langle M, A \rangle$, where $M$ is jointly $s$-subgaussian.
Applying Gaussian comparison theorem (a consequence of Talagrands's generic chaining~\cite{Talagrand21}) as in \cite{diakonikolas2024sos}, we get that with probability $1  - \delta$,
  \begin{align*}
      \sup_{\pE \in \cC} \pE\left[\frac{1}{\sqrt{n}}u^\top M v\right] &= \sup_{T \in \cA} \langle M, T \rangle \\
      &\lesssim s\E_G\left[\sup_{T \in \cA} \langle G, T \rangle\right] + s \sup_{A \in \cT}\|A\|_\fr \sqrt{\log(1/\delta)}\\
      &\lesssim s R + s \sqrt{\frac{\log(1/\delta)}{n}},
  \end{align*}
where the last line uses SoS-Cauchy Schwarz inequality to get $\|A\|_\fr^2 = \frac{1}{n} \sum_{i,j} \left(\pE[w_iu_i v_j]\right)^2 \leq \frac{1}{n} \sum_{i,j} \left(\pE[w_i^2u_i^2]\right) \left(\pE[v_j^2]\right) = \frac{1}{n}$.  
   
\end{proof}
Applying the result above we obtain the following SoS certificate for the \textsf{SSV} problem for subgaussian data.
\begin{corollary}[High-probability $\eta$-\textsf{SSV} certificates for subgaussian data]
\label{cor:high-prob-ssv-subgaussian}
Let $M$ be $\R^{n\times d}$ random matrix with $n\geq d$ such that $\E[M] = 0$ and the entries of $M$ are jointly $s$-subgaussian.
Let $\cA$ denote the constraints
\begin{align*}
    \cA \coloneqq \Set{w_i^2 = w_i  \, \forall i\in[n]\,, \sum_{i=1}^n w_i \leq \eta n \,, w_iu_i = u_i \,\forall i\in[n] \,, \norm{u}^2 = 1, \norm{v}^2 = 1}\,.
\end{align*}
Then, with probability $1 - \delta$ over $M$, there exists a sum of squares proof of the following inequality: 
\begin{align*}
 \cA \proves_{O(p)}^{w,u,v} \frac{1}{\sqrt{n}} u^\top M v  \leq O(s) \cdot \left( \sqrt{\sosBound}
 + 
  \sqrt{\frac{\log(1/\delta)}{n}} \right) \,.
\end{align*}

\end{corollary}
\begin{proof}
    Let $\cC$ denote the set of degree-$m$ pseudoexpectations that satisfy the constraints $\cA$.
    Using the duality between SoS proofs and pseudoexpectations (\Cref{fact:sos-duality}),
    we see that  $\cA \proves_{m}^{u,v,w}  \frac{1}{\sqrt{n}} u^\top M v  \leq B$ if and only if $\sup_{\pE \in \cC} \pE[\tfrac{1}{\sqrt{n}}u^\top M v] \leq B$.
    Applying \Cref{thm:Gaussian-cov-resilience},
    it suffices to show that 
    $$\E\left[\sup_{\pE \in \cC} \pE[\tfrac{1}{\sqrt{n}}u^\top G v]\right] \leq  \sqrt{\sosBound} 
    ,$$ where $G$ is a matrix of \iid $\cN(0,1)$ entries.
    Leveraging the duality (\Cref{fact:sos-duality}) again, the above conclusion holds from  \Cref{thm:sparse_sing_val_full}.
 \end{proof}
Next, we obtain SoS certificates for the resilience of operator norm of subgaussian data.
\begin{corollary}[High-probability operator norm resilience of subgaussian data]
\label{cor:high-prob-sos-resilience-subgaussian}
Let $M$ be $\R^{n\times d}$ random matrix with $n\geq d$ such that $\E[M] = 0$ and the entries of $M$ are jointly $s$-subgaussian.
Let $M_1,\dots,M_n$ denote the rows of $M$.
Let $\cB$ denote the constraints
\begin{align*}
    \cB \coloneqq \Set{w_i^2 = w_i \forall i \in [n]\,, \sum_{i=1}^n w_i \leq \eta n \,,  \norm{v}^2 = 1}\,.
\end{align*}
Then, with probability $1 - \delta$ over $M$, there exists a sum of squares proof of the following inequality: 
\begin{align*}
 \cB \proves_{O(p)}^{w,v} \frac{1}{n} \sum_{i=1}^n w_i \langle M_i, v \rangle^2  \leq O(s^2)\cdot \left ( \sosBound + {\frac{\log(1/\delta)}{n}}\right) \,.
\end{align*}
\end{corollary}
\begin{proof}
Let $\cC'$ denote the pseudoexpectations over $w$ and $v$ that satisfy $\cB$.
By duality (\Cref{fact:sos-duality}), 
it suffices to show that, with probability $1-\delta$,
$\sup_{\pE \in \cC'} \pE'\left[\tfrac{1}{n}\sum_{i=1}^n w_i \langle M_i, v \rangle^2\right] \leq \sosBound^2 + s^2 \frac{\log(1/\delta)}{n}$.
Let $\cC$ denote the pseudoexpectations $w$, $v$, and $u = (u_1,\dots,u_n)$ that satisfy $\cB \cup \{u_iw_i = u_i \forall i \in [n], \|u\|_2^2=1\}$.
We claim that the following equality holds for any matrix $M$:
\begin{align}
    \sup_{\pE \in \cC'} \pE\left[\frac{1}{n}\sum_{i=1}^n w_i \langle M_i, v \rangle^2\right] = \left(\sup_{\pE \in \cC} \pE\left[ \sum_{i=1}^n \tfrac{1}{\sqrt{n}} w_i u_i \langle M_i, v \rangle \right]\right)^2 \,.
    \label{eq:linearized-sos}
\end{align}
Deferring the proof of this equality below, we proceed with the remainder of the proof.
The equation \Cref{eq:linearized-sos} gives an equivalent ``linear'' (in data) representation, permitting the use of \Cref{thm:Gaussian-cov-resilience}.   
Applying \Cref{eq:linearized-sos} twice and using \Cref{thm:Gaussian-cov-resilience}, we get that with probability $1-\delta$ over $M$,
\begin{align*}
    \sup_{\pE \in \cC'} \pE\left[\sum_{i=1}^n w_i \langle M_i, v \rangle^2\right] 
    &= \left(\sup_{\pE \in \cC} \pE\left[\sum_{i=1}^n \tfrac{1}{\sqrt{n}} w_i u_i \langle M_i, v \rangle\right]\right)^2 \\
    &\lesssim s^2 \left(\E\left[\sup_{\pE \in \cC} \pE\left[\sum_{i=1}^n \tfrac{1}{\sqrt{n}} w_i u_i \langle G_i, v \rangle \right]\right]\right)^2 + \frac{s^2 \log(1/\delta)}{n} \\
    &\lesssim s^2 \left(\E\left[\sup_{\pE \in \cC'} \pE\left[\tfrac{1}{\sqrt{n}}\sum_{i=1}^n w_i \langle G_i, v \rangle^2\right]\right] \right) + \frac{s^2 \log(1/\delta)}{n} \\
    &\lesssim s^2 \sosBound + \frac{s^2 \log(1/\delta)}{n}.
\end{align*}
To finish the proof, it remains to prove \Cref{eq:linearized-sos}.
Given any $\pE \in \cC$ such that $ \sum_i\pE[w_i\langle M_i, v\rangle^2] > 0$; otherwise, the left hand side in \Cref{eq:linearized-sos} is 0 and the right hand side is nonnegative.
consider the extension $\pE'$ over $v, u,$ and $w$ such that assigns $u_i := \frac{w_i \langle M_i , v\rangle }{\sqrt{\pE[\sum_i w_i \langle M_i,v\rangle^2 ]}}$, i.e., for any polynomial $p$,  $\pE'[p(v,w,u_1,\dots,u_n)]$ is defined to be $  \pE\left[p\left(v,w, \frac{w_1 \langle M_1 , v\rangle}{\sqrt{\pE[\sum_i w_i \langle M_i,v\rangle^2 ]}},\dots,\frac{w_n \langle M_n , v\rangle}{\sqrt{\pE[\sum_i w_i \langle M_i,v\rangle^2 ]}} \right)\right]$.
The extended pseudoexpecation $\pE'$ belongs to $\cC$, and the right hand side is equal to the left hand side in \Cref{eq:linearized-sos}.
The claim in \Cref{eq:linearized-sos} follows by taking a supremum over $\cC'$.

\end{proof}

\paragraph{Transfering alternate form of operator norm resilience certificates.}
We now show to extend the alternate formulation in \Cref{thm:resilience_Gauss_full_alternate}
to subgaussian data (and obtain high-probability estimates).
We will again use the methodology of~\cite{diakonikolas2024sos}.
We use the following version of H\"older's Inequality for pseudo-expectations.
\begin{fact}
    \label{fact:pe_hoelder}
    Let $m$ be even and let $\pE$ be a degree-$m$ pseudo-expectations over variables $g_i, h_i$.
    Then,
    \[ 
        \pE \Brac{\tfrac 1 n \sum_{i=1}^n g_i^{m-1} h_i} \leq \pE \Brac{\tfrac 1 n \sum_{i=1}^n g_i^m}^{1-1/m}\pE \Brac{\tfrac 1 n \sum_{i=1}^n h_i^m}^{1/m} \,.
    \]
\end{fact}

Our main theorem is the following.
\begin{theorem}
    \label{thm:op_resilience_alternate_subgaussian}
    Let $n,p \in \N, \eta \in [0,1], \delta > 0$ be such that $\eta n$ is an integer and $p \geq 2$ is a power of 2.
    Let $M \in \R^{n \times d}$ be a random matrix with $n\geq d$ such that $\E M = 0$ and the entries of $M$ are jointly $s$-subgaussian.
    Let $M_1, \ldots, M_n$ denote the rows of $M$.
    Let $\cB$ denote the set of system of constraints
    \[
        \cB \coloneqq \Set{\forall i \in [n] \colon w_i^2 = w_i \,, \sum_{i=1}^n w_i = \eta n}\,.
    \]
    Then, with probability at least $1-\delta$, there exists a sum-of-squares proof of the following inequality 
    \begin{align*}
        \cB \proves_{O(p)}^{w, v} \Paren{\tfrac 1 n \sum_{i=1}^n w_i \cdot \iprod{M_i, v}^2}^p \leq &(O(s^2))^p \cdot \biggl[ \sosBound^p 
        +\Paren{\frac {\log(1/\delta)} n}^p \biggr] \cdot \norm{v}^{2p} \,.
    \end{align*}
\end{theorem}

\begin{proof}
    Let $\cC'$ be the set of degree-$O(p)$ pseudo-expectations that satisfy the system $\cB$.
    By strong duality for the system $\mathcal{B}$ (\cref{fact:strong_duality_alternate}), it is enough to show that with probability $1 - \delta$ over $M$:
    \[
         \sup_{\pE \in \cC'} \frac{\pE\Brac{\Paren{\tfrac 1 n \sum_{i=1}^n w_i \cdot \iprod{M_i, v}^2}^p}} {\pE \Brac{\norm{v}^{2p}}} \leq (O(s^2))^p \Paren{\sosBound^p + \Paren{\frac {\log(1/\delta)} n}^p} \,.
    \]
    Let $\cC$ be the set of pseudo-expectations over variables $w, v$ and (additional) scalar variable $r$ satisfying $\cB \cup \Set{r \geq 0}$.
    Similar to the previous proofs, we will be able to show that
    \begin{align}
        \label{eq:equiv_one}
        \sup_{\pE \in \cC'} \frac{\pE\Brac{\Paren{\tfrac 1 n \sum_{i=1}^n w_i \cdot \iprod{M_i, v}^2}^p}} {\pE \Brac{\norm{v}^{2p}}} = \Paren{ \sup_{\pE \in \cC}  \frac{\pE\Brac{\Paren{\tfrac 1 n \sum_{i=1}^n w_i \cdot \iprod{M_i, v}^2} \cdot r^{p-1}}}  {\Paren{\pE \Brac{\norm{v}^{2p}}}^{1/p} \cdot \Paren{\pE r^p}^{(p-1)/p}}  }^p \,.
    \end{align}
    Further, let $\cC^*$ be the set of pseudo-expectations over $w_i, v, r$ and additional variables $u_1,\dots,u_n$ satisfying $\cB \cup \Set{ r \geq 0}$, then we will soon show that
    \begin{align}
        \nonumber
        &\Paren{ \sup_{\pE \in \cC}  \frac{\pE\Brac{\Paren{\tfrac 1 n \sum_{i=1}^n w_i \cdot \iprod{M_i, v}^2} \cdot r^{p-1}}}  {\Paren{\pE \Brac{\norm{v}^{2p}}}^{1/p} \cdot \Paren{\pE r^p}^{(p-1)/p}}  }^p \\
        &\qquad= \Paren{ \sup_{\pE \in \cC}  \frac{\pE\Brac{\Paren{\tfrac 1 n \sum_{i=1}^n w_i \cdot u_i \iprod{M_i, v}} \cdot r^{p-1}}}  {\Paren{\pE \Brac{\norm{v}^{2p}}}^{1/(2p)} \cdot \Paren{\pE r^p}^{(p-1)/(2p)} \cdot \Paren{\pE \Brac{\tfrac 1 n \sum_{i=1}^n u_i^2 r^{p-1}}}^{1/2} } }^{2p}
        \label{eq:equiv_two}
    \end{align}
    As before, combining \Cref{eq:equiv_two} with the the comparison inequalities between gaussian and subgaussian processes (similar to \cite{diakonikolas2024sos} in \Cref{thm:Gaussian-cov-resilience}), 
    we have the following conclusion (here, $G \sim \cN(0,1)^{n\times d}$ with rows $G_1,\dots,G_n$):
    with probability $1-\delta$:
    \begin{align*}
        &\sup_{\pE \in \cC'} \frac{\pE\Brac{\Paren{\tfrac 1 n \sum_{i=1}^n w_i \cdot \iprod{M_i, v}^2}^p}} {\pE \Brac{\norm{v}^{2p}}} \\
        &\leq O(s)^{2p} \Paren{ \E_{G} \sup_{\pE \in \cC}  \frac{\pE\Brac{\Paren{\tfrac 1 n \sum_{i=1}^n w_i \cdot u_i \iprod{G_i, v}} \cdot r^{p-1}}}  {\Paren{\pE \Brac{\norm{v}^{2p}}}^{1/(2p)} \cdot \Paren{\pE r^p}^{(p-1)/(2p)} \cdot \Paren{\pE \Brac{\tfrac 1 n \sum_{i=1}^n u_i^2 r^{2(p-1)}}}^{1/2} } }^{2p} \\
        &\qquad + O(s)^{2p} \log(1/\delta)^p \cdot \Paren{\mathrm{diam}(T)}^{2p}\,,
    \end{align*}
    where $T$ is the induced indexed set of the canonical process, defined in \Cref{eq:definition-T} and $\mathrm{diam}(T) = \sup_{t \in T}\|t\|_2$.

    We can then again use \Cref{eq:equiv_one,eq:equiv_two} to upper bound the first term on the right-hand side (without the $O(s)^{2p}$):
    \begin{align*}
        &\Paren{ \E_{G} \sup_{\pE \in \cC}  \frac{\pE\Brac{\Paren{\tfrac 1 n \sum_{i=1}^n w_i \cdot u_i \iprod{G_i, v}} \cdot r^{p-1}}}  {\Paren{\pE \Brac{\norm{v}^{2p}}}^{1/(2p)} \cdot \Paren{\pE r^p}^{(p-1)/(2p)} \cdot \Paren{\pE \Brac{\tfrac 1 n \sum_{i=1}^n u_i^2 r^{2(p-1)}}}^{1/2} } }^{2p} \\
        &\qquad = \E_{G} \sup_{\pE \in \cC'} \frac{\pE\Brac{\Paren{\tfrac 1 n \sum_{i=1}^n w_i \cdot \iprod{G_i, v}^2}^p}} {\pE \Brac{\norm{v}^{2p}}}\,.
    \end{align*}
    But by strong duality (cf.~\cref{fact:strong_duality_alternate}) and~\cref{thm:resilience_Gauss_full_alternate}, this term is at most $\sosBound^p$.
    Note that the constraint system $\cB$ has a constraint of the form $\sum_{i=1}^n w_i = \eta n$, which is stronger than the one $\sum_{i=1}^n w_i \leq \eta n$ used in~\cref{thm:resilience_Gauss_full_alternate}.
    This stronger form simplifies the arguments for strong duality.

    Thus, it only remains to show that $\mathrm{diam(T)} \leq 1/\sqrt{n}$.
    For $\pE \in \cC^*$, let
    \[
        c(\pE) = \Paren{\pE \Brac{\norm{v}^{2p}}}^{1/(2p)}  \Paren{\pE r^p}^{(p-1)/(2p)}  \Paren{\pE \Brac{\tfrac 1 n \sum_{i=1}^n u_i^2 r^{2(p-1)}}}^{1/2} \,.
    \]
    Then, formally the set $T$ is equal to 
    \begin{align}
    \label{eq:definition-T}
        \Set{\Paren{\frac 1 {c(\pE) \cdot n} \pE \Brac{w_i u_i r^{p-1} v}}_{i=1}^n \,,  \mid \pE \in \cC^* \,,} \,.
    \end{align}

    Consider any element $t \in T$, then (shortening $c(\pE)$ to $c$)
    \begin{align*}
        \norm{t}^2 = \frac 1 {c^2 n^2} \sum_{i=1}^n \sum_{j=1}^n \Paren{\pE w_i u_i r^{p-1} v_j}^2 \,.
    \end{align*}
    Now since all $\pE \in \cC^*$ satisfy $\Set{r \geq 0}$, it follows from a variant of the classical Cauchy-Schwarz for pseudo-expectations (cf.~\cref{fact:modified_pE_CS} and $w_i \leq 1$), that 
    \begin{align*}
        \sum_{i=1}^n \sum_{j=1}^n \Paren{\pE w_i u_i r^{p-1} v_j}^2 &\leq \Paren{\sum_{i=1}^n \pE r^{p-1} u_i^2} \Paren{\sum_{j=1}^d \pE v_j^2 r^{p-1}} = \Paren{\sum_{i=1}^n \pE r^{p-1} u_i^2} \pE r^{p-1} \norm{v}^2 \\
        &\leq \Paren{\pE \sum_{i=1}^n r^{p-1} u_i^2} \Paren{\pE r^{p}}^{1-1/p} \Paren{ \pE \norm{v}^{2p} }^{1/p} = c^2 \cdot n \,.
    \end{align*}
    Thus $\norm{t}^2 \leq 1/n$ and $\mathrm{diam}(T) \leq 1 / \sqrt{n}$.

    We next justify~\cref{eq:equiv_one,eq:equiv_two}.
    We start with~\cref{eq:equiv_one}.
    By H\"older's Inequality for pseudo-expectations (cf.~\cref{fact:pe_hoelder}) it holds that
    \[
        \pE\Brac{\Paren{\tfrac 1 n \sum_{i=1}^n w_i \cdot \iprod{x_i, v}^2} \cdot r^{p-1}} \leq \Paren{ \pE \Brac{r^p} }^{(p-1)/p} \cdot \Paren{\pE\Brac{\Paren{\tfrac 1 n \sum_{i=1}^n w_i \cdot \iprod{x_i, v}^2}^p}}^{1/p}\,.
    \]
    Thus, the right-hand side is at most the left-hand side.
    Picking $r = \tfrac 1 n \sum_{i=1}^n w_i \cdot \iprod{x_i,v}^2 \geq 0$ achieves this inequality.
    We can justify~\cref{eq:equiv_two} similarly:
    By the same variant of the classical Cauchy-Schwarz for pseudo-expectations (cf.~\cref{fact:modified_pE_CS}, this uses that $\pE$ satisfies $\Set{r \geq 0}$), it holds that
    \[
        \pE\Brac{\Paren{\tfrac 1 n \sum_{i=1}^n w_i \cdot u_i \iprod{x_i, v}}  r^{p-1}} \leq \Paren{\pE\Brac{\Paren{\tfrac 1 n \sum_{i=1}^n w_i \iprod{x_i, v}^2}  r^{p-1} } }^{1/2} \Paren{\pE\Brac{ \Paren{\tfrac 1 n \sum_{i=1}^n u_i^2 }  r^{p-1}}}^{1/2} \,.
    \]
    And picking $u_i = w_i \iprod{x_i,v}$ shows that the equality can be achieved (since $w_i^2 = w_i$).
\end{proof}

\section{Applications: Robust Statistics, Subspace Distortion, and 2-to-\texorpdfstring{$p$}{p} Norms}
\label{sec:applications}
In this section, we apply the sparse singular value and operator norm resilience certificates obtained in the previous section to problems in robust statistics and to certifying upper bounds on the distortion of random subspaces and the 2-to-$p$ norm of random operators.
In particular, we give the applications to robust covariance and (both covariance-aware and Euclidean) mean estimation in~\cref{sec:cov_est_full}, the application to subspace distortion in~\cref{sec:distortion}, and the application to certifying 2-to-$p$ norms of random operators in~\cref{sec:2_to_p_full}.

\subsection{Robust Covariance and Covariance-Aware Mean Estimation}
\label{sec:cov_est_full}

In this section, we will derive our results for robust covariance estimation as well as (covariance-aware and Euclidean) mean estimation.
In particular, proving~\cref{thm:cov-intro,thm:mean-aware-intro,thm:main-intro}.
They are a direct consequence of our new resilience certificate and ideas implicit in prior works.
We repeat the full argument here for completeness.

\paragraph{Background.} 
In \emph{robust covariance estimation}, the goal is to estimate 
the covariance matrix $\Sigma \in \R^{d \times d}$ of an unknown probability distribution $\cD$ on $\R^d$ from $\eta$-corupted samples.
Without some assumptions on $\cD$, robust covariance estimation is information-theoretically impossible. A standard assumption in the literature is that 
$\cD$ is (hypercontractive)-subgaussian. 
For the subgaussian setting, using $n \geq d$ samples, 
error $\alpha = O(\eta \log(1/\eta) + \sqrt{d/n})$ is information-theoretically achievable (for estimating in the relative spectral norm) and best possible up to constants.
Note that by taking $n = \tilde{O}(d/\eta^2)$, one can make the error 
as small as $O(\eta \log(1/\eta))$.
In contrast, for $\alpha \gtrsim \sqrt{\eta}$ (the regime of interest in this paper), the computationally efficient algorithm from \cite{KotSS18} requires a larger sample complexity of $\Omega(\eta d^2/\alpha^2)$ to get error $\alpha$.

\new{Our main algorithm result is the following, nearly matching the low-degree testing lower bounds (\Cref{prop:low-degree-cov-estimation-gaussian}), even for the Gaussian case:}

\begin{theorem}[Robust Covariance Estimation in Relative Spectral Norm; Full Version of~\cref{thm:cov-intro}]
    \label{thm:cov_est_full}
    Let $n,d,p \in \N$ such that $p\geq 2$.
    Let $\cX^* = \Set{x_1^*, \ldots, x_n^*}$ be a set of $n$ \iid samples from a mean-zero $s$-hypercontractive-subgaussian distribution with unknown covariance $\Sigma$ in dimension $d$.
    Let $\eta,\alpha > 0$ be such that $\Omega(s^2\sqrt{\eta}) \leq s^2\alpha < 1/10$ and $\eta$ smaller than a sufficiently small absolute constant.
    There exists an algorithm that given $\eta$, $p$, $\alpha$, and an $\eta$-corruption $x_1, \ldots, x_n$ of $\cX^*$ runs in time $n^{O(p)}$, and with probability at least $1-\delta$, outputs an estimate $\hat{\Sigma}$ such that
    \[
        (1-s^2\alpha) \Sigma \preceq \hat{\Sigma} \preceq (1+s^2\alpha) \Sigma \,,
    \]
    as long as $n = \tilde{\Omega}(\tfrac{\eta^2 d^{2+10/p}}{\alpha^4} + \tfrac {d^{1+5/p}} {\alpha^2} + \frac{d}{\eta \alpha} + 
    \frac{\log(1/\delta)}{\alpha^2}) + \frac{\log(1/\delta)}{\eta \alpha}$.
\end{theorem}

\cref{thm:cov_est_full} follows from a known connection between certifying resilience and covariance estimation~\cite{KotSS18}.
We repeat the proof below for convenience.
Let $\cB$ be the following auxiliary constraint system in $n$ scalar-valued variables $\tilde{w}_i$
\begin{equation*}
    \cB = \Set{\tilde{w}_i^2 = \tilde{w}_i \,, \sum_{i=1}^n \tilde{w}_i = 2\eta n} \,.
\end{equation*}
Consider the system of polynomial inequalities $\cA$ below.
We remark that the last constraint in $\cA$ can be formally modelled using polynomial inequalities by introducing auxiliary variables for the coefficients of the SoS proof.
This is a by-now standard technique and we refer the reader to~\cite{KotSS18,fleming2019semialgebraic} for a more detailed discussion.
Our algorithm searches for a pseudo-expectation satisfying $\cA$.%
\footnote{Typically, one uses $w_i$ to denote the ``inlier'' samples instead of $1-w_i$, since this is more convenient in the proofs. We use this notation to be consistent with the notation in previous sections.}
Note that $v$ is a variable in the SoS proof we are searching for in the last constraint.
Let $p'$ be the power of two that is closest to and larger than $p$.
Note that $p' \leq 2p$.
\begin{equation*}
    \cA = 
    \left\{
        \begin{aligned}
            &\forall i \in [n] \colon  & w_i^2 &= w_i \,, \\
            &&\sum_{i=1}^n w_i &\leq \eta n \,, \\
            && (1-w_i) x_i' &= (1-w_i) x_i \,, \\
            && \Sigma' &= \tfrac 1 n \sum_{i=1}^n x_i' (x_i')^\top \\
            &\exists \text{ SoS proof in $(\tilde{w},v)$ that } &\cB \proves_{O(p)}^{\tilde{w},v} \Paren{\tfrac 1 n \sum_{i=1}^n \tilde{w}_i \iprod{x_i', v}^2}^{p'} &\leq \alpha^{p'} \cdot \iprod{v, \Sigma' v}^{p'}\
        \end{aligned}
    \right\} \,.
\end{equation*}
Note that $w_i$ and $\tilde{w}_i$ are distinct variables and that the upper bound on their respective sums is different by design ($\eta n$ versus $2\eta n$).
\begin{mdframed}
    \begin{algorithm}[Robust Covariance Estimation]
      \label{alg:robust_cov_est}\mbox{}
      \begin{description}
      \item[Input:] $x_1, \ldots, x_n$, $\eta$-corruption of $\cX^*$ and corruption level $\eta$.
      \item[Output:] Estimate $\hat{\Sigma}$ of $\Sigma$.
      
      \begin{enumerate}
      \item Find degree-$O(p)$ pseudo-expectation $\pE$ satisfying $\cA$.
      \item Output $\hat{\Sigma} = \pE \Sigma'$.
      \end{enumerate}
      \end{description}
    \end{algorithm}
\end{mdframed}

\cref{thm:cov_est_full} will follow from the following lemma implicit in~\cite{KotSS18}.
We defer the proof to the end of this section.

\begin{definition}
\label{def:niceness-covariance}
Let $\eta \in (0,0.1)$,  $\alpha \in (0,1/2)$, and $p\in \N$ be a power of two. 
We say that a multiset $S = \{y_1,\dots,y_n\}$
is $(\eps,\alpha,p)$-nice with respect to $\Sigma$ if the following holds:
\begin{enumerate}
        \item $\cB \proves_{O(p)}^{\tilde{w},v} \Paren{\tfrac 1 n \sum_{i=1}^n \tilde{w}_i \iprod{\Sigma^{-1/2}y_i, v}^2}^{p} \leq (\alpha/2)^{p} \|v\|^{2p}$.
        \item $(1-\alpha) \tilde\Sigma \preceq \Sigma \preceq (1+ \alpha) \tilde{\Sigma} $, where  $\tilde{\Sigma} = \tfrac 1 n \sum_{i=1}^n y_iy_i^\top$.
\end{enumerate}
    
\end{definition}

\begin{lemma}
\label{lem:deterministic-covariance}
  Let $x_1^*, \ldots, x_n^*$ be $(\eta,\alpha,p)$-nice with respect to $\Sigma$ for $\eta \leq 0.1$, $\alpha \in (0,1/2)$, and an even $p$ that is a power of $2$
    (\Cref{def:niceness-covariance}).
    Let $\cX$ be an $\eta$-corruption of $x^*$'s.
    Consider the system $\cA$ defined above on $\cX$.
    Then, it holds that for every unit vector $u$:
    \begin{enumerate}
        \item $
        \cA \proves_{O(p)} \Iprod{u, \Paren{\Sigma' - \tilde{\Sigma}}u}^2 \leq 12^2\alpha^2 \cdot \Iprod{u, \tilde{\Sigma} u}^2 \,$, where $\tilde{\Sigma} = \tfrac{1}{n}\sum_i x^*_i(x^*_i)^\top$.
        \item $
        \cA \proves_{O(p)} \Iprod{u, \Paren{\Sigma' - \tilde{\Sigma}}u}^2 \leq 24^2\alpha^2 \cdot \Iprod{u, \Sigma u}^2 \,.$
    \end{enumerate}
    In particular, there exists an algorithm that given $\eta$, $\alpha$, $p$, and $\cX$, runs in time $(n+d)^{O(p)}$ and outputs an estimate $\hat{\Sigma}$ such that
        $(1-O(\alpha)) \Sigma \preceq \hat{\Sigma} \preceq (1+O(\alpha)) \Sigma \,$.
    \end{lemma}

With this in hand, we can prove~\cref{thm:cov_est_full} below.
\begin{proof}[Proof of \cref{thm:cov_est_full}]
    Let $p'$ be the power of two that is closest to and larger than $p$.
    Define $\alpha' := s^2 \alpha$.
    Given \Cref{lem:deterministic-covariance},
    it suffices to show that with probability at least $1-\delta$, 
    the points $x^*_1,\dots,x^*_n$ satisfy $(\eta, \alpha', p')$-niceness from \Cref{def:niceness-covariance}.
    
    Towards that goal, we note that $\Sigma^{-1/2} x_i^*$ are \iid samples from an isotropic, mean-zero $s$-subgaussian distribution.
    By \Cref{thm:op_resilience_alternate_subgaussian},
    the first condition in \Cref{def:niceness-covariance} holds with probability $1 -\delta$ if
    $$
     s^2 \left( \max\Set{p \cdot \log^{O(1)}(n) d^{\frac{5}{2p}} \max \Set{\sqrt{\frac{d}{n}}, \frac{\sqrt{\eta d}}{n^{1/4}}}, \sqrt{\eta}}   + \frac{\log(1/\delta)}{n}\right)    \lesssim \alpha'\,.
    $$
    Since $\alpha' \gtrsim s^2 \sqrt{\eta}$,
    the above inequality is satisfied 
    if $n \geq  \tilde{\Omega}(\tfrac{s^8\eta^2 d^{2+10/p}}{\alpha'^4} + \tfrac {s^4d^{1+5/p}} {\alpha'^2} + \tfrac{s^2\log(1/\delta)} \alpha')$. 
    To establish the second condition in \Cref{def:niceness-covariance},
    we use standard concentration results on the empirical covariance matrix of subgaussian data to get that $n = \Omega\left(\frac{s^4 d}{\alpha'^2} + \frac{s^4\log(1/\delta)}{\alpha'^2} \right)$ samples suffice~\cite[Section 4.7]{Vershynin18}; here we use that $\alpha' \leq 1$ to ignore the other terms.
    Since $s \geq 1$, $\alpha' \leq 1$, the desired sample complexity follows by substituting $\alpha'=s^2 \alpha$. 

\end{proof}

\subsubsection{Proof of \Cref{lem:deterministic-covariance}}
As mentioned earlier, the ideas here are implicit in \cite{KotSS18}.

\begin{proof}
    Let $u \in \R^d$ be an arbitrary unit vector.
    Since for $p$ even and any $C > 0$ it holds that $\{X^p \leq C^p\} \proves_p^X X^2 \leq C^2$, (\cref{fact:sos_square_root}), it is enough to show that
    \[
        \Iprod{u, \Paren{\Sigma' - \tilde{\Sigma}} u}^p \leq (12 \alpha)^p \Iprod{u, \tilde{\Sigma}u}^p \,.
    \]
    Let $w_i^* = \Ind(x_i \neq x_i^*)$.
    Note that by definition at most $\eta n$ of the $w_i^*$ are non-zero.
    Further, our constraints imply that $(1-w_i^*)(1 - w_i) x_i'  = (1-w_i^*)(1 - w_i) x_i^*$.
    For convenience define $r_i = (1-w_i^*)(1 - w_i)$.
    Using the definition of $\Sigma'$ and $\tilde{\Sigma}$ and SoS almost triangle inequality (cf.\ \cref{fact:sos_triangle_large_power}) it immediately follows that
    \begin{align*}
        \cA \proves_{O(p)} \Iprod{u, \Paren{\Sigma' - \tilde{\Sigma} } u}^p &= \Paren{\tfrac 1 n \sum_{i=1}^n \iprod{x_i', u}^2 - \tfrac 1 n \sum_{i=1}^n \iprod{x_i^*, u}^2 }^p \\
        &\leq 2^p\Paren{\tfrac 1 n \sum_{i=1}^n (1-r_i)\iprod{x_i', u}^2}^p + 2^p\Paren{\tfrac 1 n \sum_{i=1}^n (1-r_i)\iprod{x_i^*, u}^2 }^p\,.
        \numberthis\label{eq:sos-proof-robust-covariance}
    \end{align*}
    Observe that sum of $1-r_i$ is small as follows:
    \[
        \cA \proves_{O(1)}^{w_i} \sum_{i=1}^n \Paren{1 - r_i} = \sum_{i=1}^n \Paren{w_i^* + w_i - w_i^* w_i} \leq 2 \eta n \,.
    \]
    Without loss of generality, assume that $\cA$ implies that this sum is exactly $2\eta n$ (this can be achieved by including a few more points in the respective sums).
    Then, substituting $\tilde{w}_i = (1-r_i)$ and $v = u$ we can apply the SoS proof that
    \[
        \cB \proves_{O(p)}^{\tilde{w},v} \Paren{\tfrac 1 n \sum_{i=1}^n \tilde{w}_i \iprod{x_i', v}^2}^p \leq \alpha^p \cdot \Iprod{v, \Sigma' v}^p
    \]
    and conclude that the first term in \Cref{eq:sos-proof-robust-covariance} is at most $(2\alpha)^p \Iprod{v, \Sigma' v}^p$.

    Further, by assumption on $x^*$'s, we have that
    (i)
    $\cB \proves_{O(p)}^{\tilde{w},v} \Paren{\tfrac 1 n \sum_{i=1}^n \tilde{w}_i \iprod{x_i^*, v}^2}^p \leq (\alpha/2)^p \iprod{v, {\Sigma}v }^p$, which follows by the substitution $v \mapsto \Sigma^{1/2}v$ 
    and (ii) $\cB \proves_{O(p)}^{v}  (v^\top \Sigma v)^p \leq 2^p (v^\top \tilde{\Sigma} v)^p$; the claim in (ii) follows by the inequalities $\Sigma \preceq (1+ \alpha) \tilde{\Sigma}$ and  $\alpha \le 1$
    and the fact that the linear matrix inequalities have sum of squares proof for quadratic forms.
   Hence, the second term in \Cref{eq:sos-proof-robust-covariance} is at most $\alpha^p \iprod{v, \tilde\Sigma v}^p$.
    Taken together, we have shown that (using again SoS almost triangle inequality)
    \begin{align*}
        \cA &\proves_{O(p)} \Iprod{u, \Paren{\Sigma' - \tilde{\Sigma}} u}^p \leq (2\alpha)^p \Paren{\Iprod{v, \Sigma' v}^p + \Iprod{v, \tilde{\Sigma} v}^p} \\
        &\leq (6 \alpha)^p \Iprod{v, \tilde{\Sigma} v}^p + (4\alpha)^p\Iprod{u, \Paren{\Sigma' - \tilde{\Sigma}} u}^p \,,
    \end{align*}
    which implies the first claim by rearranging and the fact that $6 \alpha/(1 - 4 \alpha) \leq 12 \alpha$ since $\alpha \leq 0.1$.
    The second claim follows by the first claim and the second condition in \Cref{lem:deterministic-covariance}.

    \paragraph{Feasibility and accuracy.}
    We begin by showing that the constraint system $\cA$ is feasible.
    Indeed, we select $w_i = \Ind(x_i \neq x_i^*)$ and $x_i' = x_i^*$ and $\Sigma' = \tilde{\Sigma} = \tfrac 1 n \sum_{i=1}^n x_i^* (x_i^*)^\top$.
    Then all of the constraints are satisfied, where the last constraint uses that $\cX^*$ satisfies \Cref{def:niceness-covariance} (both the first and the second conditions as discussed above in the proof).
    Recall that the algorithm outputs $\widehat{\Sigma} = \pE \Sigma'$, where $\pE$ is any pseudoexpectation satisfying $\cA$; such a $\pE$ must exist because $\cA$ is satisfiable.
    Establishing the desired accuracy guarantee  is equivalent to showing that for all unit vectors $u \in \R^d$ it holds that $\Iprod{u, (\hat{\Sigma} - \tilde{\Sigma}) u}^2 \leq 72\alpha^2 \cdot \Iprod{u, \tilde{\Sigma} u}^2$.
    However, this follows by Cauchy-Schwarz for pseudo-expectations and the second claim above.

    We now discuss the runtime.
    Note that both the number of variables and constraint in $\cA$ is polynomial in $(n+d)^p$ (this uses that we can succinctly represent the last constraint).
    Thus, by~\cref{fact:pE_efficient_optimization} we can compute a pseudo-expectation $\pE$ satisfying $\cA$ in time $(n+d)^{O(p)}$.
\end{proof}

\subsubsection{Robust Covariance-Aware Mean Estimation}
\label{sec:cov_mean_est_full}

We will show the following results for covariance-aware mean estimation by combining our robust covariance estimator from the previous section and existing estimators for (nearly) isotropic sub-Gaussian distributions.
\begin{theorem}[Robust Covariance-Aware Mean Estimation; Full Version of~\cref{thm:mean-aware-intro,thm:mean-intro}]
    \label{thm:cov_mean_est_full}
    Let $n,d,p \in \N$.
    Let $\cX^* = \Set{x_1^*, \ldots, x_n^*}$ be a set of $n$ \iid samples from an $s$-hypercontractive-subgaussian distribution with unknown mean $\mu$ and unknown covariance $\Sigma$ in dimension $d$.
    Let $\eta,\alpha > 0$ be such that $\Omega(s^2\sqrt{\eta}) \leq s^2\alpha < 1/10$ and $\eta$ smaller than a sufficiently small absolute constant.
    There exists an algorithm that given $\eta$ and an $\eta$-corruption $x_1, \ldots, x_n$ of $\cX^*$ runs in time $n^{O(p)}$ and with probability at least $1-\delta$ outputs an estimate $\hat{\mu}$ such that
        
        $$\Norm{\Sigma^{-1/2}(\hat{\mu} - \mu)} \leq O(s\new{\sqrt{\eta \alpha }})\, \qquad \text{and} \qquad \Norm{(\hat{\mu} - \mu)} \leq O\left(s\new{\sqrt{\eta \alpha \|\Sigma\|_\op}} \right)\,$$
    as long as 
    $n = \tilde{\Omega}(\tfrac{\eta^2 d^{2+10/p}}{\alpha^4} + \tfrac {d^{1+5/p}} {\alpha^2} + 
    \frac{d}{\eta \alpha} + 
    \frac{\log(1/\delta)}{\alpha \min(\eta,\alpha)})$.
\end{theorem}

The above result improves on \cite{KotSS18} in terms of sample complexity: the algorithm in \cite{KotSS18}, coupled with \cite{DiaKan22-book}, requires $\frac{\eta d^2}{\alpha^2} + \frac{d}{\eta \alpha}$ samples even for $s=O(1)$ and $\delta = 0.1$.
On the other hand, \Cref{thm:cov_mean_est_full} uses much fewer (roughly) $\frac{\eta^2 d^2}{\alpha^4} + \frac{d}{\eta \alpha}$ samples. Observe that $\frac{d}{\eta \alpha}$ samples are needed even information-theoretically,
and the computational lower bounds presented in \Cref{sec:low-degree-cov-mean-est} show that for $\alpha = \Theta(1)$,  $\eta^2d^2$ many samples are needed.

We now collect necessary definitions that we shall use, starting with the notion of \emph{stability}~\cite{DiaKKLMS19-siam,SteCV18} as well as known algorithmic guarantees for robust mean estimation for (nearly) isotropic subgaussian distributions.
We use the following defintion from~\cite[Definition 2.1]{DiaKan22-book}
\begin{definition}[$(\eta,\delta)$-stability]
    Let $0 < \eta < 1/2$ be such that $\delta \geq \eta$.
    Let $\cX = \Set{x_1^*, \ldots, x_n^*} \sse \R^d$ and $\mu \in \R^d$.
    For a set $S \sse [n]$, define $\mu_S = \tfrac 1 {\card{S}} \sum_{i \in S} x_i^*$ and $\Sigma_S = \tfrac 1 {\card{S}} \sum_{i \in S} (x_i^* - \mu)(x_i^* - \mu)^\top$. 
    The set  $\cX$ is called $(\eta,\delta)$-stable with respect to $\mu$, if for all $S \sse [n]$ of size at least $(1-\eta) n$ it holds that (i)
         $\norm{\mu_S - \mu} \leq \delta$,
        and (ii) $\norm{\Sigma_S - I_d}_\op \leq \delta^2/\eta$.
\end{definition}
We further use the following facts (cf.~\cite[Proposition 3.3 and Theorem 2.11]{DiaKan22-book}; see also \cite{DiaKKLMS19-siam}).
\begin{fact}[Stability of Isotropic Subgaussians]
    \label{fact:stab_iso_sub_g}
    Let $0 < \eta < 1/2$ and $\gamma \geq 0$.
    Let $\cX = \Set{x_1^*, \ldots, x_n^*}$ a set of $n$ \iid samples from an isotropic $s$-subgaussian distribution with mean $\mu$.
    Then, with probability at least $1-\exp(-\gamma^2 n)$, $\cX$ is $(\eta, \delta = s\eta \sqrt{\log(1/\eta)} + s\gamma)$ stable with respect to $\mu$ as long as $n \geq \Omega(d/\gamma^2)$.
\end{fact}
\begin{fact}[Stability Implies Robust Mean Estimation]
    \label{fact:stability_implies_robust_est}
    Let $0 < \eta < 1/2, \eta \leq \delta $.
    And let $\cX = \Set{x_1^*, \ldots, x_n^*}$ be an $(\eta, \delta)$-stable set with respect to some vector $\mu$.
    Then, there exists an algorithm $\cA_{\mathsf{robust\_mean}}$ that takes as input an $\eta$-corruption of $\cX$ and $\eta$, runs in time $n^{O(1)}$ and outputs a vector $\hat{\mu}$ such that $\norm{\hat{\mu} - \mu} \leq O(\delta)$.
\end{fact}

We will need the following straightforward extension of~\cref{fact:stab_iso_sub_g} to distributions which are only nearly isotropic.
\begin{fact}[Stability of Nearly Isotropic Subgaussians]
    \label{fact:stab_nearly_iso_sub_g}
    Let $0 < \eta < 1/2, \gamma < 1, s\geq 1, s^2\alpha \leq 1/4$.
    Let $\cX^{**} = \Set{x_1^{**}, \ldots, x_n^{**}}$ be a set of $n$ \iid samples from an isotropic $s$-subgaussian distribution with mean $\mu^{**}$.
    Then, with probability at least $1-\exp(-\gamma^2 n)$, for all matrices $\Pi$ with $\|\Pi^\top \Pi- I\|_\op\leq s^2\alpha$, the set $\cX^* = \Set{\Pi x_1^{**}, \ldots, \Pi x_n^{**}}$ is $(\eta, O(\delta'))$-stable with respect to $\Pi\mu^{**}$ for $\delta' = s \eta \sqrt{\log(1/\eta)} + s\gamma + s\sqrt{\eta \alpha}$ and   $n \geq \Omega(d/\max(\gamma^2, \eta \alpha))$.
\end{fact}
\begin{proof}
    The set $\cX^{**}$ consists of $n$ \iid samples from an isotropic 1-subgaussian distribution with mean $\mu^{**}$.
    Let $\delta \lesssim s\left(\eta\sqrt{\log(1/\eta)} + \gamma \right)$.
    By~\cref{fact:stab_iso_sub_g}, $\cX^{**}$ is $(\eta, \delta)$-stable with respect to $\mu^{**}$ with probability at least $1-\exp(-\gamma^2n)$.
    We condition on this event.
    
    Fix a $\Pi$ satisfying the condition in the lemma and define $x_i^* = \Pi x_i^{**}$ and $\mu^{*} = \Pi \mu^{**}$.
    Let $S \sse [n]$ be an arbitrary set of size at least $(1-\eta)n$.
    Define
    \begin{align*}
        \mu_S^* &= \tfrac 1 {\card{S}} \sum_{i \in S} x_i^*\,, & \Sigma_S^* &= \tfrac 1 {\card{S}} \sum_{i \in S} (x_i^* - \mu)(x_i^* - \mu)^\top \,,\\
        \mu_S^{**} &= \tfrac 1 {\card{S}} \sum_{i \in S} x_i^{**}\,, & \Sigma_S^{**} &= \tfrac 1 {\card{S}} \sum_{i \in S} (x_i^{**} - \mu^{**})(x_i^{**} - \mu^{**})^\top \,.
    \end{align*}
    Note that $\mu_S^{*} = \Pi \mu_S^{**} $ and $\Sigma_S^{*} = \Pi  \Sigma_S^{**} \Pi^\top$.
    By stability of $\cX^{**}$ we know that $\norm{\mu_S^{**} - \mu^{**}} \leq \delta$.
    Thus,
    \[
        \norm{\mu_S^{*} - \mu^*} = \norm{\Pi(\mu_S^{**} - \mu^{**})} \leq \norm{\Pi}_\op \cdot \norm{\mu_S^{**} - \mu^{**}} \leq O(\delta) \,,
    \]
    where we use that  $\|\Pi\|_\op$ is in $[0.5,1.5]$.
    Further, we know that $\norm{\Sigma_S^{**} - I_d} \leq \delta^2/\eta$.
    Thus, 
    \begin{align*}
        \norm{\Sigma_S^{*} - I_d} &= 
        \norm{\Pi\Sigma_S^{**}\Pi^\top - I_d}_\op
        \leq \norm{\Pi\Sigma_S^{**}\Pi^\top - \Pi  \Pi^\top}_\op+ \norm{\Pi\Pi^\top - I}_\op
        \\
        &\leq \norm{\Pi}_\op^2 \norm{\Sigma_S^{**} - I}_\op  + s^2\alpha\\
        &= O(1)\cdot \frac{\delta^2 + \eta O(s^2\alpha))} \eta \leq \frac{O(\delta + s\sqrt{\eta \alpha})^2}  \eta \,.
    \end{align*}
\end{proof}

We can now describe our algorithm to show~\cref{thm:cov_mean_est_full}.
In what follows, $\cX^* = \Set{x_1^*, \ldots, x_n^*}$ is an \iid sample from a 1-subgaussian distribution with unknown mean $\mu$ and unknown covariance $\Sigma$.
For simplicity, we assume that $n/2$ is an even integer. 
\begin{mdframed}
    \begin{algorithm}[Robust Covariance-Aware Mean Estimation]
      \label{alg:robust_cov_mean_est}\mbox{}
      \begin{description}
      \item[Input:] $\cX = \{x_1, \ldots, x_n\}$, which is $\eta$-corruption of $\cX^*$ and corruption level $\eta$.
      \item[Output:] Estimate $\hat{\mu}$ of $\mu$ in Mahalanobis norm.
      
      \begin{enumerate}
      \item Let $\cX_{\mathsf{symmetrized}} = \Set{\tfrac 1 {\sqrt{2}} (x_i - x_{i+1}) \mid i = 1, \ldots, n/2, i\text{ odd}}$
      \item Run~\cref{alg:robust_cov_est} on input $\cX_{\mathsf{symmetrized}}$ and $\eta$ to obtain $\hat{\Sigma}$.
      \item Compute $\tilde{\cX} = \Set{\hat{\Sigma}^{-1/2} x \mid x \in \cX}$.
      \item Run $\cA_{\mathsf{robust\_mean}}$ on input $\tilde{\cX}$ and $\eta$ to obtain $\hat{\mu}'$.
      \item Output $\hat{\mu} = \hat{\Sigma}^{1/2}\hat{\mu}'$.
      \end{enumerate}
      \end{description}
    \end{algorithm}
\end{mdframed}

We can now prove~\cref{thm:cov_mean_est_full}
\begin{proof}
    Let $\cX^*_{\mathsf{symmetrized}}$ be obtained from $\cX^*$ in the same way that $\cX_{\mathsf{symmetrized}}$ is obtained from the input.
    Note that $\cX^*_{\mathsf{symmetrized}}$ consists of $n/2$ \iid samples from a mean-zero $s$-hypercontractive-subgaussian distribution with covariance $\Sigma$.
    Further, $\cX^*_{\mathsf{symmetrized}}$ is a $2\eta$-corruption of $\cX_{\mathsf{symmetrized}}$.
    It follows from~\cref{thm:cov_est_full} that, with probability at least $1-\delta$, the estimate $\hat{\Sigma}$ obtained in Step 3 satisfies $(1-s^2\alpha) \Sigma \preceq \hat{\Sigma} \preceq (1+s^2\alpha) \Sigma$.
    Note that this also implies that $(1-O(s^2\alpha)) \Sigma^{-1} \preceq \hat{\Sigma}^{-1} \preceq (1+O(s^2\alpha)) \Sigma^{-1}$.
    Let this be the event $\mathcal{E}$.

    Let $x^{**}_i = \Sigma^{-1/2} x^*_i$ for all $i \in [n]$ and define $\cX^{**}$ to be the analogous set.
    Further define, $\tilde{x}_i^* := \hat{\Sigma}^{-1/2}x_i^* = \hat{\Sigma}^{-1/2}\Sigma^{1/2}x_i^* = \Pi x^{**}_i$ and define $\tilde{\cX}^* = \Pi  \cX^{**}$ to be the analogous set, where $\Pi:=\hat{\Sigma}^{-1/2}\Sigma^{1/2}$.
    On the event $\cE$, we have that $\|\Pi\Pi^\top - I\|_\op \leq O(s^2\alpha)$.
    
    The set $\cX^{**}$ consists of $n$ \iid samples from an isotropic $s$-subgaussian distribution with mean $\mu^{**}$. 
    Applying \Cref{fact:stab_nearly_iso_sub_g},
    with the fact that $\|\Pi \Pi^\top - I \|_\op \leq O(s^2\alpha)$,
    we get that, with probability $1-\delta$, the set $\tilde{\cX}^*$
    is $(\eta, \delta')$-stable with respect to $\tilde{\mu} = \Pi \mu^{**} = \hat{\Sigma}^{-1/2} \mu$
    for $\delta' = s \eta \sqrt{\log(1/\eta)} + \sqrt{\eta s^2 \alpha} + s \sqrt{\frac{\log(1/\delta)}{n}}$, where we use that $n \geq d/(\eta \alpha)$. We denote  this event by $\cE'$.
    The proposed algorithm will be correct on the event $\cE \cap \cE'$, which holds with probability $1-2 \delta$ by a union bound.
    
    Since $\alpha \geq \eta$ and $n$ large enough, we have that $\delta' = s O( \sqrt{\eta \alpha } + \sqrt{\frac{\log(1/\delta)}{n}})$, which is $O(s \sqrt{\eta \alpha})$ under the assumptions on $n$.
    It follows by~\Cref{fact:stability_implies_robust_est} that the output $\hat{\mu}'$ of $\cA_{\mathsf{robust\_mean}}$ in Step 5 satisfies $\norm{\hat{\mu}' - \tilde{\mu}} \leq O( s \new{\sqrt{\eta \alpha}})$.
    Since $s^2\alpha \leq 1/4$,
    we have that $x^\top \Sigma^{-1} x = \Theta(1) x^\top \hat{\Sigma}^{-1}x$ for all $x$.
    Hence,
    \begin{align*}
        \Norm{\Sigma^{-1/2} \Paren{\hat{\mu} - \mu}}^2 &= \Paren{\hat{\mu} - \mu}^\top \Sigma^{-1} \Paren{\hat{\mu} - \mu} \leq O(1) \cdot \Paren{\hat{\mu} - \mu}^\top \hat{\Sigma}^{-1} \Paren{\hat{\mu} - \mu} \\
        &=\Norm{\hat{\Sigma}^{-1/2} \Paren{\hat{\mu} - \mu}}^2 = \Norm{\hat{\mu}' - \tilde{\mu}}^2 \leq O(\delta') \,,
    \end{align*}
    which implies the desired error guarantee.
    The conclusion for the Euclidean norm also follows from this guarantee.
    
\end{proof}

\subsection{Distortion of Random Subspaces}
\label{sec:distortion}

In this section we use our results to show new certificates for the distortion of random subspaces
In this setting, our certificates are (nearly) optimal in terms of the relation of $n$ (ambient dimension), $d$ (subspace dimension) and $\Delta$ (target distortion to be certified), as evidenced by low-degree polynomial lower bounds~\cite{mao2021optimal,chen2022well} (see the discussion at the end of this section).
Recall the definition of the distortion of a subspace $X \sse \R^n$ (cf.~\cref{def:distortion}).
The distortion of $X$, denoted by $\Delta = \Delta(X)$ is defined as
\[
        \Delta(X) \coloneqq \max_{x\in X ,\,x \neq 0} \frac{\sqrt{n} \norm{x}_2}{\norm{x}_1} \,.
\]
The following related definition will be useful for us.
It is taken from~\cite[Defnition 2.10]{guruswami2010almost}.
\begin{definition}
    \label{def:well_spreadness}
    Let $t \in [n], \e > 0$.
    A subspace $X \sse \R^n$ is called \emph{$(t,\e)$-well-spread} if for all $x \in X$ and every set $S \sse [n]$ such that $\card{S} \leq t$, it holds that $\norm{x_{S^c}}_2 \geq \e \norm{x}_2$.\footnote{For a set $T$, we use the notation $x_T$ to denote the $\card{T}$-dimensional vector consisting of the coordinates of $x$ indexed by $T$.} 
\end{definition}
It is known that well-spread subspaces have low-distortion.
In particular, if a subspace $X$ is $(t,\e)$-well-spread, it has distortion at most $\tfrac{\sqrt{n}}{\sqrt{t} \e^2}$ (cf.~\cite[Lemma 2.11]{guruswami2010almost}).
As we will see, our certificates for resilience, immediately certify an upper bound on the well-spreadness of subspaces.

We can now state our main theorem.
All subspaces will be represented as the columnspan of matrices.
\begin{theorem}[Full Version of~\cref{thm:distortion-us}]
    \label{thm:distortion_full}
    Let $n,d,p \in \N$ and $1 \leq \Delta \leq \sqrt{n}$ such that $d \leq n^{1-5/(p+5)}/\poly(\log n)$.
    There exists an algorithm $\cA$ that takes as input a matrix $X \in \R^{n\times d}$, runs in time $n^{O(p)}$, and outputs a real number $\tau$ such that
    \begin{itemize}
        \item The distortion of the columnspan of $X$ is at most $\tau$ (always).
        \item If the entries of $X$ are \iid $\cN(0,1)$\footnote{Using~\cref{cor:high-prob-sos-resilience-subgaussian}, we could also replace this by the condition that the entries of $X$ are mean-zero jointly subgaussian.}, then $\tau \leq \tilde{O}(\frac{d^{1/2 + 5/(2p)}}{n^{1/4}} )$ with probability at least $1-e^{-d^{0.01}}$.
    \end{itemize}
\end{theorem}
In particular, our algorithm certifies that a random $d$-dimensional subspace has distortion at most $\Delta$ as long as $d \leq \tilde{O}(\tfrac{\sqrt{n} \Delta^2} {(\sqrt{n} \Delta^2)^{5/(2p)}})$.
Taking $p =O(\log n)$ we can, in quasi-polynomial time, certify distortion $\Delta$ as long as $d \leq \tilde{O} (\sqrt{n} \Delta^2)$.

The algorithm we use is the following.
\begin{mdframed}
    \begin{algorithm}[Certifying Low Distortion of Random Subspaces]
      \label{alg:cert_distortion}\mbox{}
      \begin{description}
      \item[Input:] A matrix $X \in \R^{n \times d}$ with rows $X_1, \ldots, X_n$.
      \item[Output:] $\tau \in \R_{\geq 0}$, upper bound on the distortion of the columnspan of $X$.
      
      \begin{enumerate}
      \item Compute $\norm{X^\top X - n \cdot I_d}$. If this is larger than $n/5$, output $\tau = \sqrt{n}$ and terminate. Else, continue.
      \item Compute $\eta^*$ as the largest $\eta \in \{1/n, 2/n, \ldots, 1\}$ such that there exists an SoS proof, in scalar-valued variables $w_1, \ldots, w_n$ and vector-valued variable $v$ (in $d$ dimensions), that
        \[
            \Set{ w_i^2 = w_i \,, \sum_{i=1}^n w_i \leq \eta n \,, \norm{v}^2 = 1} \proves_{O(p)}^{w_i, v} \sum_{i=1}^n w_i \iprod{X_i, v}^2 \leq  n/2 \,.
        \]
        Output $\tau = \sqrt{n}$ and terminate if no such $\eta$ exists.
        \item Output $\tau = 4/\sqrt{\eta^*}$.
      \end{enumerate}
      \end{description}
    \end{algorithm}
\end{mdframed}

\begin{proof}[Proof of~\cref{thm:distortion_full}]
    We first argue that~\cref{alg:cert_distortion} runs in time $n^{O(p)}$.
    Indeed, this follows since we can decide the existence of an the SoS proofs in Step 2 in time $n^{O(p)}$.
    
    We next show soundness, i.e., that $\tau$ is always an upper bound on the distortion.
    In case we output $\tau = \sqrt{n}$, this is trivial since the distortion can be at most $\sqrt{n}$.
    Let $\eta^*$ be the largest value in $\{1/n, \ldots, 1\}$ such that there is an SoS proof that 
    \[
        \Set{ w_i^2 = w_i \,, \sum_{i=1}^n w_i \leq \eta^* n \,, \norm{v}^2 = 1} \proves_{O(p)}^{w_i, v} \sum_{i=1}^n w_i \iprod{X_i, v}^2 \leq  n/2 \,.
    \]
    To argue the remaining cases, we will show that the columnspan of $X$ is always $(\eta^* n, \tfrac 1 2)$-well-spread (we remark that the choice of $1/2$ is arbitrary here).
    By~\cite[Lemma 2.11]{guruswami2010almost} this implies that the distortion is at most $\tfrac{4\sqrt{n}}{ \sqrt{\eta^* n}} = 4/\sqrt{\eta^*} = \tau$.
    Let $x = X v$, for some arbitrary unit vector $v \in \R^d$.
    Let $S \sse [n]$ be the set of the $\eta^* n$ largest coordinates of $x$.
    Then, it is enough to show that $\norm{x_{S^c}}_2 \geq \tfrac 1 2 \norm{x}_2$.
    Which is equivalent to showing that $\norm{x_S}_2^2 \leq 3/4 \norm{x}_2^2$.
    Note that by the existence of the SoS proof and since $S$ has size $\eta^* n$, we know that $\norm{x_S}_2^2 = \sum_{i \in S} \iprod{X_i, v}^2 \leq n/2$.
    Further, by the condition checked in Step 1, it follows that $\norm{x}^2 = v^\top (X^\top X - n \cdot I_d) v + n \geq (4/5) n$.
    Thus, $3/4 \norm{x}^2 \geq (3/5) n \geq n/2 = \norm{x_S}_2^2$.
    It follows that $\tau$ always is a valid upper bound on the distortion.

    Next, assume that the entries of $X$ are \iid $\cN(0,1)$.
    We show that $\tau \leq \tilde{O}(\frac{d^{1/2 + 5/(2p)}}{n^{1/4}} )$ with probability at least $1-e^{-d^{0.01}}$.
    For simplicity, we show it when the failure probability is $2e^{-d^{0.01}}$.
    First, by standard concentration bounds (cf., e.g.~\cite[Theorem 4.6.1 and Equation 4.22]{Vershynin18}) it holds that $\norm{X^\top X - n \cdot I_d} \leq O(d) \leq n/5$ with probability at least $1-2\exp(-n)$.
    Thus, the condition in Step 1 holds with that probability. 
    
    Lastly, we will show that with probability at least $1-e^{-d^{0.001}}$ there exists an SoS proof that
    \[
        \Set{ w_i^2 = w_i \,, \sum_{i=1}^n w_i \leq \eta^* n \,, \norm{v}^2 = 1} \proves_{O(p)}^{w_i, v} \sum_{i=1}^n w_i \iprod{X_i, v}^2 \leq  n/2 \,.
    \]
    for $\eta^* = \tilde{\Omega}(\sqrt{n}/d^{1+5/p})$ (rounded to the nearest multiple of $1/n$).
    Indeed, by~\cref{thm:resilience_Gauss_full} we know that, with this probability, we can certify an upper bound of
    \[
        (p \cdot \log n)^{O(1)} \cdot d^{5/(2p)} \cdot \max\Set{\sqrt{\frac d n}, \frac{\sqrt{\eta^* d}}{n^{1/4}}} \cdot n + O(1) \log(1/\delta)\,.
    \]
    Since $d \leq n^{1-5/(p+5)}/\poly(\log n)$, or equivalently, $n \geq \tilde{\Omega}(d^{1+5/p})$, the first term in the maximum is at most $d^{-5/(2p)}/\poly(\log n)$.
    Similarly, for our choice of $\eta^*$, the second term in the maximum is at most $d^{-5/(2p)}/\poly(\log n)$ as well.
    Thus, the whole expression is at most $n/2$.
\end{proof}

\paragraph{Low-degree hardness for obtaining better certificates.}

A natural question is whether our certificates achieve the optimal trade-off between $d,n$, and $\Delta$.
Existing lower bounds, in the framework of low-degree polynomials, suggest that this is indeed the case.\footnote{When compared to the guarantees of the certificates we obtain in quasi-polynomial time we leave it as an open question to give a certification algorithm achieving the same trade-off that runs in polynomial time. This is not ruled out by the lower bounds below.}
See~\cref{sec:low-degree-lower-bounds} for more background on this framework and formal definitions.
For a distinguishing problem, the low-degree conjectures posits that if the so-called \emph{degree-$D$} likelihood ratio stays bounded for $D = \Omega(\polylog n))$, then no polynomial-time algorithm exists.
The prior works of \cite[Theorem 4.5]{mao2021optimal} (see also \cite[Theorem 5.6]{chen2022well}) have shown the following:\footnote{The results in \cite{mao2021optimal,chen2022well} are phrased in the language of the ``sparse vector in a subspace'' problem. We have restated them in our language, see the discussion after the theorem.}
\begin{theorem}
    \label{thm:low-degree_distortion}
    Let $d,n \in \N, \Delta \in [\log^{10} n,\sqrt{n}]$ be such that $d = \tilde{\Omega}(\sqrt{n} \Delta^2)$.
    Given as input a matrix $X \in \R^{n \times d}$, the degree-$\log^{\Omega(1)}(n)$ likelihood ratio of the following distinguishing problem stays bounded.
    \begin{enumerate}
        \item The entries of $X$ are \iid distributed as $\cN(0,1)$.
        \item The columnspan of $X$ has distortion at least $\Omega(\Delta)$ with high probability.
    \end{enumerate}
\end{theorem}

In particular, the subspace in the first case has (true) distortion $O(1)$.
Yet, any certification algorithm that certifies an upper bound better than $\Omega(\Delta)$ would solve the distinguishing problem.

The subspace constructed in the second distribution contains a vector $v \in \R^n$, referred to as the ``planted sparse vector'', whose entries are \iid distributed as
\[
    v_i = \begin{dcases}
        0 \,, &\quad \text{with probability $1-\tfrac 1 {\Delta^2}$,} \\
        N\Paren{0,\tfrac {\Delta^2} { n}}\,, &\quad \text{with probability $\tfrac 1 {\Delta^2}$.}
    \end{dcases}
\]
It follows from standard concentration bounds that with high probability $\norm{v}_2 \approx 1$ and
\[
    \norm{v}_1 \approx \tfrac \Delta {\sqrt{n}} \cdot (\text{\#non-zero coordinates of $v$}) \approx \tfrac {\sqrt{n}} \Delta \,.
\]
Thus, the distortion of $X$ is at least $\Omega(\sqrt{n}/(\tfrac {\sqrt{n}} \Delta)) = \Omega(\Delta)$, with high probability.

\subsection{\texorpdfstring{$2 \rightarrow p$}{2-to-p} Norm of Random Matrices}
\label{sec:2_to_p_full}
For a vector $x$, we use $\|x\|_p$ to denote the $\ell_p$-norm of $x$; recall that $\|x\|$ denotes the Euclidean norm.
\begin{definition}[$2 \rightarrow p$ norm]
  Let $A \in \R^{d \times n}$.
  For $p \geq 1$, we define the $2 \rightarrow p$ norm of $A$ as
  \[
  \|A\|_{2 \rightarrow p} = \max_{x \neq 0} \frac{ (\E_{i \sim [n]} |\iprod{A_i,x}|^p)^{1/p}}{ \|x\|} = n^{-1/p} \cdot \max_{x \neq 0} \frac{ \|Ax\|_p}{\|x\|} \, .
  \]
  Here $A_1,\ldots,A_n$ are the rows of the matrix $A$.
\end{definition}

The factor of $n^{-1/p}$ is merely a notational convenience.
Standard concentration bounds readily show that with high probability $\|A\|_{2 \rightarrow p} \leq O(1 + d^{p/2}/n)^{1/p} \approx \sqrt{d} \cdot n^{-1/p}$ for $n$ small enough (the regime of interest in this paper).

\subsubsection{State of the Art}
The following fact captures the state of the art for polynomial-time certifiable bounds on the $2 \rightarrow p$ norm of a random matrix.

\begin{fact}[\cite{barak2012hypercontractivity}]
\label{fact:barak-two-to-four}
  There is a polynomial-time algorithm which takes a $d \times n$ matrix $A$ and outputs an upper bound on $\|A\|_{2 \rightarrow 4}$.
  With high probability over $A \sim \cN(0,1)^{d \times n}$, this bound is at most $O(1+d^2/n)^{1/4}$.
\end{fact}

Using the Riesz–Thorin theorem, this $2 \rightarrow 4$ bound, combined with the fact that with high probability $\E_{i \sim [n]} \iprod{A_i,x}^2 \approx 1$ for all unit $x$, implies a certification algorithm for the $2 \rightarrow p$ norm for any $p \in [2,4]$, which we capture in the following corollary.

\begin{corollary}[Corollary of \Cref{fact:barak-two-to-four}]
\label{cor:2-to-p-sota}
    Fix $p \in [2,4]$.
    There is a polynomial-time algorithm which takes a $d \times n$ matrix $A$ and outputs an upper bound on $\|A\|_{2 \rightarrow p}$.
    If $n \gg d$, with high probability over $A \sim \cN(0,1)^{d \times n}$, this bound is at most $O(1+d^2/n)^{1/2 - 1/p}$.
\end{corollary}

Except if $p=4$ or if $n = \Omega(d^2)$, the bound $(1 + d^2/n)^{1/2 - 1/p}$ is asymptotically larger than the true value of $\|A\|_{2 \rightarrow p}$.
Our result brings the certifiable bound asymptotically closer to the true value of $\|A\|_{2 \rightarrow p}$ for every $p \in (2,4)$.

\subsubsection{Our Result and Proof}
We improve on the bound $(1 + d^2/n)^{1/2 - 1/p}$ from \Cref{cor:2-to-p-sota} using \Cref{thm:sparse_sing_val_full}, by showing:
\begin{theorem}
\label{thm:2-to-p-us}
    Fix $p \in [2,4]$. For every $\eps > 0$, there is a polynomial-time algorithm which takes a random matrix $A \sim \cN(0,1)^{d \times n}$ and outputs an upper bound on $\|A\|_{2 \rightarrow p}$.
    If $n \gg d$, this bound is with high probability at most
    \begin{align}
    \label{eq:2-to-p-cert-value-us}
    \tilde{O} \Paren{ 1+ \frac {d^{2+\eps}}{n}}^{\frac 14 - \frac 1 {2p}} + \tilde{O} \Paren{ \sqrt{d} \cdot n^{-1/p} } \, .
    \end{align}
    If $\eps = 0$, the bound can be certified in quasipolynomial time.
\end{theorem}

Note that the second term in the bound of \Cref{thm:2-to-p-us} appears in the true value of the $2 \rightarrow p$ norm, up to logarithmic factors.
Up to $d^{\eps}$, the first term is the square root of the previous state-of-the-art certifiable bound.

\begin{proof}[Proof of \Cref{thm:2-to-p-us}]
    We prove the $\eps > 0$ case; the $\eps = 0$ case follows the same argument.
    By standard concentration of measure combined with \Cref{thm:resilience_Gauss_full}, for every $\eps > 0$ there is a polynomial-time algorithm which with high probability over $A$ certifies that the following all hold:
    \begin{enumerate}
        \item $|\iprod{A_i,A_j}| \leq \tilde{O}(\sqrt d)$ for all $i \neq j$
        \item $\|A_i\|^2 \leq \tilde{O}(d)$ for all $i$
        \item $\|A\|_{2 \rightarrow 2} \leq O(1)$
        \item For every $\eta \in [0,1]$, for every $S \subseteq [n]$ with $|S| = \eta n$, and for every unit vector $v \in \R^d$, it holds that
        \[
        \frac 1 n \sum_{i \in S} \iprod{A_i,v}^2 \leq \tilde{O} \Paren{ \frac{\sqrt{\eta} \cdot d^{1/2 + \eps}}{n^{1/4}} + d^{\eps} \cdot \sqrt{\frac d n }} \, .
        \]
    \end{enumerate}
    Hence, it will suffice to show that for every $A \in \R^{d \times n}$ satisfying the above conditions, $\|A\|_{2 \rightarrow p}$ obeys the bound in \Cref{eq:2-to-p-cert-value-us}.
    We start by developing several facts which hold for any such $A$.

    First of all, by Markov's inequality and the fact that $\|A\|_{2 \rightarrow 2} \leq O(1)$, for any unit vector $v$ and any $t \geq 0$,
    \begin{align}
    \label{eq:2-to-p-proof-1}
    \Pr_{i \sim [n]} ( \iprod{A_i,v}^2 \geq t ) \leq O \Paren{ \frac 1 {t^2} } \, .
    \end{align}

    Second, for any $t \gg \sqrt{d}$,
    \begin{align}
    \label{eq:2-to-p-proof-2}
    \Pr_{i \sim [n]} ( \iprod{A_i,v}^2 \geq t ) \leq  \tilde{O}\Paren{ \frac d {tn}} \, .
    \end{align}
    To see \Cref{eq:2-to-p-proof-2}, note that if $\Pr_{i \sim [n]} ( \iprod{A_i,v}^2 \geq t ) \geq \eta$, then there is a set $S \subseteq [n]$ with $|S| = \eta n$ such that $\sum_{i \in S} \iprod{A_i,v}^2 \geq t \eta n$.
    At the same time, $\sum_{i \in S} \iprod{A_i,v}^2$ is at most the maximum eigenvalue of $\eta n \times \eta n$ matrix whose entries are $\iprod{A_i,A_j}$ for $i,j \in S$.
    This maximum eigenvalue is at most the maximal $\ell_1$ norm of any row of the matrix, which is at most $\tilde{O}(d + \eta n \sqrt{d})$, since we have assumed $\|A_i\|^2 \leq \tilde{O}(d)$ and $|\iprod{A_i,A_j}| \leq \tilde{O}(\sqrt{d})$ if $i \neq j$.
    Putting these inequalities together gives \Cref{eq:2-to-p-proof-2}.

    Third, by assumption, for any $t \geq 0$,
    \begin{align}
    \label{eq:2-to-p-proof-3}
    \Pr_{i \sim [n]} ( \iprod{A_i,v}^2 \geq t ) \leq  \tilde{O}\Paren{ d^{\eps} \cdot \frac d {\sqrt n} \cdot t^{-2} + d^\eps \cdot \sqrt{\frac d n} \cdot t^{-1} } \, .
    \end{align}
    
    Next, we write
    \[
    \|A\|_{2 \rightarrow p}^p = \max_{\|v\|=1} \E_{i \sim [n]} |\iprod{A_i,v}|^p = \max_{\|v\|=1} \int_{0}^\infty \Pr_{i \sim [n]} (|\iprod{A_i,v}|^p \geq t) \, dt \, .
    \]
    Our strategy is to break this integral into several pieces and use different bounds on each piece.
    In particular,
    \[
    \int_{0}^\infty = \int_0^1 + \int_{1}^{(d/\sqrt{n})^{p/2}} + \int_{(d/\sqrt{n})^{p/2}}^{d^{p/4}} + \int_{d^{p/4}}^{\tilde{O}(d^{p/2})} + \int_{\tilde{O}(d^{p/2})}^\infty \, .
    \]
    Since we are integrating a probability, the integrand is always at most $1$, so the first of the integrals will be at most $1$.
    And, by assumption, for every $i$ we have $\iprod{A_i,v}^p \leq \|A_i\|^p \leq \tilde{O}(d^{p/2})$, so the last of the integrals is $0$.
    Dropping the $dt$ for compactness, we get
    \begin{align*}
    \int_{0}^\infty \Pr_{i \sim [n]} (|\iprod{A_i,v}|^p \geq t)  & \leq 1 +
    \underbrace{\int_{1}^{(d/\sqrt{n})^{p/2}} (\cdots)}_{A} + \underbrace{\int_{(d/\sqrt{n})^{p/2}}^{d^{p/4}} (\cdots)}_{B} + \underbrace{\int_{d^{p/4}}^{\tilde{O}(d^{p/2})} (\cdots)}_{C} \;,
    \end{align*}
    where $(\cdots) = \Pr_{i \sim [n]} (|\iprod{A_i,v}|^p \geq t) \, dt$.

    Now we need to bound the three integrals $A,B,C$.
    We start with $A$.
    By \Cref{eq:2-to-p-proof-1},
    \begin{align*}
    \int_1^{(d/\sqrt{n})^{p/2} }\Pr_{i \sim [n]} (|\iprod{A_i,v}|^p \geq t) \, dt &\leq O(1) \cdot \int_{1}^{(d/\sqrt{n})^{p/2} } t^{-2/p} \, dt \leq O(1+d/\sqrt{n})^{(p/2)(1-2/p)} \\
    &= O(1) + O(d/\sqrt{n})^{p/2 -1}\, .
    \end{align*}
    Now we use \Cref{eq:2-to-p-proof-3} to bound $B$.
    Since we are using \Cref{eq:2-to-p-proof-3} in the regime that the term $d^{\eps} \cdot \tfrac d {\sqrt n} \cdot t^{-2}$ dominates, we obtain
    \begin{align*}
    \int_{(d/\sqrt{n})^{p/2}}^{d^{p/4}} \Pr_{i \sim [n]} (|\iprod{A_i,v}|^2 \geq t^{2/p}) & \leq \tilde{O} \Paren{ d^{\eps} \cdot \frac d {\sqrt n} \cdot \int_{(d/\sqrt{n})^{p/2}}^{d^{p/4}} t^{-4/p}}  \\
    & \leq \tilde{O} \Paren{d^{\eps} \cdot \frac d {\sqrt n} \cdot (d/\sqrt{n})^{(p/2)(1-4/p)}} \\
    & \leq \tilde{O} \Paren{d^{\eps} \cdot \Paren{\frac d {\sqrt n}}^{p/2 -1}} \, .
    \end{align*}

    Next, we bound $C$ using \Cref{eq:2-to-p-proof-2}, obtaining
    \begin{align*}
        \int_{d^{p/4}}^{\tilde{O}(d^{p/2})} \Pr_{i \sim [n]} (|\iprod{A_i,v}|^2 \geq t^{2/p}) \leq \tilde{O} \Paren{ \frac d n \cdot \int_{d^{p/4}}^{\tilde{O}(d^{p/2})} t^{-2/p} }
        \leq \tilde{O} \Paren{ \frac d n \cdot d^{(p/2)(1-2/p)}}  \, .
    \end{align*}

    All together, we have obtained:
    \[
     \int_{0}^\infty \Pr_{i \sim [n]} (|\iprod{A_i,v}|^p \geq t) \leq O(1) + d^{\eps} \cdot \tilde{O} \Paren{ \frac d {\sqrt{n}}}^{p/2-1}  +  \tilde{O} \Paren{ \frac d n \cdot d^{(p/2)(1-2/p)}}  \, .
    \]

    This gives us
     \[
     \int_{0}^\infty \Pr_{i \sim [n]} (|\iprod{A_i,v}|^p \geq t) \leq \tilde{O} \Paren { d^{\eps p / 4} \cdot \Paren{\frac d {\sqrt n} }^{p/4} + \frac d n \cdot d^{p/2 - 1}}  \, .
    \]
    Taking the $p$-th root finishes the proof.
\end{proof}

\subsubsection{Sparse Principal Component Analysis}
In sparse principal component analysis, or \emph{sparse PCA}, the goal is to find one or more sparse directions $v \in \R^n$ which have large variance with respect to a high-dimensional dataset $X_1,\ldots,X_N \in \R^n$.
We may assume that $\E_{i \sim [n]} X_i = 0$, in which case this problem is mathematically identical to finding the maximizer(s) in \Cref{prob:sparse-sing-vect} where $M \in \R^{N \times n}$ has rows $X_1,\ldots,X_N$.\footnote{Here we have switched notation to use ``$N$'' in place of ``$d$'' to avoid the risk of confusion from using ``$d$'' as the number of samples in a dataset rather than its dimension.}

Since sparse PCA is hard to approximate for worst-case datasets \cite{chan2016approximability}, most recent work focuses on average-case variants.
Average-case studies of sparse PCA typically involve data sampled from the following \emph{sparse spiked covariance} model.

\begin{definition}[Sparse spiked covariance \cite{johnstone2009consistency}]
  To draw $N$ samples from the $\eta$-sparse spiked covariance model in $n$ dimensions with signal-to-noise ratio $\beta > 0$, first draw a random unit vector $v$ conditioned on having $\eta n$ nonzero coordinates, then draw $X_1,\ldots,X_N \overset{\text{\iid}}{\sim} \cN(0, I + \beta vv^\top)$.
\end{definition}

There are several flavors of average-case sparse PCA, in roughly increasing order of difficulty:
\begin{itemize}[leftmargin=*]
\item \textbf{Testing:} Distinguish $N$ samples from the sparse spiked covariance model from $N$ samples drawn \iid from $\cN(0,I)$.
\item \textbf{Recovery (a.k.a. estimation):} Given $N$ samples from the sparse spiked covariance model, find a unit vector $\hat{v}$ such that $\iprod{v,\hat{v}}^2 \geq 1-o(1)$.
\item \textbf{Certification:} Given $N$ samples from $\cN(0,I)$, output a certificate that no $\eta$-sparse unit vector $v$ achieves $\tfrac 1 N \sum_{i \leq N} \iprod{X_i,v}^2 \geq \Omega(\beta)$. 
The certification variant is simply \Cref{prob:sparse-sing-vect} in the Gaussian matrix case.
\end{itemize}

\noindent For all three variants, the main algorithmic question is: for which $\beta = \beta(N,n,\eta)$ can the problem be solved in polynomial time?

\paragraph{Relationship among variants.}
Algorithms for recovery and certification can in particular be used for testing, while recovery and certification are typically thought to be incomparable.
However, SoS-based certification algorithms can be used not only to solve the recovery problem, but also to solve a robust form of it, where an adversary is allowed to make entrywise-bounded perturbations to $X_1,\ldots,X_N$ \cite{d2020sparse}.

\paragraph{Our Setting: Moderate samples, moderate sparsity.}
The theoretical literature on sparse PCA is extensive, and we do not attempt a full survey.
Instead, we narrow our attention to the range of $\eta,n,$ and $d$ where our algorithm improves on the state of the art.
First of all, the best known algorithms for sparse PCA differ depending on whether $\eta \ll 1/\sqrt{n}$ or $\eta \gg 1/\sqrt{n}$.
In the former ``very sparse'' case, entrywise thresholding applied to the sample covariance matrix is a useful tool for identifying the nonzero coordinates of $v$; this technique is used in the \emph{diagonal thresholding} and \emph{covariance thresholding} algorithms of \cite{johnstone2009consistency,deshp2016sparse}, which both correspond roughly to \cref{thm:cert-inner-products}.
We restrict our attention to this ``moderate sparsity'' case, where $1 \gg \eta \gg 1/\sqrt{n}$.
Additionally, we restrict attention to the ``moderate samples'' case, where $\sqrt{n}\ll N \ll n$, as the algorithmic landscape of sparse PCA also changes when $N \gg n$ or $N \ll \sqrt{n}$.

\paragraph{Summary of prior work.}  
In this moderate sparsity, moderate samples, a polynomial-time algorithm is known for the certification problem only when $\beta \gg \sqrt{n \eta}$; this algorithm corresponds to \Cref{thm:t-th-moment-cert}.
However, polynomial-time algorithms for the testing and recovery variants are known which allow smaller $\beta$, namely, $\beta \gg \sqrt{\eta} n^{3/4} / \sqrt{N}$.
(One may check that the latter improves on $\sqrt{n \eta}$ in the moderate-samples moderate-sparsity regime.)
These algorithms are implicit in \cite{mao2021optimal}.

\paragraph{Consequences of \Cref{thm:main-intro} for sparse PCA.}
We improve the state of the art for sparse PCA certification in the moderate-samples moderate-sparsity regime almost to match the state of the art for testing and recovery.
In particular, \Cref{thm:sparse_sing_val_full}, the quantitative version of \Cref{thm:main-intro}, immediately implies that for every $\eps > 0$ there is a polynomial-time algorithm for the certification variant of sparse PCA whenever $\beta \gg \sqrt{\eta} n^{3/4} / N^{1/2 - \eps}$, and a quasipolynomial time algorithm when $\beta \gg \sqrt{\eta} n^{3/4} / \sqrt{N}$.

In addition to bring the state of the art for certification in line with the easier variants, this has (at least) two significant consequences for recovery and testing.
First, our algorithm expands the range of $\beta$ for which recovery and testing can be performed in polynomial time in the presence of adversarial entry-wise perturbations of $X_1,\ldots,X_N$ \cite{d2020sparse}.
Second, the recovery and testing algorithms of \cite{mao2021optimal} work only for Gaussian distributuions, while our algorithm works whenever $X_1,\ldots,X_N$ are (centered) subgaussian.

\section{Low-Degree Lower Bounds}
\label{sec:low-degree-lower-bounds}
In this section,
we state the computational lower bounds for the family of low-degree polynomial tests.
We first present the necessary background in \Cref{sec:low-degree-lower-bound-background}.
We state robust covariance estimation for Gaussian data in \Cref{sec:low-degree-cov-est-resilience-gaussian},
covariance-aware mean estimation for Gaussian data in \Cref{sec:low-degree-cov-mean-est},
and covariance-aware mean estimation for subgaussian data in \Cref{sec:low-degree-cov-est-subgaussian}.

\subsection{Background on Low-Degree Method \& Non-Gaussian Component Analysis}
\label{sec:low-degree-lower-bound-background}
We will construct hard instances using the Non-Gaussian Component Analysis (NGCA) from \cite{DiaKS17}. 
We begin by defining the high-dimensional hidden direction distribution: 
\begin{definition}[High-Dimensional Hidden Direction Distribution] \label{def:high-dim-distribution}
For a unit vector $v \in \R^d$ and a distribution $A$ on the real line with probability density function $A(x)$, define $P_{A,v}$ to be a distribution over $\R^d$, where $P_{A,v}$ is the product distribution whose orthogonal projection onto the direction of $v$ is $A$, 
and onto the subspace perpendicular to $v$ is the standard $(d{-1})$-dimensional normal distribution. 
That is, the density of $P_{A,v}$ satisfies $P_A(x) := A(v^\top x) \phi_{\bot v}(x)$, where $\phi_{\bot v}(x) = \exp\left(-\|x - (v^\top x)v\|_2^2/2\right)/(2\pi)^{(d-1)/2}$.
\end{definition}
The Non-Gaussian Component Analysis (NGCA) testing problem is defined as follows: 
\begin{problem}[Non-Gaussian Component Analysis]
\label{prob:generic_hypothesis_testing}
Let $A$ be a distribution over $\R$. 
For the sample size $n$ and dimension $d$, we define the following hypothesis testing problem with input $(y_1,\dots,y_n)$ with each $y_i \in \R^d$:
\begin{itemize}
\item $H_0$: $y_1,\dots, y_n$ are sampled i.i.d.\ from $\cN(0,\bI_d)$.
\item $H_1$: First a unit vector $v$ is sampled randomly from $\cS^{d-1}$, and then conditioned on $v$, 
$y_1,\dots, y_n$ are sampled i.i.d.\ from $P_{A,v}$.
\end{itemize}
\end{problem}
The computational lower bounds for NGCA were first developed for the statistical query algorithms in \cite{DiaKS17}.
In this paper, we will focus on the restricted family of low-degree polynomial tests, which are closely related to the statistical query lower bounds~\cite{BreBHLS21}.  
\newcommand{\Adv}{\mathrm{Adv}}
The degree-$D$ advantage of a hypothesis testing problem, defined below, lies at the core of this family.
\begin{definition}[Degree-$D$ Advantage of a Testing Problem]
Consider a testing problem of distinguishing $H_0: y \sim \cQ$ versus $H_1: y\sim \cP$ for input $y \in (\R^d)^n$.
The degree-$D$ advantage is defined as follows:
\begin{align}
\Adv_{\leq D, n} := \max_{f \in \R[y]_{\leq D}} \frac{\E_{\cP}[f(y)]}{\sqrt{E_{\cQ}[f^2(y)]}},
\end{align}
where we take the maximum over all degree-$D$ polynomials $f$ over $nd$-dimensional $y$.
\end{definition}
The low-degree conjecture posits that the hypothesis testing problem is computationally hard to solve if $\Adv_{\leq D, m} = O(1)$.
We refer the reader to \cite{Hopkins-thesis,KunWB19} for more details about this conjecture.
In the rest of this section, we shall show that the degree-$D$ advantage is bounded for various NGCA testing problems.

For the NGCA problem defined in \Cref{prob:generic_hypothesis_testing}, \cite{mao2021optimal} gave an expression that relies only on the (normlaized probabilist's) Hermite coefficients.

\begin{lemma}[Lemma 6.4 in \cite{mao2021optimal} ]
\label{lem:ldlr-Hermite}
Consider the NGCA problem with the hidden distribution $A$,
then the associated degree-$D$ advantage satisfies the following:
\begin{align}
\label{eq:degree-d-advantange-ngca}
\Adv_{\leq D, n}
&\leq \sum_{i=0; i \text{ is even}}^D \left(\frac{i}{d}\right)^{\frac{i}{2}}  \sum_{\cI \in \N^n; |\cI| = i} \prod_{j=1}^n \left(\E_{X \sim A} [h_{\cI_j}(x)]\right)^2 \,.
\end{align}
\end{lemma}
The next result generalizes the arguments in the proof of \cite[Theorem 4.5]{mao2021optimal} in a straightforward manner by replacing the conditions in Lemma 6.6 therein by a more general condition below.
\begin{lemma}[Degree-$D$ Advantage and Growth of Hermite Coefficients; Implicit in \cite{mao2021optimal}]
\label{cor:moment-bounds}
Suppose the univariate distribution $A$ satisfies the following conditions for an even  $j_* \in \N$, $\kappa,\tau  \in \R_+$:
\begin{itemize}
   \item $\E_{X \sim A}h_j(X) = 0 $ for $j \in [j_*-1] $.
   \item For $j \geq j_*$, $\left|\E_{X \sim A}h_j(X)\right|^2 \leq j^j \kappa \tau^{-\frac{j}{j_*}}$.
 \end{itemize} 
Then if $n \kappa \geq 1$ and $n \leq \frac{1}{\poly(D)} d^{\frac{j_*}{2}} \frac{d^{\frac{j_*}{2}} \tau}{\kappa}$, then the degree-$D$ advantage of the NGCA problem with the hidden distribution  $A$ satisfies $\Adv_{\leq D, n} \leq O(1)$.
\end{lemma}
\begin{proof}
As mentioned earlier, we closely follow the proof arguments  in \cite[Theorem 4.5]{mao2021optimal} and simplify the bound in \Cref{lem:ldlr-Hermite}
for the given $A$.

For positive integers $1 \leq m \leq i \leq D$, 
define the set of multi-indices that (i) have support $m$, (ii) sum up to $i$, and (iii) each non-zero element is at least $j_*$:
$$\cA(i,m) = \Big\{\cI \in \N^n: |\cI| = i, |\cI|_0 = m; \alpha_i \in \{0\} \cup \left\{j_*, j_*+1, \dots\right\}\Big\}.$$
It can be seen that the cardinality of $\cA(i,m)$ is at most $ n^m i^i$ because the number of ways to select the support of $\cI \in \cA(i,m)$ is $\binom{n}{m} \leq n^m$, and once the support is selected, each element in the support has at most $i$ choices. 
Moreover, $\cA(i,m)$ is non-empty only if $m \leq \lfloor i/j_* \rfloor$.

The reason behind the restriction of non-zero elements in $\cI \in \cA(i,m)$ to be at least $j_*$ is because  
for any $\cI \not \in \cA(i,m)$, the expression  $\prod_{j \in \cI}\left(\E_{X \sim A}h_j(X)\right)^2 = 0$ in \Cref{lem:ldlr-Hermite}. 
Thus, we can restrict our attention to $\cI \in \cup_{i\geq j_*}\cup_{m=1}^{\lfloor i/j_*\rfloor}\cA(i,m)$ while evaluating \Cref{eq:degree-d-advantange-ngca}.
For any $\cI \in \cA(i,m)$ with $i \geq j_*$, the assumption on the Hermite coefficients imply that
\begin{align}
\label{eq:hermite-product-I-Aim}
\prod_{j \in \cI}\E_{X \sim A}\left(h_j(X)\right)^2 \leq \prod_{j \in \cI: j > 0} j^j \kappa \tau^{ - \frac{j}{j_*}} = i^{2i} \kappa^{m}\tau^{-\frac{i}{j_*}}\,.
\end{align}

We now calculate the following expression that corresponds to $|\cI| = i$ in \Cref{eq:degree-d-advantange-ngca}: for any $i \geq  j_*$, we obtain
\begin{align*}
 \sum_{\cI \in \N^n; |\cI| = i} \prod_{j=1}^n \left(\E_{X \sim A} [h_{\cI_j}(x)]\right)^2 &=   \sum_{m=1}^{ \lfloor i/j_*\rfloor}  \sum_{\cI \in \cA(i,m)} \prod_{j=1}^n \left(\E_{X \sim A} [h_{\cI_j}(x)]\right)^2 \\
&\leq \sum_{m=1}^{ \lfloor i/j_*\rfloor} |\cA(i,m)| i^{3i} \kappa^{m }\tau^{-\frac{i}{j_*}} \\
&\leq \sum_{m=1}^{ \lfloor i/j_*\rfloor} n^m i^{2i} \kappa^{m } \tau^{ - \frac{i}{j_*}} \\
&\leq i^{4i} \tau^{-\frac{i}{j_*}} \max(n\kappa, (n \kappa)^{\lfloor i/j_*\rfloor}) \\
&= i^{4i} \tau^{-\frac{i}{j_*}} (n \kappa) ^{i/j_*} 
= i^{4i} (n\kappa\tau^{-1})^{i/j_*},
\end{align*}
where we used that $n\kappa \geq 1$.

We are now ready to evaluate the complete expression in \Cref{eq:degree-d-advantange-ngca} below.
\begin{align*}
\Adv_{\leq D, n}
&\leq \sum_{i=0; i \text{ is even}}^D \left(\frac{i}{d}\right)^{\frac{i}{2}}  \sum_{\cI \in \N^n; |\cI| = i} \prod_{j=1}^n \left(\E_{X \sim A} [h_{\cI_j}(x)]\right)^2 \\
&\leq \sum_{i=0; i \text{ is even}}^D \left(\frac{i}{d}\right)^{\frac{i}{2}}   i^{4i}   \left(n\kappa\tau^{-1}\right)^{i/j_*}\\
&\leq \sum_{i=0; i \text{ is even}}^D \left( i^{4.5} \left(\frac{ n \kappa}{\tau d^{ \frac{j_*}{2}}} \right)^{\frac{1}{j_*}} \right)^i  \,,  
\end{align*} 
which is less than $O(1)$ if $n \ll \frac{1}{\poly(D)} \frac{d^{\frac{j_*}{2}} \tau}{\kappa}$.

\end{proof}

\subsection{Robust Covariance Estimation and Resilience for Gaussian Data}
\label{sec:low-degree-cov-est-resilience-gaussian}
In this section, we state the low-degree lower bounds  for the following testing problem:
\begin{problem}[Robust Gaussian Covariance Testing in Operator Norm]
 \label{def:cov-testing}
 Given corruption rate $\eta \in (0,1/2)$, deviation $\alpha \in (C\eta,1)$ for a large constant $C$, sample size $n\in \N$ and dimension $d \in \N$, consider the following distribution testing problem with input $y = (y_1,\dots,y_n) \in \R^{nd}$: 
 \begin{enumerate}
\item $H_0$: $y_1,\dots, y_n$ are sampled i.i.d.\ from $\cN(0,\bI_d)$.
\item $H_1$: First a unit vector $v$ is sampled randomly from $\cS^{d-1}$, and then 
$y_1,\dots, y_n$ are sampled i.i.d.\ from $(1-\eta) \cN(0, \bI_d - \alpha vv^\top) + \eta Q_v$, where $Q_v$ are arbitrary.
 \end{enumerate}
\end{problem}
Observe that the hypothesis testing problem above is easier than robust covariance estimation of Gaussian data with relative spectral norm error less than $c \alpha$ for a small enough constant $c>0$.\footnote{Moreover, the testing problem is also easier than computing robustly $\widehat{\Sigma}$ satisfying $0.5 \Sigma\preceq \widehat{\Sigma} \preceq \frac{c}{(1-\alpha)} \Sigma$ for a tiny constant $c>0$.}
Since certifying operator norm resilience gives an algorithm for covariance estimation in spectral norm (\Cref{thm:cov_est_full}), the same instance also give evidence of computational hardness for the problem of certifying operator norm resilience. 

The main result of this subsection is the following computational lower bound:
\begin{proposition}
\label{prop:low-degree-cov-estimation-gaussian}
Consider \Cref{def:cov-testing} with $\alpha \in (0, 1)$ and $\eta \gtrsim 1/\sqrt{d}$.
There exist a choice of $Q_v$'s such that the degree-$D$ advantage of \Cref{def:cov-testing} is $O(1)$
for
any
$d \ll n \ll \frac{d^2\eta^2}{\poly(D) \alpha^4}$.
\end{proposition}
The above lower bound matches the computational guarantee of \Cref{thm:cov_mean_est_full} of using $\eta^2d^2/\alpha^4$ samples (up to the $d^\eps$ factor).
 Observe that the condition $n\gg d$ in \Cref{prop:low-degree-cov-estimation-gaussian} is rather mild as it is needed to solve the robust covariance estimation information-theoretically~\cite{10.1214/17-AOS1607}.
\begin{proof}
Consider the NGCA instance with  $A = (1-\eta) \cN(0,1 -\alpha) + 0.5 \eta \cN(\delta, 1) + 0.5 \eta \cN(-\delta, 1) $ for $\delta \in \R_+$ to be decided soon. 
Observe that it is a special case of \Cref{def:cov-testing} with $Q_v = P_{A',v}$ for $A' = 0.5 \cN(\delta, 1) + 0.5 \cN(-\delta, 1)$.
To apply \Cref{lem:ldlr-Hermite}, we need to calculate $\E_{X \sim A} [h_j(X)]$.
We shall choose $\delta$ so that the variance of $A$ is equal to $1$.
It can be checked that taking $\delta$ to satisfy the following equation works $(1-\eta) (1-\alpha) + \eta (\delta^2 + 1) = 1$, which is satisfied for $\delta  \asymp \sqrt{\frac{\alpha}{\eta}}$.

We now calculate the quantities required to apply \Cref{cor:moment-bounds}.
We claim that for any $j \in \{1,2,3\}$, $\E_{X \sim A} [h_j(X)] = 0$, which follows by the fact that $A$ mathces the first three moments of $\cN(0,1)$.
The symmetry of $A$ further implies that $\E_{X \sim A} [h_j(X)] = 0$ for any odd $j$,
Next for any even $j \geq 4$,
a direct calculation using \Cref{item:her-unit-variance,item:her-ornstein} implies that
\begin{align*}
\E_{X \sim A} [h_j(X)] &= (1-\eta) \E_{\cN(0,1-\alpha)}[h_j(X)] + (\eta/2) \E_{\cN(\delta,1)}[h_j(X)] + (\eta/2) \E_{\cN(-\delta,1)}[h_j(X)]\\
&= (1-\eta) (\sqrt{\alpha})^{j}h_j(0) + \eta \frac{\delta^j}{\sqrt{j!}},
\end{align*}
leading to the following coarse bounds on the square value:
\begin{align*}
\left(\E_{X \sim A} [h_j(X)]\right)^2
&\lesssim \left(\frac{\alpha^{j/2}}{j!!}\right)^2 + \eta^2 \frac{\alpha^j}{ \eta^j j!} \lesssim \eta^2(\alpha/\eta)^j 
\lesssim  \eta^2 \left(C\frac{\alpha^4}{\eta^4}\right)^{\frac{j}{4}}.
\end{align*}

We can now apply \Cref{cor:moment-bounds} with $j_* = 4$, $\kappa= \eta^2$, $\tau = \eta^4/\alpha^4$, we get that the degree-$D$ advantage is $O(1)$ for any $n \ll \frac{1}{\poly(D)} d^2 \frac{\eta^4}{\alpha^4} \frac{1}{\eta^2} = \frac{d^2\eta^2}{ \poly(D)\alpha^4}$ and $n \gg 1/\eta^2$.
Since $\eta \ll 1/\sqrt{d}$, the condition $n\gg 1/\eta^2$ is satisfied if  $n \gg d$.
\end{proof}

\subsection{Covariance-Aware Mean Testing and Estimation}
\label{sec:low-degree-cov-mean-est}
In this section, we focus on the problem of covariance-aware mean estimation.
We first define the following related testing problem that asks us to (robustly) distinguish between $\cN(0, \vec I)$ and $\cN(\mu,\Sigma)$ with $\|\Sigma^{-1/2} \mu\|_2 \geq \alpha$; See \Cref{sec:testing-vs-estimation-cov-aware} for its relation with the estimation problem.

\begin{problem}[Covariance-Aware Robust Mean Testing]
 \label{def:rm-testing}
 Given corruption rate $\eta \in (0,1/2)$, Mahalanobis norm $\alpha \in \R_+$, sample size $n\in \N$ and dimension $d \in \N$, consider the following distribution testing problem with input $y = (y_1,\dots,y_n) \in \R^{nd}$: 
 \begin{enumerate}
\item $H_0$: $y_1,\dots, y_n$ are sampled i.i.d.\ from $\cN(0,\bI_d)$.
\item $H_1$: First a unit vector $v$ is sampled randomly from $\cS^{d-1}$, and then 
$y_1,\dots, y_n$ are sampled i.i.d.\ from $D_v:= (1-\eta) \cN(\mu_v, \vec \Sigma_v) + \eta Q_v$, where $Q_v$ are arbitrary and the mean of the inlier Gaussian distribution is at least $\alpha$  in Mahalanobis norm,  $\|\Sigma^{-1/2} \mu\|_2 \geq \alpha$.
 \end{enumerate}
\end{problem}
In fact, the lower bound instance in the following lower bound would satisfy $\mu_v = \delta v$ and   $\Sigma_v = \I - (1 - \sigma^2)vv^\top $  for scalars $\delta$, $\sigma$ satisfying $\alpha = \delta/\sigma$ and $\delta \asymp \sqrt{\eta}$.

\begin{theorem}[Low-Degree Hardness of Covariance-Aware Mean Testing]
\label{thm:low-degree-hardness-cov-aware-mean}
    Consider $\eta \gg 1/\sqrt{d}$.
    For any $\alpha > 0$, 
    there exist a choice of $\mu_v$, $\Sigma_v$, and $Q_v$'s such that the degree-$D$ advantage of \Cref{def:rm-testing} is $O(1)$ for any $n \gg d $ and $n \ll \frac{d^2\eta^2}{\poly(D)}$.
    In fact, the conclusion holds even when the minimum eigenvalue of $\Sigma_v$ is at least $\Omega\left(\min\left(\frac{\eta}{\alpha^2}, 0.5\right)\right)$. 
\end{theorem}
\begin{proof}
We shall take $Q_v$'s so that the distribution $D_v$ correspond to an NGCA instance $P_{A,v}$ for $A = (1-\eta) \cN(\delta, \sigma^2) + \eta Q$ for $\delta = \alpha \sigma$. 
Thus, $D_v = (1-\eta) \cN(\delta v, (I-vv^\top) + \sigma^2vv^\top  ) + \eta P_{Q,v}$.
The next lemma shows that there exists a distribution $Q$ so that $A$ matches many moments with $\cN(0,1)$.

\begin{lemma}
\label{lem:univariate-dist-new}
\looseness=-1
There exists a positive constant $\eta_0$ such that for all $\eta \in (0,\eta_0)$ and for every $\sigma \in (0,1/2)$, there exists a univariate distributions $A$ and $Q$ such that $A := (1 - \eta) \cN(\delta, \sigma^2) + \eta Q$ such that 
\begin{enumerate}[label=(\roman*)]
    \item \label[sublemma]{item:A-1} $\delta = 0.001 \sqrt{\eta}$.    
    \item \label[sublemma]{item:A-2} $A$ matches three moments with $\cN(0,1)$, i.e., $\E_{X \sim A}[h_{j}(X)] = 0$ for $j \in \{1,2,3\}$.
 \item \label[sublemma]{item:A-4} For  $j\geq 4$, $\left(\E_{X \sim A}[h_{j}(X)]\right)^2 \lesssim   \eta^2 \left(\frac{c}{\eta^4}\right)^{j/4}$ for a consant $c>0$.
\end{enumerate}
\end{lemma}
We defer the proof of \Cref{lem:univariate-dist-new} to \Cref{sec:proof_of_lem_univariate_dist}.

Since our goal is to show a lower bound, we will assume that $\alpha > 2\delta$, i.e., $\sigma < 1/2$. For any such $\alpha$, \Cref{lem:univariate-dist-new} is applicable.

Applying \Cref{cor:moment-bounds} with $j_* = 4$, $\kappa = \eta^2$ and $\tau = \eta^4$, we get that the degree-$D$ advantage is $O(1)$ for any $n \ll  \frac{d^2 \eta^4}{\poly(D)\eta^2} = \frac{d^2 \eta^2}{\poly(D)} $
and $n \gg d \gg \frac{1}{\eta^2}$.
\end{proof}
\paragraph{Proof of \Cref{lem:univariate-dist-new}}
\label{sec:proof_of_lem_univariate_dist}
We now provide the proof of \Cref{lem:univariate-dist-new}.
\begin{proof}
We establish these properties one by one.
By equating the moments of $A$ with $\cN(0,1)$,  
we obtain the following equivalent conditions for the desired distribution $Q$:
\begin{enumerate}
    \item $\E_A[X] = (1-\eta) \delta + \eta \E_Q[X]  = 0 $ implies 
     $\E_Q[X] = -\frac{(1-\eta) \delta}{\eta}$.   
    \item $\E_A[X^2] = (1-\eta) \left(\delta^2 + \sigma^2\right) + \eta \E_Q[X^2]  = 1 $ implies
        $\E_Q[X^2] = \frac{1 - (1-\eta)(\delta^2 + \sigma^2)}{\eta}$.
\item $\E_A[X^3] = (1 - \eta) (\delta^3 + 3 \delta \sigma^2) + \eta \E_Q[X^3] = 0$ implies
    $\E_Q[X^3] = - \frac{(1-\eta) \delta (\delta^2 + 3 \sigma^2)}{\eta }$
\end{enumerate}

We will choose $Q$ to be of the form $\sum_{j=1}^4 w_i\cN(x_i,1) $ for some $|x_i| \leq B$ for some finite $B$ (to be decided soon) and convex weights $w_i$.
Let $F$ be the discrete distribution $\sum_i w_i \I_{x_i}$.
The conditions on the moments of $Q$
now imply the following equivalent conditions for the moments of $F$:
\begin{align*}
    \sum_i w_i x_i &= -\frac{(1-\eta) \delta}{\eta}\\
    \iff \E_F[X] &= -\frac{(1-\eta) \delta}{\eta}, \,\,\qquad \text{and} \numberthis\label{eq:moment-x-w-1} \\
    \sum_i w_i (x_i^2 + 1) &= \frac{1 - (1-\eta)(\delta^2 + \sigma^2)}{\eta} \\
    \iff \E_F[X^2] &= \frac{1 - (1-\eta)(\delta^2 + \sigma^2) - \eta}{\eta}, \,\,\qquad \text{and} \numberthis\label{eq:moment-x-w-2}\\
    \sum_i w_i (x_i^3 + 3x_i) &= - \frac{(1-\eta) \delta (\delta^2 + 3 \sigma^2)}{\eta } \\
\iff \E_F[X^3] &= - \frac{(1-\eta) \delta (\delta^2 + 3 \sigma^2)}{\eta } - 3 \left(-\frac{(1-\eta) \delta}{\eta}\right) \numberthis\label{eq:moment-x-w-3},
\end{align*}
where the last step uses \Cref{eq:moment-x-w-1} and the fact that $\E[X^3] = \mu^3 + 3 \mu \sigma^2$ for $X \sim \cN(\mu,\sigma^2)$.
The following lemma gives the existence of such a $F$ in a bounded interval:
\begin{lemma}
\label{lem:discrete-x-w}
There exists a positive constant $\eta_0$ such that for all $\eta \in (0,\eta_0)$ and $\sigma^2 \in (0,0.1)$: there exists a discrete distribution $F$ over $[-10/\sqrt{\eta}, 10/\sqrt{\eta}]$ with support of size at most $4$ elements that matches the moments in \Cref{eq:moment-x-w-1,eq:moment-x-w-2,eq:moment-x-w-3} such that $\delta = \sqrt{\eta}/1000$.
\end{lemma}
\begin{proof}[Proof of \Cref{lem:discrete-x-w}]
 We shall use \Cref{fact:moment-matching}, stating that there exists a distribution $F$ over $[-B,B]$ that matches these three moments if and only if the following matrices are PSD:
 \begin{align}
    \begin{bmatrix}
        1 + \frac{\E_F[X]}{B} & \frac{\E_F[X]}{B}+ \frac{\E_F[X^2]}{B^2} \\
        \frac{\E_F[X]}{B}+ \frac{\E_F[X^2]}{B^2} &  \frac{\E_F[X^2]}{B^2} + \frac{\E_F[X^3]}{B^3} 
    \end{bmatrix} \,\, \text{and}\,\,
    \begin{bmatrix}
        1 - \frac{\E_F[X]}{B} & \frac{\E_F[X]}{B}- \frac{\E_F[X^2]}{B^2} \\
        \frac{\E_F[X]}{B}- \frac{\E_F[X^2]}{B^2} &  \frac{\E_F[X^2]}{B^2} - \frac{\E_F[X^3]}{B^3} 
    \end{bmatrix}
\end{align}
It suffices to enforce the following conditions:
\begin{enumerate}
    \item $|\E_F[X/B]| \leq 0.1$
    \item $|\E_F[X/B]| \leq 0.1 \E[X^2/B^2]$
    \item $\E_F[X^2/B^2] \leq 0.1$
    \item $|\E_F[X^3/B^3]| \leq 0.1 \E[X^2/B^2]$    
\end{enumerate}
Indeed, if this happens then, the diagonal elements are non-negative, and their product is at least $0.9^2 \E_F[X^2/B^2]$, which is larger than the the product of off-diagonal entries, which is at most $1.1^2(\E[X^2/B^2])^2 \leq 0.5 \E_F[X^2/B^2]$. Thus, these two matrices will be PSD. 

We will now show that for the claimed values, we have $|\E_F[X/B]| \leq 0.0005$, $\E[X^2/B^2] \geq 0.005$, and $\E_F[X^3/B^3] \leq 0.0005$.
We will choose $B$ such that $\delta/B$ to be $10^{-4} \eta$ and $B = 10/\sqrt{\eta}$.
Indeed, $|\E_F[X/B]|$ is less than $(\delta/B\eta)$, which is less than 0.0001 in  absolute value.
The second moment, $\E_F[X^2/B^2]$, becomes $10^{-2}(1 - (1-\eta)(\delta^2 + \sigma^2) - \eta)$, which satisfies the desired bound if $\eta + \delta^2 + \sigma^2 \leq 0.5$; The latter happens for small enough $\eta$, $\delta$, and $\sigma^2$.
Finally,  
$$\E_F[X^3/B^3] = -\frac{\delta}{B^3} \cdot \frac{(1-\eta)}{\eta}  \left(\delta^2 + 3\sigma^2 + 3\right) \leq 0.0005\,.$$

\end{proof}

\Cref{lem:discrete-x-w}  establishes \ref{item:A-1} and \ref{item:A-2} directly.
We now turn our attention to \ref{item:A-4}. Define $a_j = \E_{A}[h_j(X)] $. 
Since $A$ matches three moments with $\cN(0,1)$, we obtain that $a_1 = a_2 = a_3 = 0$. For $j \geq 4$, we use \Cref{item:her-unit-variance,item:her-ornstein} to get:
\begin{align*}
    a_j &= (1-\eta) \E_{X \sim \cN(\delta, \sigma^2)}[h_j(X)] +  \sum_i \eta w_i \E_{X \sim \cN(x_i, 1)} [h_j(X)] \\
    &= (1-\eta) (1- \sigma^2)^{j/2} h_j\left(\frac{\delta}{\sqrt{1 -\sigma^2}}\right) + \sum_i \eta w_i \frac{x_i^j}{\sqrt{i!}},
\end{align*}
which implies that  $|a_j|^2 \leq  \left(\frac{c}{\sqrt{j}} + \eta\left(\frac{100}{\eta}\right)^{j/2} \frac{1}{\sqrt{j!}} \right)^2
\lesssim \eta^2 (c/\eta^4)^{j/4}$ for a constant $c > 0$.

\end{proof}
\subsubsection{Covariance-Aware Mean: Testing Versus Estimation}
\label{sec:testing-vs-estimation-cov-aware}
In this section, we consider the distinction between the testing problem \Cref{def:rm-testing}  and the estimation problem considered in \Cref{thm:cov_mean_est_full}.
While the testing problems are usually trivially easier than the corresponding estimation problems, the situation for the covariance-aware mean estimation problem is subtle because the underlying metric $\|\Sigma^{-1/2} (\cdot)\|_2$ is also unknown.
The next result shows that the testing problem in \Cref{thm:low-degree-hardness-cov-aware-mean} can be solved efficiently with an estimation subroutine with a sample overhead of at most $\alpha^2/\eta$.

\begin{proposition}[Testing using an estimation algorithm]
\label{prop:testing-vs-estimation}
Let $\cA$ be an algorithm that (i) takes as input $\eta$-corrupted data from $\cN(\mu,\Sigma)$, (ii) uses $n$ samples and time $T$, and (iii) returns  $\widehat{\mu}$ such that $\|\Sigma^{-1/2}(\widehat{\mu} - \mu)\|_2 \leq 0.1\alpha$ with probability $0.99$.

Then there is an algorithm $\cA'$ that (i) takes $C n \cdot \min(\alpha^2/\eta,1)$ samples from \Cref{def:rm-testing}, (ii) runs in time $O(T) + \poly(d, n, 1/\eta)$, and solves the hard instance in \Cref{thm:low-degree-hardness-cov-aware-mean} with probability $1 - 0.99\exp(-c \alpha^2/\eta)$.

In particular, combining \Cref{prop:testing-vs-estimation,thm:low-degree-hardness-cov-aware-mean} gives an evidence of computational hardness of achieving error
\begin{itemize}
    \item $\|\Sigma^{-1/2}(\widehat{\mu} - \mu)\| = O(1)$ with using $o(d^2 \eta^3)$ samples up to $\polylog(d)$ factors.
    \item $\|\Sigma^{-1/2}(\widehat{\mu} - \mu)\| = O(\sqrt{\eta})$ with using $o(d^2 \eta^2)$ samples up to $\polylog(d)$ factors.

\end{itemize}
\end{proposition}
\begin{proof}
Let $k \asymp \alpha^2/\eta$.
The tester $\cA'$ is as follows:
\begin{enumerate}
    \item Let $S_1,\dots,S_k$ be $k$ datasets of $n$ independent samples from $\Cref{def:rm-testing}$. 
    \item Generate $k$ independent random $d$-dimensional rotation matrices  $U_1,\dots, U_k$.
    \item For $i \in [k]$:
    \begin{enumerate}
        \item Let $S_i' \gets\{U_i x_j : x_j \in S_i\}$ be the randomly rotated data.
        \item Let $y_i \gets U_i^\top \cA(S_i', \alpha,\eps)$ be the (inverted) output of the estimation algorithm $\cA$. 
    \end{enumerate}
    \item Compute $r:= \tfrac{1}{k}\max_{u \in \cS^{d-1}} \sum_{i \in [k]} (y_i^\top u)^2\1_{\|y_i\|_2 \leq 10\alpha + \sqrt{\eta}}$, the largest eigenvalue of $y_i$'s.
   \item If $r \geq c \eta $, output ``$H_1$'', else ``$H_0$''.    
\end{enumerate}

We now establish the correctness of the procedure, showing that $r$ is indeed large for the alternate distribution and small for the null distribution.
Let $\delta = 0.001\sqrt{\eta}$ denote the mean of the inlier Gaussian under $H_1$ in the hard instance from \Cref{thm:low-degree-hardness-cov-aware-mean}.

\paragraph{Alternate distribution.} 
We will show that with high probability, $v^\top y_i \geq \delta/2$ and $\|y_i\|_2 \leq 10 \alpha $ for majority of $i \in [k]$, which would imply that $r \geq \delta^2/8$.
Recall that in the hard instance,  $\Sigma_v^{-1/2} = (I - vv^\top) + \sigma^{-1}vv^\top$ for $\sigma = \delta/\alpha$ and $\mu_v = \delta v$, and thus for an arbitrary $x \in \R^d$,
\begin{align*}
\|\Sigma^{-1/2}(x - \mu)\|_2 \geq \left|v^\top \left((I - vv^\top) + \sigma^{-1}vv^\top\right)(x - \delta v)\right| = \left|\frac{(v^\top x - \delta)}{\sigma}\right| \,.
\end{align*}
Hence, the output of the estimation algorithm must satisfy $\left|v^\top\widehat{\mu} - \delta\right| \leq 0.1 \alpha\sigma = 0.1\delta$ with probability at least $0.99$.
Moreover, on the same probability event, since $\Sigma \preceq I$, we also have that $\|\widehat{\mu}\|_2 \leq \|\Sigma^{-1/2}(\widehat{\mu} - \mu)\| + \delta \leq \alpha + \delta$.
In fact, the same conclusion is true for all $ i \in [k]$ for $y_i$, 
because we even though we perform the rotation $U_i$, the definition of $y_i$ inverts the rotation.
By a Chernoff argument, with probability at least $1 - 0.01 \exp(-ck)$, 
for at least half of the $y_i$'s, $|v^\top y_i - \delta| \leq 0.1 \delta$ and hence $|v^\top y_i| \geq 0.5 \delta$.

\paragraph{Null distribution.}
When the underlying distribution is isotropic Gaussian, the rotational invariance of Gaussians imply that $U_i$ is independent of $S_i'$.
Thus, $U_i$ is also independent of $\cA(S_i', \alpha,\eps)$, implying that $y_i$ is a uniformly random vector of norm $\|y_i\|$. 
Thus $z_i:= y_i/\|y_i\|$ are uniformly random unit vectors that are independent of each other.
If we define $L = \tfrac{1}{k} \sum_iz_iz_i^\top$, then $\E[H] = \tfrac{1}{d}I$.
Applying concentration of the covariance of uniformly random vectors~\cite{Vershynin18}, we obtain that with probability at least $1 - \exp(- c k)$,
$\|H\|_\op \leq \frac{1}{d} \left(1 + \sqrt{\frac{d}{k}} + \frac{d}{k} \right) \leq \frac{1}{d} + \frac{1}{\sqrt{kd}} + \frac{1}{k}$, which is $O\left(\frac{1}{k}\right)$.

Therefore, with probability at least $1 - \exp(-ck)$, thus $r$ is at most $(\alpha^2 +  \eta)$ times $O(1/k)$.
For this to be less than $\delta^2/2$, we need $k \gtrsim \alpha^2/\delta^2 = \Theta(\alpha^2/\eta)$. 
\end{proof}

\subsubsection{Consequences on Differentially Private Covariance-Aware Mean Estimation}

In this section, we state the consequences of \Cref{prop:testing-vs-estimation} on differentially private covariance-aware mean estimation by using the reduction between robust estimation and differential privacy~\cite{DwoLei09,GeoHop22}.

Let $\nstar_{\priv,\eff}(\alpha,\beta,\eps,\delta)$ denote the sample complexity of the best polynomial-time algorithm that is $(\eps,\delta)$-DP and when the input is a set of $n$ i.i.d.\ samples from $\cN(\mu,\Sigma)$, then with probability $1-\beta$, it outputs $\widehat{\mu}$ satisfying $\|\Sigma^{-1/2}(\widehat{\mu}-\mu)\|_2 \leq \alpha$.
Similarly, define $\nstar_{\robust,\eff}(\alpha,\beta,\eta)$ denote the sample complexity of the best polynomial-time algorithm that takes as input of $\eta$-corrupted set of  $n$ i.i.d.\ samples from $\cN(\mu,\Sigma)$ and with probability $1-\beta$, it outputs $\widehat{\mu}$ satisfying $\|\Sigma^{-1/2}(\widehat{\mu}-\mu)\|_2 \leq \alpha$.

\begin{conjecture}
\label{conjecture:robust-cov-aware-mean-est}
For any $\alpha> 0$, 
$\nstar_{\robust,\eff}(\alpha,\beta,\eta) \gg\eta^2 d^2 \max\left(\frac{\eta}{\alpha^2}, 1\right) $.    
\end{conjecture}
The conjecture above is substantiated by \Cref{prop:testing-vs-estimation},
\begin{proposition}
\label{prop:private-lower-bound}
Under \Cref{conjecture:robust-cov-aware-mean-est}, for any $\alpha > 0$
   \begin{align*}
       \nstar_{\priv,\eff}(\alpha,\beta,\eps,\delta) \gg \max_{t \in (0,1/2)} \min\left( d^{2-2t} \min(d^{-t}\alpha^{-2},1), d^t \cdot \frac{\log(1/\beta)}{\eps}, d^t \frac{\log(1/\delta)}{\eps}  \right)\,.
   \end{align*}
\end{proposition}
We now give context behind \Cref{prop:private-lower-bound}.
Let $\nstar_\priv(\alpha,\beta,\eps,\delta)$ denote the information theoretic sample complexity  of covariance-aware mean estimation using an $(\eps,\delta)$-DP algorithm  to obtain error $\alpha$ with probability $1-\beta$.
To simplify the discussion, consider the regime of $\alpha = 1$ and $\eps = 1$.
The following bound on the information-theoretic sample complexity from \cite[Theorem 3.2]{BroGSUZ21}  states $\nstar_\priv(1,\beta,1,\delta)$ is at most 
\begin{align*}
    \nstar_\priv(1, \beta, 1,\delta) \lesssim d + \log(1/\beta) + \log(1/\delta).
    \end{align*}
However, \Cref{prop:private-lower-bound} implies that this sample complexity can not be achieved efficiently (under \Cref{conjecture:robust-cov-aware-mean-est}), stating that $\nstar_{\priv, \eff}(1,\beta, 1,\delta) \geq \max_{t \in (0,1/2)} \min(d^{2-3t}, d^t \log(1/\beta), d^t \log(1/\delta))$.
In particular, any computationally-efficient algorithm that uses $d^{2 - \Omega(1)}$ samples must have polynomial prefactor $d^{\Omega(1)}$ either on  $\log(1/\delta)$ or $\log(1/\beta)$.
Indeed, existing computationally-efficient algorithms  use more samples than the information-theoretic sample complexity~\cite{BroHS23,KDH23}.

We now give the proof of \Cref{prop:private-lower-bound}.
\begin{proof}
   We shall use the following reduction stating that high-probability private estimators are automatically robust~\cite{DwoLei09,GeoHop22}  (see \cite[Theorem 3.1]{GeoHop22}). The reduction states that if an algorithm $\cA$ with domain $(\R^d)^n$ satisfies:
\begin{enumerate}
    \item $\cA$ is $(\eps,\delta)$-DP
\item If $S$ is a set of $n$ i.i.d.\ samples from $\cN(\mu,\Sigma)$, then with probability $1 - \beta$, $\cA$ outputs $\widehat{\mu} \in \R^d$ such that $\|\Sigma^{-1/2}(\widehat{\mu}-\mu)\|_2 \leq \alpha$.
\end{enumerate}
Then $\cA$ also satisfies the following conclusions: If $S$ is a set of $\eta$-corrupted samples from $\cN(\mu,\Sigma)$ for 
\begin{align}
\label{eq:rob-to-private-corruption}
\eta \asymp \min\left( \frac{\log(1/\beta)}{\eps n}, \frac{\log(1/\delta)}{\eps n + \log n}\right), 
\end{align}
then with probability $1- \poly(\beta)$, $\cA$ computes an estimate $\widehat{\mu}$ such that $\|\Sigma^{-1/2}(\widehat{\mu} - \mu)\|_2 \lesssim \alpha$.

    Let $\cA$ be an $(\eps,\delta)$-DP computationally-efficient algorithm that uses fewer than the claimed samples, achieved for a certain $t$.
    Then \Cref{eq:rob-to-private-corruption} implies that $\cA$ is robust to 
    \begin{align*}
        \eta_* \asymp \min\left( \frac{\log(1/\beta)}{\eps n}, \frac{\log(1/\delta)}{\eps n + \log n}\right) \geq d^{-t}
    \end{align*}
However, \Cref{conjecture:robust-cov-aware-mean-est} implies that any computationally-efficient algorithm requires at least $d^2 \eta^2\min(\eta\alpha^{-1},1)$ samples.
\end{proof}

\subsection{Robust Covariance Estimation  for Subgaussian Distributions}
\label{sec:low-degree-cov-est-subgaussian}
In this section, we consider the problem of robust covariance estimation of \emph{subgaussian} distributions (and the sample complexity of certifying operator norm resilience in \Cref{rem:phase-transition} for \emph{Gaussian} data) in the low-error regime of $\alpha \ll \sqrt{\eta}$.
We begin by defining the following testing problem:
\begin{problem}[Robust Subgaussian Covariance Testing in Operator Norm]
 \label{def:cov-testing-subgaussian}
 Given corruption rate $\eta \in (0,1/2)$, deviation $\alpha \in (\eta \log(1/\eta),0.1)$, sample size $n\in \N$ and dimension $d \in \N$, consider the following distribution testing problem with input $y = (y_1,\dots,y_n) \in \R^{nd}$: 
 \begin{enumerate}
\item $H_0$: $y_1,\dots, y_n$ are sampled i.i.d.\ from $\cN(0,\bI_d)$.
\item $H_1$: First a unit vector $v$ is sampled randomly from $\cS^{d-1}$, and then 
$y_1,\dots, y_n$ are sampled i.i.d.\ from $(1-\eta) D_v + \eta Q_v$, where $Q_v$ are arbitrary and $D_v$ is a centered subgaussian distribution with covariance $\Sigma_v$ satisfying $\|\Sigma_v - \bI_d\|_\op \geq \alpha$.
 \end{enumerate}
\end{problem}
Information-theoretically, $\Theta(d/\alpha^2)$ samples are necessary and sufficient to solve the above testing problem for any $\alpha = \Omega(\eps \log(1/\eps))$.
As in the Gaussian setting, the hypothesis testing problem above is easier than robust covariance estimation of subgaussian data with relative spectral norm error less than $c \alpha$ for a small enough constant $c>0$. 

The main result of this subsection is the following computational lower bound for the low-degree polynomial tests.
\begin{proposition}
\label{prop:low-degree-subgaussian-cov}
Fix $i_* \in \N$ to be an even constant and let $c_{i_*}$ be a  small enough constant and $c_{i_*}'$ be a large enough constant.
Let $1/\sqrt{d} \ll \eta \leq 10^{-5}$ satisfying 
\begin{align*}
c_{i_*}' \eta^{1 - \frac{2}{i_*+2}} \ll \alpha \ll c_{i_*} \eta^{1 - \frac{2}{i_*}}.
\end{align*}
There exist a choice of $Q_v$'s and $D_v$'s such that the degree-$D$ advantage of \Cref{def:cov-testing} is $O(1)$ for any $n \gg d $ and $n \ll \frac{d^{\frac{i_* + 2}{2}} \eta^{i_*}  }{\poly(D){\alpha}^{i_* + 2}}$
\end{proposition}

The low-degree lower bound above suggests that getting error $c_p \eta^{1 - \frac{2}{p}}$ efficiently for a tiny constant $c_p$ requires at least roughly
$\frac{d^{0.5p+ 1} \eta^p}{\eta^{p - \frac{4}{p}}} = d^{0.5p + 1 } \eta^{\frac{4}{p}}$ samples.
To explain the parameter regime above, we instantiate the result for certain small values of $i_*$:
\begin{enumerate}[label=(\Roman*)]
	\item $i_*=2$: For $\alpha \in (c_2'\sqrt{\eta}, c_2)$, $n$ must be larger than $\frac{d^2 \eta^2}{\alpha^4}$.
	\item $i_*=4$: For $\alpha \in (c_4'\eta^{\frac{2}{3}}, c_4 \sqrt{\eta})$, $n$ must be larger than $\frac{d^3 \eta^4}{\alpha^6}$.
	\item $i_*=6$: For $\alpha \in (c_6'\eta^{\frac{3}{4}}, c_6 \eta^{\frac{2}{3}})$, $n$ must be larger than $\frac{d^4 \eta^6}{\alpha^8}$.
\end{enumerate}

\begin{remark}[Phase transition at $\eta^{\frac{1}{4}}$ for certifying $\eta$-\textsc{SSV}]
\label{rem:phase-transition}
Consider the $\eta$-SSV problem for Gaussian data.
As mentioned in \Cref{sec:transfer-lemma}, certifying a bound of $\sqrt{\alpha}$ for \textsc{SSV} gives a certificate of $\alpha$ for the operator norm resilience for \emph{subgaussian} data.
Hence, certifying a bound of $\sqrt{\alpha}$ for the $\eta$-\textsc{SSV} (with Gaussian data) problem solves \Cref{def:cov-testing-subgaussian} efficiently.
Taking $\alpha = c \sqrt{\eta}$ for a tiny constant $c > 0$,
the above low-degree lower bounds give evidence that efficiently certifying a bound of $ \sqrt{c}\eta^{1/4}$ for the $\eta$-\textsc{SSV} problem  requires at least $d^3 \eta = \omega(d^{2.5})$ samples (under $\eta\gg 1/\sqrt{d}$).
In contrast, we gave an efficient certification to certify a bound of $C\eta^{1/4}$ for a large constant $C$ using only $O(d^{2 + \eps})$ samples.

\end{remark}

We now give the proof of \Cref{prop:low-degree-subgaussian-cov}.
\begin{proof}
We again choose an NGCA instance with hidden distribution $A'$ from the following lemma.

\begin{lemma}
\label{lem:univariate-A-subgaussian}
Consider $\alpha, \eta, i_*$ satisfying the conditions of \Cref{prop:low-degree-subgaussian-cov}.
There exists a univariate distribution $A'$ satisfying the following conditions:
\begin{enumerate}
  \item $A' = (1-\eta) D' + \eta Q'$ for two distributions  $D'$ and $Q'$.
  \item $D'$ is subgaussian and has variance $1-\alpha$.
  \item $A'$ matches first $i_*+1$ moments with $A$ and is symmetric around $0$.
  \item Setting $j_* = i_* +2$, we have that $\E_{A'}h_j(X) = 0$ for $j \in [j_* - 1]$ and for $j\geq j_*$, 
  \begin{align*}
  \left(\E_{A'}[h_j(X)]\right)^2 \lesssim \eta^2 \left(C^j \frac{\alpha^j}{\eta^j}\right)^{\frac{j}{j_*}}.
  \end{align*} 
\end{enumerate}
\end{lemma}

We defer the proof of the lemma above to \Cref{sec:proof-univariate-A-subgaussian} and first show that the lemma above suffices for the proof.
Applying \Cref{cor:moment-bounds} with $j_* = i_*+2$, $\kappa = \eta^2$ and $\tau = \left(\frac{\eta}{\alpha}\right)^{i_* + 2}$, we get that the degree-$D$ advantage is $O(1)$ for any $n \ll  \frac{d^{\frac{i_* + 2}{2}} \left( \frac{\eta}{\alpha}\right)^{i_* + 2} }{\poly(D)\eta^2} = \frac{d^{\frac{i_* + 2}{2}} \eta^{i_*}  }{\poly(D){\alpha}^{i_* + 2}} $
and $n \gg \frac{1}{\eta^2}$.

\end{proof}

\subsubsection{Proof of \Cref{lem:univariate-A-subgaussian}}
\label{sec:proof-univariate-A-subgaussian}

\begin{proof}

Let $\phi_{\mu,\sigma^2}(\cdot)$ denote the pdf of $\cN(\mu,\sigma^2)$.
We will again use an NGCA instance, with the hidden distribution $A'$ having pdf equal to $A'(x) = 0.5 A(x) + 0.5 A(-x)$, i.e., the symmetrized version of $A$.
For a degree-$i_*$ polynomial $p$ to be decided soon, the pdf of  $A$ is defined to be 
\begin{align*}
 A(x) = (1-\eta) \left(\phi_{0,1 - \alpha}(x) + p(x) \1_{[-1,1]}(x)\right) + \eta \left(0.5\phi_{\delta',1}(x) + 0.5 \phi_{-\delta',1}(x)\right),
\end{align*}
where 
$\delta' \in \R_+$.
Let $D$ corresponds to the distribution with pdf $\phi_{0,1 - \alpha}(x) + p(x) \1_{[-1,1]}(x)$
and $Q$ correspond to the distribution with pdf $0.5\phi_{\delta',1}(x) + 0.5 \phi_{-\delta',1}(x)$.
That is, the distribution $D$ corresponds to a modification of $\cN(0,1-\alpha)$ in a bounded interval.
Since the tail probabilities of $D$ and $\cN(0,1-\alpha)$ are equal, we get that $D$ is a subgaussian distribution.
We shall choose the polynomial $p(x)$ so that (i) $D$ is non-negative, (ii) $D$ integrates to $1$, (iii) $D$ is zero mean, and (iv) $D$ has variance $1-\alpha$. 

Set $\delta'$ so that the second ``moment'' of $A$ is $1$. 
That is,
\begin{align*}
(1-\eta) \left(1 - \alpha\right) + \eta \left(\delta'^2 + 1 \right) = 1,
\end{align*}
which is satisfied for $\delta' \asymp \sqrt{\alpha/\eta}$, which is bigger than $1$.
We now choose the polynomial $p$ to be the unique (at most) degree-$i_*$ polynomial so that $A$ matches $\{0\} \cup [i_*]$ ``moments'' with $G \sim \cN(0,1)$, which shall establish (ii)-(iv) above (we shall ensure that (i) is also satisfied at the end). 
Defining $a_i := \int_{-1}^1 x^i p(x)dx$ for each $i \in [0,\dots,i]$,  we want  $a_0 =0$ and for $i\in [i_*]$, 
\begin{align*}
(1-\eta)\left(  \E\left[\left(\sqrt{1-\alpha}G\right)^i\right] + a_i\right) + \eta \left(0.5 \E\left[\left(\delta' + G\right)^i\right] + 0.5 \E\left[\left(\delta' - G\right)^i\right]\right) = \E[G^i],
\end{align*}
which is equivalent to $a_i = 0$ for odd $i$ and for even $i$, $a_i$ must satisfy
\begin{align*}
a_i &= \frac{1}{1-\eta}\left(\E[G^i] - (1-\eta) (1 - \alpha)^{\frac{i}{2}} \E[G^i] - \eta \E\left[\left(\delta' + G\right)^i\right]\right) \\
&= \frac{1 - (1-\eta)(1-\alpha)^{\frac{i}{2} }}{1-\eta} \E[G^i] - \frac{\eta}{1-\eta}\E\left[\left(\delta' + G\right)^i\right]\,.
\end{align*}
In particular, if $\alpha \ll \frac{1}{i}$, then the required $a_i$'s satisfy
\begin{align*}
|a_i| &\lesssim \left(\alpha i\right)(Ci)^{i/2} + \eta 2^i \left((\delta')^i + i^{i/2}\right) \lesssim (O(i))^{O(i)} \max(\alpha, \eta \delta'^i)\lesssim (O(i))^{O(i)} \max\left(\alpha, \eta \left(C\frac{\alpha}{\eta}\right)^{\frac{i}{2}}\right)\,.
\end{align*}
Under the assumptions on $i_*$, $\alpha$, and $\eta$ in the theorem statement,
it follows that $|a_i| \leq c_{i_*}''$ for a tiny constant $c_{i_*}''$.
Lemma 8.18 in \cite{DiaKan22-book} implies that there there exists a degree-$i_*$ polynomial $p$ with $\int_{-1}^1 p(x)dx = a_i$ for all $i \in \{0\} \cup [i_*]$ and  
\begin{align*}
\max_{x \in [-1,1]} p(x) \leq O_{i_*}(\max_{i\leq i_*}|a_i|) \leq O_{i_*}(|a_{i_*}|) \leq 10^{-5}\,.
\end{align*}
Since the pdf $\phi_{0,1-\alpha}$ is at least $0.0001$ in the interval $[-1,1]$, we get that $D$ and hence $A$ is a valid distribution.
So far, we have shown that $A$ matches the first $[i_*]$ moments with $\cN(0,1)$ and assigns sufficient mass to a subgaussian distribution.
Since $A'$ is a symmetrized version of $A$ and $\cN(0,1)$,
we see that $A'$ also matches first $[i_*]$ moments with $\cN(0,1)$ and also assigns enough probability to a subgaussian component.
In fact, $A'$ matches first $[i_* +1]$-moments with $\cN(0,1)$ since both of them would be zero by symmetry.

We now calculate the Hermite coefficients of $A'$.
Recall that $h_j(\cdot)$ is an odd function for odd $j$ and an even function for an even $j$.
By symmetry of $A'$, $\E_{A'}[h_j(X)] = 0$ for odd j.
For any even $j\geq j_*$,
we have
\begin{align*}
\E[h_j(X)] 
&= (1-\eta) \E_{\cN(0,1-\alpha)}[h_j(X)] + (1-\eta) \int_{-1}^1 h_j(x)p'(x) dx + \eta \E_{\cN(0,1)}[h_j(X + \delta')]\\
&= (1-\eta) (\alpha)^{\frac{j}{2}} h_j(0) + (1-\eta)\int_{-1}^1 h_j(x)p'(x) dx + \eta \frac{\delta'^{j}}{\sqrt{j!}},
\end{align*}
which implies that 
\begin{align*}
\left|\E[h_j(X)]\right|^2 &\lesssim 1 + \eta^2 \left(C \frac{\alpha}{\eta}\right)^{j} 
\lesssim \eta^2 \left(C^j \frac{\alpha^j}{\eta^j}\right)^{\frac{j}{j_*}},
\end{align*}
where we use that for $j\geq j^*$, the second term dominates the first since $\alpha \gg \eta^{1 - \frac{2}{j_*}}$.

\end{proof}

\section*{Acknowledgements}
SBH thanks Larry Guth for enlightening conversations which led us to consider the consequences of our results for sparse PCA and $2\rightarrow p$ norm certification.
ST thanks Sidhanth Mohanty for helpful conversations about graph matrices and Jun-Ting (Tim) Hsieh for helpful discussions about subspace distortion.

\printbibliography

\appendix

\section*{Appendix}

\section{Sum-of-Squares Facts}
\label{sec:sos_facts}

The following fact is taken from~\cite[Lemma A.1]{boaz_dictionary}
\begin{fact}[SoS AM-GM]
    \label{fact:sos_am_gm}
    Let $x_1, \ldots, x_n$ be indeterminates for $n$ even.
    Then,
    \[
        \Set{x_1 \geq 0, \ldots, x_n \geq 0} \proves_{n}^{x_i} \prod_{i=1}^n x_i \leq \tfrac 1 n\sum_{i=1}^n x_i^n \,.
    \]
\end{fact}

\begin{fact}[SoS Cauchy-Schwarz]
    \label{fact:sos_cs}
    Let $x,y$ be vector-valued indeterminates (of size $n$).
    Then,
    \[
        \proves_4^{x,y} \iprod{x,y}^2 \leq \norm{x}^2 \norm{y}^2 \,.
    \]
\end{fact}
\begin{proof}
    First note that since $XY \leq \tfrac 1 2 X^2 + \tfrac 1 2 Y^2$ has a sum-of-squares proof, it holds that
    \[
        \proves_2^{x,y} \iprod{x,y} = \sum_{i=1}^n x_i y_i \leq \tfrac 1 2 \sum_{i=1}^n x_i^2 + \tfrac 1 2 \sum_{i=1}^n y_i^2 = \tfrac 1 2 \norm{x}^2 + \tfrac 1 2 \norm{y}^2 \,.
    \]
    Applying this to $x \otimes y$ and $y \otimes x$ (overloading notation), we obtain
    \[
        \proves_4^{x,y} \iprod{x,y}^2 = \iprod{x \otimes y,y \otimes x} \leq \norm{x \otimes y}^2 = \norm{x}^2 \norm{y}^2 \,.
    \]
\end{proof}

\begin{fact}
    \label{fact:spectral_upper_bound}
    Let $A \in \R^{n \times n}$ be a matrix and $x,y$ be vector-valued indeterminates (of size $n$).
    Then,
    \[\proves_{4} (x^\top A y)^2 \leq \norm{A}^2 \norm{x}^2 \norm{y}^2 \,.\]
\end{fact}
\begin{proof}
    Note that the matrix $A^TA$ is PSD-dominated by $\norm{A}^2 \cdot I_n$.
    Thus, by SoS Cauchy-Schwarz (cf.~\cref{fact:sos_cs}) it follows that
    \[
        \proves_4^{x,y} (x^\top A y)^2 \leq \norm{Ay}^2 \norm{x}^2  =  y^T(A^T A )y \norm{x}^2  \leq \norm{A}^2 \norm{x}^2 \norm{y}^2 \,.
    \]
\end{proof}

\begin{fact}
    \label{fact:sos_square_root}
    Let $X$ be an indeterminate and $C \geq 0$ be an absolute constant.
    Then, for every $p \geq 2$ 
    \[\Set{X \geq 0, X^p \leq C^p} \proves_p^X X \leq C \,.\]
    Further, for even $p$, we can drop the constraint $\Set{X \geq 0}$.
\end{fact}
\begin{proof}
    Let $ \gamma = C^{(p-1)/p}$.
    Then by~\cref{fact:sos_am_gm} it holds that (where the last factors are repeated $p-1$ times)
    \begin{align*}
        \Set{X \geq 0\,,X^p \leq C^p} \proves_2^X X &= \Paren{\tfrac 1 \gamma X} \cdot \gamma^{1/(p-1)} \cdot \ldots \cdot\gamma^{1/(p-1)} \leq \tfrac 1 p\Paren{\frac {X^p} {\gamma^p} + (p-1) \cdot \gamma^{p/(p-1)}} \\
        &\leq \tfrac 1 p \Paren{ \tfrac {C^p} {\gamma^p} + (p-1) \cdot \gamma^{p/(p-1) }} = C \,.
    \end{align*}
    To see the claim for even $p$, we start with $p = 2$.
    Note that instead of using~\cref{fact:sos_am_gm}, we can use that
    \[
        \emptyset \sststile{2}{X}\Paren{\tfrac 1 \gamma X} \cdot \gamma \leq \tfrac 1 2 \Paren{\frac {X^2} {\gamma^2} + \gamma^2} \,,
    \]
    and then continue as before.
    For general even $p$, we can first derive that $X^2 \leq C^2$ using that $X^2 \geq 0$ and then apply the proof of $p = 2$ above.
\end{proof}
\begin{fact}
    \label{fact:sos_triangle}
    Let $x_1, \ldots, x_n$ be indeterminates.
    Then,
    \[
        \proves_2^{x_1,\ldots,x_n} \Paren{\sum_{i=1}^n x_i}^2 \leq n \sum_{i=1}^n x_i^2 \,.
    \]
\end{fact}
\begin{proof}
    Using that $2x_ix_j \leq x_i^2 + x_j^2$ has a degree-2 SoS proof, we obtain
    \[
        \proves_2^{x_1,\ldots,x_n} \Paren{\sum_{i=1}^n x_i}^2 = \sum_{i=1}^n x_i^2 + \sum_{\substack{i,j \in [n],\\ i \neq j}} x_ix_j \leq n \sum_{i=1}^n x_i^2 \,.
    \]
\end{proof}

The next fact is taken from~\cite[Lemma A.2]{KotSS18}.
\begin{fact}
    \label{fact:sos_triangle_large_power}
    Let $X,Y$ be indeterminates and $p$ be even.
    Then, $\proves_p^{X,Y} (X+Y)^p \leq 2^p(X^p + Y^p)$.
\end{fact}

Recall that for a matrix $M$, $\norm{M}_t^t = \trace M^t$.
\begin{claim}
    \label{clm:SoS_op_to_schatten_p}
    Let $v$ be a vector-valued indeterminate (in dimension $d$) and $M$ be a matrix-valued indeterminate (in dimension $d \times d$).
    Then for $t \geq 2$ a power of $2$,
        \begin{align*}
        \Set{M^\top = M} \proves_{O(t)}^{v, M} \left(v^\top M v\right)^t \leq \norm{M}_t^t \norm{v}^{2t} \,.
        \end{align*}
    \end{claim}
    \begin{proof}
        We start with the case $t = 2$ and then use an induction.
        By SoS Cauchy-Schwarz (\cref{fact:sos_cs}) it follows that
        \[
            \proves_{O(1)}^{v,M} (v^\top M v)^2 \leq \norm{v}^2 \norm{M v}^2 \,.
        \]
        Denote by $M_i$ the rows of $M$.
        Then,
        \[
            \proves_{O(1)}^{v,M} \norm{M v}^2 = \sum_{i=1}^n \iprod{M_i, v}^2 \leq \sum_{i=1}^n \norm{M_i}^2 \norm{v}^2 = \Paren{\trace M^2} \norm{v}^2 \,.
        \]
        Thus, $(v^\top M v)^2 \leq (\trace M^2) \norm{v}^4$.

        Suppose we have shown the statement for some $t=2^\ell$.
        Then it follows that
        \begin{align*}
            \proves_{O(t)}^{v,M} \Paren{v^\top M v}^{2^{\ell+1}} = \Paren{\Paren{v^\top M v}^2 }^{2^\ell} \leq \Paren{ \norm{v}^2 \norm{Mv}^2 }^{2^\ell} = \norm{v}^{2^{\ell +1}} \Paren{v^\top \Paren{M^\top M} v}^{2^\ell} \,.
        \end{align*}
        By induction, this is at most (where we also use the constraint $M^\top = M$)
        \[
            \norm{v}^{2^{\ell +1}} \cdot \norm{v}^{2^{\ell +1}} \norm{M^\top M}_{2^\ell}^{2^\ell} = \norm{M}_{2^{\ell+1}}^{2^{\ell+1}} \norm{v}^{2\cdot 2^{\ell +1}}\,.
        \]
    \end{proof}

    \begin{fact}
        \label{fact:modified_pE_CS}
        Let $\pE$ be a degree-3 pseudo-distribution over variables $a_1, b_1, \ldots, a_n, b_n$ and $c$, that satisfies $\Set{c \geq 0}$.
        Then, it holds that
        \[
            \pE \sum_{i=1}^n a_i b_i c \leq \sqrt{\pE \sum_{i=1}^n a_i^2 c} \sqrt{\pE \sum_{i=1}^n b_i^2 c} \,.
        \]
    \end{fact}
    \begin{proof}
        For any $\gamma > 0$ it holds for any $i \in [n]$ that
        \[
            \Set{c \geq 0} \sststile{3}{a_i,b_i, c} a_i b_i c \leq \tfrac 1 {2 \gamma} a_i^2 c + \tfrac {\gamma} 2 b_i^2 c \,.
        \]
        And thus also that $\Set{c \geq 0} \sststile{3}{a_i,b_i,c} \sum_{i=1}^n a_i b_i c \leq \tfrac 1 {2 \gamma} \sum_{i=1}^n a_i^2 c + \tfrac {\gamma} {2} \sum_{i=1}^n b_i^2 c$.
        Choosing $\gamma^2 = \tfrac {\pE \sum_{i=1}^n a_i^2 c}{\pE \sum_{i=1}^n b_i^2 c}$, the claim follows after applying $\pE$.
    \end{proof}

\end{document}